\documentclass[12pt]{article}
\usepackage[UKenglish]{babel} 
\usepackage[T1]{fontenc} 
\usepackage[utf8]{inputenc} 
\usepackage{amsmath}
\usepackage{amsthm}
\usepackage{amssymb}
\usepackage{cite} 
\usepackage{bm} 
\usepackage{indentfirst} 
\usepackage[inline]{enumitem} 
\usepackage{comment} 
\usepackage{todonotes} 
\usepackage{csquotes}
\usepackage{hyperref} 
\newtheorem{theorem}{Theorem} 
\newtheorem{sumtheorem}{Theorem}

\newtheorem{lemma}{Lemma} 
\newtheorem{corollary}{Corollary}
\newtheorem*{assumption*}{Assumption}
\newtheorem{definition}{Definition}
\newtheorem*{claim*}{Claim}
\makeatletter
 \DeclareFontFamily{OMX}{MnSymbolE}{}
 \DeclareSymbolFont{MnLargeSymbols}{OMX}{MnSymbolE}{m}{n}
 \SetSymbolFont{MnLargeSymbols}{bold}{OMX}{MnSymbolE}{b}{n}
 \DeclareFontShape{OMX}{MnSymbolE}{m}{n}{
	 <-6>  MnSymbolE5
	<6-7>  MnSymbolE6
	<7-8>  MnSymbolE7
	<8-9>  MnSymbolE8
	<9-10> MnSymbolE9
   <10-12> MnSymbolE10
   <12->   MnSymbolE12
 }{}
 \DeclareFontShape{OMX}{MnSymbolE}{b}{n}{
	 <-6>  MnSymbolE-Bold5
	<6-7>  MnSymbolE-Bold6
	<7-8>  MnSymbolE-Bold7
	<8-9>  MnSymbolE-Bold8
	<9-10> MnSymbolE-Bold9
   <10-12> MnSymbolE-Bold10
   <12->   MnSymbolE-Bold12
 }{}
 
 \let\llangle\@undefined
 \let\rrangle\@undefined
 \DeclareMathDelimiter{\llangle}{\mathopen}%
					  {MnLargeSymbols}{'164}{MnLargeSymbols}{'164}
 \DeclareMathDelimiter{\rrangle}{\mathclose}%
					  {MnLargeSymbols}{'171}{MnLargeSymbols}{'171}
\makeatother

\newcommand{\const}{\,{\rm const}\,}

\newcommand{\td}{\text{d}}
\newcommand{\lie}{\mathcal{L}}
\newcommand{\ord}{\mathcal{O}}
\newcommand{\spn}{{\rm span}}

\title{\textbf{A classification of supersymmetric Kaluza-Klein black holes with a single axial symmetry}}
\author{David Katona\footnote{d.katona@sms.ed.ac.uk}
\\ \\ 
\small \sl School of Mathematics and Maxwell Institute for Mathematical Sciences, 
\\ 
\small \sl University of Edinburgh, King's Buildings, Edinburgh, EH9 3JZ, UK }

\date{}

\begin{document}
	\maketitle
	\begin{abstract}
		We extend the recent classification of five-dimensional, supersymmetric asymptotically flat black holes with only a single axial symmetry 
		to black holes with Kaluza-Klein asymptotics. This includes a similar class of solutions for which the supersymmetric Killing field 
		is generically timelike, and the corresponding base (orbit space of the supersymmetric Killing field) is of multi-centred 
		Gibbons-Hawking type. These solutions are determined by four harmonic functions on $\mathbb{R}^3$ with simple poles at the centres 
		corresponding to connected components of the horizon, and fixed points of the axial symmetry. The allowed horizon topologies are $S^3$, 
		$S^2\times S^1$, and lens space $L(p, 1)$, and the domain of outer communication may have non-trivial topology with non-contractible 2-cycles. 
		The classification also reveals a novel class of supersymmetric \mbox{(multi-)black} rings for which the supersymmetric Killing field is globally 
		null. These solutions are determined by two harmonic functions on  $\mathbb{R}^3$ with simple poles at centres corresponding to 
		horizon components. We determine the subclass of Kaluza-Klein black holes that 
		can be dimensionally reduced to obtain smooth, supersymmetric, four-dimensional multi-black holes. This gives a classification of 
		four-dimensional asymptotically flat supersymmetric multi-black holes first described by Denef {\it et al}.
	\end{abstract}	
\newpage
\tableofcontents

\newpage

\section{Introduction} \label{sec_intro}

Black holes are in the focus of gravitational research. In four-dimensional vacuum gravity, or Einstein-Maxwell theory, 
asymptotically flat black holes have a surprisingly simple moduli space due to the well-known uniqueness theorems 
(see e.g. \cite{heusler_black_1996,chrusciel_stationary_2012}). In contrast, higher dimensional general relativity has a much richer 
structure, and black hole uniqueness does not hold even in the asymptotically flat vacuum case (for review see e.g. \cite{hollands_black_2012}), which 
became clear with the discovery of rotating vacuum black holes with $S^2\times S^1$ horizon topology, known as black rings \cite{emparan_rotating_2002}. 
Rather surprisingly, for a range of asymptotic charges, black rings coexist with the spherical Myers-Perry black holes \cite{myers_black_1986}, 
providing an explicit example of non-uniqueness.

Much is known in general about higher dimensional stationary black holes. The topology of the horizon is restricted to be of positive Yamabe 
type, i.e. they admit metrics with positive scalar curvature~\cite{galloway_generalization_2006}, which becomes less restrictive 
as we go to higher dimensions. Further restrictions have been derived in the literature for black holes with an axial 
symmetry~\cite{hollands_further_2011}. This assumption was motivated by the rigidity theorem, which states that analytic solutions with 
rotating black holes must also admit an axial $U(1)$ symmetry~\cite{hollands_higher_2007,moncrief_symmetries_2008,hollands_stationary_2009}.  
The topology of the domain of outer communication (DOC) is also restricted by topological censorship. In the asymptotically flat case it is required 
to be simply connected~\cite{friedman_topological_1995}. This has been generalised to asymptotically Kaluza-Klein spacetimes~\cite{chrusciel_topological_2009},
where the quotient space of the DOC by the symmetry group corresponding to translations in the compact dimensions must be simply connected. 
In cases when the spacetime admits biaxial $U(1)^2$ symmetry, a uniqueness theorem for black holes has also been established~\cite{hollands_uniqueness_2008,hollands_uniqueness_2008-1,hollands_uniqueness_2011}, 
however it has been conjectured (for the vacuum case) that solutions with fewer symmetries must exist~\cite{reall_higher_2004}. Evidence for the 
existence of such solutions has long been gathering in the literature~\cite{emparan_new_2010,dias_instability_2010}, but 
the first explicit examples have been constructed only recently in five-dimensional supergravity~\cite{katona_supersymmetric_2023}.

Black hole non-uniqueness is also present in higher dimensional supergravity theories, even among supersymmetric solutions. For five-dimensional 
minimal supergravity the first known black hole was the BMPV solution \cite{breckenridge_d--branes_1997} with a spherical horizon. Later 
a black ring solution \cite{elvang_supersymmetric_2004}, and concentric black ring/black hole solutions \cite{gauntlett_concentric_2005} were also found. 
Asymptotic charges of the latter one can overlap with those of the BMPV solution, and even more surprisingly, they can have greater total horizon area 
than the corresponding BMPV solution. This has since been shown for even single-black hole solutions of this theory \cite{horowitz_comments_2017,breunholder_supersymmetric_2019}. 
This finding is particularly puzzling given that the microscopic derivation of Bekenstein-Hawking entropy in string theory for given charges matches 
the entropy of the BMPV solution \cite{breckenridge_d--branes_1997,strominger_microscopic_1996}. This `entropy enigma' is yet to be resolved, but 
microscopic counting of entropy provides further motivation to determine the moduli space of supersymmetric black holes.

Although the full moduli space is yet to be explored, this theory is known to admit various black hole solutions. On top of the aforementioned 
spherical and ring solutions, black holes with lens space horizons $L(p, 1)$ have been constructed with flat asymptotics \cite{kunduri_supersymmetric_2014,tomizawa_supersymmetric_2016,tomizawa_asymptotically_2017,breunholder_moduli_2019,breunholder_supersymmetric_2019}, 
but later also with Kaluza-Klein asymptotics \cite{tomizawa_kaluza-klein_2018}. An important feature that contributes to the richness of solutions is that this 
theory admits solitonic solutions, termed `bubbling spacetimes', which are smooth, horizonless solutions admitting non-trivial topology in the form 
of non-trivial 2-cycles, supported by magnetic fluxes of the Maxwell field \cite{bena_bubbling_2006,bena_black_2008}. It is also possible to construct black 
hole solutions with `bubbling' domain of outer communication, which include examples that contribute to the single black hole entropy enigma 
mentioned above\cite{kunduri_black_2014,horowitz_comments_2017,breunholder_supersymmetric_2019}.

The general {\it local} form of a supersymmetric solution of minimal five-dimensional supergravity is known \cite{gauntlett_all_2003,bena_one_2005}. 
Using Killing spinor bilinears, one can show that such a solution admits a globally defined causal Killing field, a scalar function, and three 
2-forms. When the supersymmetric Killing field is timelike, the solution takes the form of a timelike fibration over a hyper-K\"ahler base manifold 
for which the three supersymmetric 2-forms are the complex structures. Furthermore, if the solution admits a triholomorphic $U(1)$ isometry (i.e. 
an isometry that preserves the complex structures of the base) the base manifold takes the form of a multi-centred Gibbons-Hawking space 
\cite{hawking_gravitational_1977,gibbons_gravitational_1978,gibbons_hidden_1988}, and the solution is locally fully determined by four harmonic functions on 
$\mathbb{R}^3$ \cite{gauntlett_all_2003}. This is the case if the spacetime has biaxial symmetry ($U(1)^2$) \cite{breunholder_moduli_2019}, or a single 
axial symmetry which preserves the Killing spinor \cite{katona_supersymmetric_2023}.

Using the aforementioned local results, a number of classification theorems have been proven for supersymmetric black holes. The near-horizon 
geometries have been determined in \cite{reall_higher_2004} (assuming that the supersymmetric Killing field becomes timelike outside the horizon), 
and the possible geometries are locally isometric to\footnote{In \cite{reall_higher_2004} $T^3$ has also been derived as a possible geometry, 
however it is excluded by \cite{galloway_generalization_2006}.} $S^3$ or $S^2\times S^1$. In the asymptotically flat case global results are 
known as well. First it was shown that a locally spherical black hole with a supersymmetric Killing field that is timelike on the domain of 
outer communication must be isometric to the BMPV solution \cite{reall_higher_2004}. This assumption, however, is restrictive and excludes the 
majority of the moduli space of black holes (black lenses and black holes in bubbling spacetimes). A common feature of all 
near-horizon geometries, and in fact all the solutions mentioned so far, is biaxial ($U(1)^2$) isometry. In \cite{breunholder_moduli_2019} 
a classification of asymptotically flat black hole solutions with such symmetry has been achieved. These solutions have Gibbons-Hawking 
base with the associated harmonic functions having simple poles at collinear centres on $\mathbb{R}^3$, corresponding to horizon components 
or fixed points of the triholomorphic $U(1)$ symmetry.

Recently, these results have been generalised to asymptotically flat solutions admitting only a single axial symmetry that preserves the 
Killing spinor \cite{katona_supersymmetric_2023}. Similarly to the biaxial case, these solutions are of multi-centred type, with harmonic 
functions having simple poles at generic (not necessarily collinear) points. This provided the first explicit construction of a higher 
dimensional asymptotically flat black hole with just a single axial symmetry, confirming the conjecture of \cite{reall_higher_2004} for 
supersymmetric black holes.

In this paper we will focus on asymptotically Kaluza-Klein solutions, which are asymptotically diffeomorphic to a circle fibration over 
flat Minkowski space. This includes trivial circle products, and also non-trivial fibrations like the Kaluza-Klein monopole. 
Many explicit Kaluza-Klein black hole solutions of five-dimensional supergravity are known in the literature \cite{elvang_supersymmetric_2005,gaiotto_5d_2006,bena_black_2005,ishihara_kaluza-klein_2006-1,ishihara_kaluza-klein_2006,nakagawa_charged_2008,matsuno_rotating_2008,tomizawa_charged_2008,tomizawa_squashed_2009,tomizawa_compactified_2010,tomizawa_exact_2011,tomizawa_general_2013,tomizawa_kaluza-klein_2018}, 
furthermore, a uniqueness theorem for non-supersymmetric, biaxisymmetric, spherical black holes of this theory has also been proven \cite{tomizawa_uniqueness_2010}. 
A classification of more general or supersymmetric solutions is not yet available, however. The main motivation of studying such solutions is 
that they can be dimensionally reduced to four dimensions, which might have more physical relevance for our universe. For five-dimensional minimal 
supergravity, the bosonic sector of the reduced theory contains gravity, two Maxwell fields, and two scalars, a dilaton and an axion. Previously, 
an interesting connection has been unveiled between five-dimensional Kaluza-Klein and four-dimensional asymptotically flat black holes 
of supergravity theories \cite{elvang_supersymmetric_2005,gaiotto_5d_2006,gaiotto_new_2006,behrndt_exploring_2006}. Five-dimensional 
supersymmetric (multi-)black holes/rings with Taub-NUT base space correspond to (multi-)black hole solutions of the four-dimensional theory 
first derived by Denef {\it et al.} \cite{denef_supergravity_2000,denef_split_2001,bates_exact_2011}.

The purpose of this paper is to generalise the results of \cite{katona_supersymmetric_2023} to Kaluza-Klein asymptotics, and classify all 
supersymmetric black hole or soliton solutions with an axial symmetry that `commutes' with the supersymmetry, i.e. it preserves the Killing 
spinor. The latter assumption is the supersymmetric generalisation of the usual requirement that the axial symmetry commutes with the 
stationary symmetry. Furthermore, when considering reductions of the solutions to $4D$, as a result of this assumption we obtain 
supersymmetric solutions in the lower dimensional theory.

The main difference to the asymptotically flat case is that the supersymmetric Killing field is no longer naturally identified with 
the stationary Killing field. Instead, we assume that the stationary Killing field is a constant linear combination of the axial and supersymmetric 
Killing fields. This leads to -- from the five-dimensional perspective -- two qualitatively different classes of solutions depending on whether the supersymmetric Killing field is timelike 
on a dense submanifold or globally null. The former one is similar to the one found in the asymptotically flat case~\cite{katona_supersymmetric_2023},
and its classification is provided by the following theorem (for the full statement see Theorem \ref{thm_classification}).

\begin{sumtheorem}\label{sumthm_TL}
	A supersymmetric, asymptotically Kaluza-Klein (in the sense of Definition \ref{def_KK}) black hole or soliton solution of $D=5$ minimal supergravity with an axial symmetry that 
	preserves the Killing spinor and a supersymmetric Killing field that is not globally null must have a Gibbons-Hawking base (on a dense submanifold), 
	and is globally determined by four associated harmonic functions on $\mathbb{R}^3$ which are of `multi-centred' form, with parameters satisfying a set of 
	algebraic constraints. The centres either correspond to fixed points of the axial Killing field, or connected components of the horizon, each of which has 
	topology $S^3$, $L(p,1)$ or $S^2\times S^1$.
\end{sumtheorem}

Despite the similarity to the asymptotically flat classification, there are technical differences in its proof. The main complication 
comes from the fact that the DOC is not simply connected in general, hence certain closed 1-forms do not necessarily define the global functions 
that play a key role in the asymptotically flat proof. In order to overcome this, first we show that the axial Killing field must 
be tangent to the Kaluza-Klein direction, then after excluding the possibility of orbits with a discrete isotropy group 
(exceptional orbits), we apply topological censorship to deduce 
simple connectedness of the orbit space of the axial Killing field~\cite{chrusciel_topological_2009}. This allows us to define certain 
$U(1)$-invariant functions globally. From this point, the proof is almost identical to that of the asymptotically flat case.

The lack of exceptional orbits follows purely from the axial Killing field preserving three linearly independent two-forms (Killing spinor 
bilinears, which are the complex structures of the hyper-K\"ahler base in the timelike case). This result, together with the 
orientability of the spacetime, restricts the possible horizon topologies to $S^3$, $L(p,1)$ or $S^2\times S^1$. These are the only 
allowed\footnote{Only orbifolds with positive Euler characteristic are allowed~\cite{hollands_further_2011}.} Seifert three-manifolds (with 
orientable fibres and base) that do not contain any orbifold points. Such orbifold points would require the presence of exceptional orbits 
in the DOC~\cite{hollands_further_2011}, which are ruled out as mentioned above.

Even in the case when the supersymmetric Killing field is timelike on a dense submanifold, the limit 
of its norm at infinity can be zero. We call this case {\it asymptotically null}, in contrast to the {\it asymptotically timelike} case, 
when we normalise the supersymmetric Killing field to have unit norm at infinity. These should not be confused with  
the globally null case, which we detail next.

The other main class of solutions contains spacetimes on which the supersymmetric Killing field is globally null, for which 
we have the following result (the full statement can be found in Theorem \ref{thm_classification_null}).
\begin{sumtheorem}\label{sumthm_null}
	An asymptotically Kaluza-Klein (in the sense of Definition \ref{def_KK}), supersymmetric black hole or soliton solution of 
	$D=5$ minimal supergravity with an axial Killing field $W$ preserving the Killing spinor and for which 
	the supersymmetric Killing field $V$ is globally null has a metric of the form 
	\begin{align}
		g &= -\frac{1}{\mathcal{G}}(\mathcal{Q}\td u^2 + 2\td u\td v) + \mathcal{G}^2\td x^i\td x^i\;, \label{eq_sum_null}
	\end{align}
	where $W=\partial_u$, $V=\partial_v$, and  $\mathcal{G},\mathcal{Q}$ are harmonic functions on $\mathbb{R}^3$ with simple poles at 
	centres corresponding to connected components of the horizon, each with topology $S^2\times S^1$.
\end{sumtheorem}
To our knowledge these `null' Kaluza-Klein black ring solutions have not been previously described. Similar `null' 
solutions have been previously found in~\cite{gibbons_higher-dimensional_1995}, which describe static black strings, 
those, however, cannot be compactified to obtain smooth, asymptotically Kaluza-Klein black holes, hence they are not 
part of our classification. As the supersymmetric Killing field does not become timelike outside the horizon in this case, 
the proof of the near-horizon classification in~\cite{reall_higher_2004} is no longer valid. Therefore, we extend this 
near-horizon analysis (with no symmetry assumptions) to the null case, and we find that the near-horizon 
geometry in the null case agrees with the null limit of the timelike case. This is in agreement with the results 
of~\cite{kayani_symmetry_2018} in which it has been shown that near-horizon geometries of this theory are necessarily 
maximally supersymmetric.

It turns out that, if one relaxes a condition on the harmonic functions of Theorem \ref{sumthm_TL} so that they allow the supersymmetric Killing 
field to be globally null (which is {\it a priori} not clear that one is allowed to do), then one obtains precisely the solutions of Theorem \ref{sumthm_null}. 
Presumably, this is because they have common higher dimensional origin, as both timelike and null solutions can be uplifted -- at least locally -- to obtain supersymmetric solutions of six dimensional 
minimal supergravity~\cite{gutowski_all_2003}. A general feature of the six-dimensional solutions is that the supersymmetric Killing field is 
null everywhere, similarly to the solutions in Theorem \ref{sumthm_null}. Interestingly, the cartesian coordinates in both Theorem \ref{sumthm_TL}  
and \ref{sumthm_null} (those of the Gibbons-Hawking base, and $x^i$ in (\ref{eq_sum_null}), respectively) originate from a Gibbons-Hawking 
base that is used to construct the six-dimensional solutions. It would be interesting to investigate the classification from a 
six-dimensional perspective.

In all cases we find that the axial Killing field is tangent to the Kaluza-Klein direction, thus it is natural to consider the 
reduction of these solutions to four dimensions, which are also supersymmetric. We prove that the obtained four-dimensional solution 
is smooth on and outside the horizon provided that there are no fixed points of the $U(1)$ Killing field. This is automatically true in the null class, but restricts 
the timelike class. Conversely, given a four-dimensional supersymmetric, asymptotically flat black hole solution of a certain 
four-dimensional supergravity with two scalar fields and Maxwell fields, it uplifts to a supersymmetric, asymptotically Kaluza-Klein solution of minimal supergravity in $D=5$.
A condition for a smooth uplift is that one of the Maxwell fields corresponds to a principal $U(1)$-bundle, and hence its magnetic 
charges are quantised. Thus, with this assumption, we obtain a classification theorem of four-dimensional supersymmetric black holes 
of the theory (for the detailed statement see Theorem \ref{thm_class_4D}).

\begin{sumtheorem}\label{sumthm_4D}
	Consider a supersymmetric, asymptotically flat black hole solution of $D=4$ $\mathcal{N}=2$ supergravity coupled to a vector multiplet 
	in which one of the Maxwell fields is the curvature of a connection on a principal $U(1)$-bundle over the spacetime. Furthermore, assume 
	that the supersymmetric Killing field is timelike on the DOC. Then the solution must belong to the class of multi-black holes derived in 
	\cite{denef_supergravity_2000,denef_split_2001,bates_exact_2011}.
\end{sumtheorem}

The structure of this paper and the proof of theorems are as follows. In Section \ref{sec_recap} we briefly summarise the general 
form of such solutions based on \cite{gauntlett_all_2003}, and present the local solutions in the timelike and null case separately.
In Section \ref{sec_asymptotics} we derive that the $U(1)$ Killing field must be tangent to the Kaluza-Klein direction, by using that it preserves 
the Killing spinor bilinears, and showing that any Killing field must approach that of a flat space with a compact direction. In Section \ref{sec_NHnull} 
we perform a near-horizon analysis of the null case (following and completing that of~\cite{reall_higher_2004}).
In Section \ref{sec_orbit} we analyse the structure of the orbit space. In particular, we show that any exceptional orbits are excluded,  
then using topological censorship~\cite{chrusciel_topological_2009} we show that the cartesian coordinates defined by the axial Killing field and 
the two-form bilinears define a global chart on the orbit space. This also restricts the general form of the solution to those with associated 
harmonic functions having simple poles in both the null and timelike case. 
In Section \ref{sec_sufficient} we derive the sufficient conditions for smoothness of such multi-centred solutions. In Section 
\ref{ssec_classifictaion} we state and prove the main classification theorems (Theorem \ref{thm_classification}-\ref{thm_classification_null}). 
In Section \ref{sec_reduction} we perform a Kaluza-Klein reduction to four dimensions, and show that the four-dimensional 
solution is smooth on and outside the horizon provided there are no fixed points of the axial symmetry in $5D$. 
Finally, we use this correspondence to prove a classification theorem for four-dimensional asymptotically flat black hole solutions to minimal 
supergravity coupled to a vector multiplet (Theorem \ref{thm_class_4D}).


\section{Supersymmetric solutions in five dimensions with axisymmetry} \label{sec_recap}

The theory in consideration is the bosonic sector of five-dimensional minimal supergravity, given by the action 
\begin{equation}
	S = \frac{1}{16\pi G}\int R\star 1 -2 F\wedge\star F -\frac{8}{3\sqrt{3}}F\wedge F\wedge A\;, \label{eq_action}
\end{equation}
where $F=\td A$. We work in the conventions of~\cite{reall_higher_2004}, so the signature of the metric is `mostly plus'. 
The most general \textit{local} form of a supersymmetric solution has been derived in~\cite{gauntlett_all_2003} using Killing spinor 
bilinears which globally define a function $f$, a vector field $V$, and three 2-forms $X^{(i)}$, $i=1,2,3$. These satisfy the following 
algebraic equations.
\begin{align}
	g(V,V)&=-f^2\;, \label{eq_Vnorm}\\
	\iota_VX^{(i)}&=0\;, \label{eq_VX}\\
	\iota_V\star X^{(i)}&=-f X^{(i)}\;,\label{eq_VXdual} \\
	X^{(i)}_{\mu\lambda}X^{(j)}_{\nu}{}^\lambda &= \delta_{ij}(f^2g_{\mu\nu}+V_\mu V_\nu)-f\epsilon_{ijk}X^{(k)}_{\mu\nu}\;, \label{eq_XXspacetime}
\end{align}
where $\epsilon_{ijk}$ is the Levi-Civita symbol with $\epsilon_{123}=1$. Using the Killing spinor equation, one further obtains 
\begin{align}
	&\td f = -\frac{2}{\sqrt{3}}\iota_VF\;,\label{eq_VF} \\
	&\nabla_{(\mu}V_{\nu)}=0\;,\\
	&\td V = -\frac{4}{\sqrt{3}}fF-\frac{2}{\sqrt{3}}\star(F\wedge V)\;,\label{eq_starVF}\\
	&\nabla_{\mu}X^{(i)}_{\nu\rho} = \frac{1}{\sqrt{3}}\left(2F_\mu{}^\sigma(\star X^{(i)})_{\sigma\nu\rho}-2F_{[\nu}{}^\sigma(\star X^{(i)})_{\rho]\mu\sigma}+g_{\mu[\nu}F^{\sigma\kappa}(\star X^{(i)})_{\rho]\sigma\kappa}\right)\;.\label{eq_DX}
\end{align}
In particular, $V$ is Killing, $X^{(i)}$ are closed, and $V$ preserves $X^{(i)}$ and the Maxwell field. 
Further specification of the local solution depends on whether $V$ is timelike or null in some open region, which we will describe in 
detail in the following sections.

In this paper we are considering asymptotically Kaluza-Klein spacetimes, for which we use the following 
definition. 
\begin{definition} \label{def_KK}
	$(\mathcal{M}, g)$ is stationary and asymptotically Kaluza-Klein, if 
	\begin{enumerate}[label=(\roman*)]	
		\item the domain of outer communication (DOC), denoted by $\llangle \mathcal{M} \rrangle$, has an end diffeomorphic to 
		$\mathbb{R}\times\Sigma_0$, with $\Sigma_0$ being a circle fibration over $\mathbb{R}^3\setminus B^3$, where $B^3$ denotes a 3-ball, 
		and the $\mathbb{R}$ factor corresponds to orbits of a timelike Killing field (stationary Killing field),\label{def_KK_top}
		\item the metric on this end can be written as $g_{\mu\nu}=\tilde g_{\mu\nu}+ \mathcal{O}(\tilde r^{-\tau})$ for some decay rate $\tau>0$ and
		\begin{equation}
			\tilde g=-\td u^0\td u^0 + \delta_{ij}\td u^i\td u^j + \tilde L^2\td \tilde\psi^2\;,\label{eq_metric_AKK}
		\end{equation}
		where $\tilde L$ is a positive constant, $u^0$ and $(u^i)_{i=1}^3$ are the pull-back of the cartesian coordinates on $\mathbb{R}\times \mathbb{R}^3$, 
		i.e. $\partial_0$ is the stationary Killing field, and 
		$\tilde r :=\sqrt{u^iu^j\delta_{ij}}$, $\tilde\psi$ is a $4\pi$-periodic coordinate on the fibres, and in these coordinates 
		$\partial_i^l\partial_{\tilde\psi}^k g_{\nu\rho} =\mathcal{O}(\tilde r^{-\tau-l})$ for  $1\le k+l\le 3$,\label{def_KK_FO}
		\item the components of the Ricci tensor in these coordinates fall off as $R_{\mu\nu}=\ord(\tilde r^{-\tau-2})$ and its $\mathbb{R}^3$ 
		derivatives as $\partial_iR_{\mu\nu}=\ord(\tilde r^{-\tau-2})$.\label{def_KK_Ricci}
	\end{enumerate}  \label{assumption_falloff}
\end{definition}

\noindent {\bf Remarks.}
\begin{enumerate}
	\item Definition \ref{def_KK} only requires that the spacetime is a circle fibration at infinity. This includes trivial 
	fibrations (standard Kaluza-Klein asymptotics), as well as non-trivial fibrations such as the Kaluza-Klein monopole. In terms of the `sphere' at infinity, Definition \ref{def_KK} allows not only $S^2\times S^1$, but also spherical geometries such as (squashed) $S^3$ or lens spaces, when the fibration is non-trivial. The nature of the fibration is determined by subleading terms (in coordinates of Definition \ref{def_KK}) in the metric 
	(see later in Section \ref{ssec_asymptotics}).
	\item Definition \ref{def_KK} requires that the circle direction has bounded length at infinity. It is motivated by the physical picture that one obtains a four-dimensional effective description when the Kaluza-Klein direction is small enough. Alternatively, one may consider different definitions, such as cubic volume growth of a spatial slice at infinity. This alternative definition also allows for the possibility of a spatial geometry such as the euclidean Kerr instanton, where the circle direction has unbounded length. Such spacetimes are not included in the present work. Nevertheless, it is an interesting question whether such supersymmetric solutions exist.
	\item The fall-off of the components of the Ricci tensor is in fact equivalent (through the Einstein equations) to 
	requiring that the Maxwell field falls off as $F\sim \ord(\tilde{r}^{-\tau/2-1})$ at spatial infinity. 
	This can be seen by looking at $T_{00}$, which is a positive definite quadratic in the components of the Maxwell field. The fall-off for 
	the third derivative of the metric and the derivative of the Ricci tensor is a technical assumption, to ensure that components of 
	the Riemann tensor fall off as $R^\mu{}_{\nu\lambda\kappa}=\ord(\tilde r^{-\tau-2})$ as shown in Appendix \ref{app_asymptotics}. 
	Alternatively, one can assume the fall-off of the Riemann tensor directly. 
\end{enumerate}

We now list our assumptions. These are, except for asymptotics, equivalent to Assumptions 1 and 2 of~\cite{katona_supersymmetric_2023}. 
We assume that $(\mathcal{M}, g, F)$ is a solution of (\ref{eq_action}) such that
\begin{enumerate}[label=(\roman*)]
	\item it admits a globally defined Killing spinor $\epsilon$ (i.e. supersymmetric), \label{ass5D_susy}
	\item the DOC is globally hyperbolic, that is, it admits a Cauchy surface $\Sigma$, \label{ass5D_GH}
	\item the spacetime is stationary and asymptotically Kaluza-Klein as in Definition \ref{def_KK},\label{ass5D_AKK}
	\item the supersymmetric Killing field $V$ is complete, \label{ass5D_Vcomplete}
	\item the horizon $\mathcal{H}$ admits a smooth, compact cross-section (which may not be connected), \label{ass5D_horizon}
	\item the closure of the Cauchy surface $\Sigma$ in $\mathcal{M}$ is a union of a compact set and an asymptotically Kaluza-Klein end,\label{ass5D_compact}
	\item $(\mathcal{M},g)$ admits a smooth Killing field $W$ with periodic orbits that preserves the Killing spinor $\epsilon$ and the Maxwell-field $F$,\label{ass5D_W}
	\item the stationary Killing field is an $\mathbb{R}$-linear combination of $V$ and $W$,\label{ass5D_stationary}
	\item at each point of the DOC there exists a timelike linear combination of $V$ and $W$, \label{ass5D_timelike}
	\item the metric and the Maxwell field are smooth ($C^\infty$) on and outside the horizon. \label{ass5D_smooth}
\end{enumerate}

From global hyperbolicity of the DOC and completeness of $V$ follows that the DOC has topology 
$\llangle\mathcal{M}\rrangle\simeq\mathbb{R}\times \Sigma$, where $\Sigma$ 
is a smooth manifold, which we can identify with the orbit space of $V$\footnote{
	See remark 2 after Assumption 1 in~\cite{katona_supersymmetric_2023}.
}. 

We emphasise that we do {\it not} assume that the stationary Killing field coincides with the supersymmetric Killing field, in contrast to 
the asymptotically flat solutions considered in \cite{katona_supersymmetric_2023}. Neither do we assume that the axial Killing field is 
tangent to the Kaluza-Klein direction at infinity. In Section \ref{sec_asymptotics} we show that the latter is indeed true, so the $U(1)$ 
Killing field (and the length of its orbits) is bounded in the asymptotic region. This means that its linear combination with the 
causal vector field $V$ can be timelike, so assumption \ref{ass5D_stationary} is natural to make. The solution in the coordinates adapted 
to these Killing fields generally will not be in the rest frame. Indeed, this is the case for the asymptotically Kaluza-Klein black 
holes considered in~\cite{elvang_supersymmetric_2005,tomizawa_kaluza-klein_2018}.

Another important difference to the asymptotically flat case is that in general the DOC is not simply connected. 
Therefore, we cannot {\it a priori} assume that closed 1-forms globally define scalar functions. However, in Section \ref{sec_orbit} we
will show that this is indeed the case for the $U(1)$-invariant one-forms considered during the proof.

An important consequence of assumption \ref{ass5D_W} is that all the Killing spinor bilinears are preserved by $W$, that is
\begin{align}
	W(f)=0\;, \qquad [V,W]=0\;, \qquad \lie_W X^{(i)}=0\;.\label{eq_Wbilinears}
\end{align}
The assumption that $W$ preserves $F$ is redundant. We will see that either the timelike or null region is dense in the DOC, and in both cases 
$F$ can be expressed by $U(1)$-invariant quantities (\ref{eq_Ftimelike}) and (\ref{eq_Maxwell_null}). Thus, by continuity of $\lie_WF$, the 
Maxwell field is $U(1)$-invariant everywhere.

\subsection{Timelike case with axial symmetry}

Let us define $\widetilde{\mathcal{M}}\subset\llangle \mathcal{M}\rrangle$ as the region where $f\neq 0$, that is, where $V$ is timelike. 
Around each point of $\widetilde{\mathcal{M}}$ one can define a local chart in which the metric is given by
\begin{equation}
	g|_{\widetilde{\mathcal{M}}} = -f^2(\td t + \omega)^2 + \frac{1}{f}h\;, \label{eq_hdef}
\end{equation}
where $V=\partial_t$, $\omega$ and $h$ is a 1-form and metric on the four-dimensional riemannian base manifold 
$\mathcal{B}:=\widetilde{\mathcal{M}}/\mathbb{R}_V\subset \Sigma$. Note that $h$ is invariantly defined on $\widetilde{\mathcal{M}}$, 
and $\omega$ is defined by $\iota_V \omega=0$ and $\td \omega=-\td ( f^{-2} V)$ up to a gradient. Since 
$\lie_V X^{(i)}=\iota_V X^{(i)}=0$, $X^{(i)}$ can be regarded as 2-forms on $\mathcal{B}$. Using properties of the
Killing spinor, one can show that $(\mathcal{B}, h, X^{(i)})$ is hyper-K\"ahler, i.e.\footnote{
	Unless otherwise stated, Greek letters $\mu,\nu...$ denote spacetime indices, while Latin letters $a,b,c...$ denote base space indices, 
	and $i,j,k...$ denote $\mathbb{R}^3$ indices. Indices are raised and lowered with $g$, $h$ and $\delta$ on the spacetime, base space,
	and $\mathbb{R}^3$ respectively.
}
\begin{align}
	\nabla^{(h)}_aX^{(i)}_{bc}&=0\;,\label{eq_delX}\\
	X^{(i)}_{ac} X^{(j) c}_{\phantom{(j) c} b} &= - \delta_{ij} h_{ab} + \epsilon_{ijk} X^{(k)}_{ab}\;, \label{eq_quaternion} \\
	\star_hX^{(i)} &= -X^{(i)}\;, \label{eq_XASD}
\end{align}
where $\nabla^{(h)}$ is the Levi-Civita connection of $h$, and $\star_h$ denotes the Hodge star on the base with orientation $\eta$ 
defined by the spacetime orientation $f(\td t+\omega)\wedge \eta$. On $\widetilde{\mathcal{M}}$ the Maxwell-field is given by 
\begin{equation}
	F = -\frac{\sqrt{3}}{2}\td \left(\frac{V}{f}\right) -\frac{1}{\sqrt{3}} G^+\;, \label{eq_Ftimelike}
\end{equation}
where $G^+= \frac{1}{2}(1+\star_h)f\td \omega$.

From our assumptions, using Lemma 1 of~\cite{katona_supersymmetric_2023} it follows 
that the two Killing fields commute, and $W$ defines a triholomorphic $U(1)$ action on $\mathcal{B}$ (i.e. $\mathcal{L}_W X^{(i)}=0$). 
We will use the gauge $\mathcal{L}_W t=0$, thus $W$ also preserves $\omega$. It is well-known that a hyper-K\"ahler four-manifold with 
triholomorphic $U(1)$ action can locally be written in Gibbons-Hawking form~\cite{gibbons_hidden_1988}
\begin{equation}
	h = \frac{1}{H}(\td \psi+\chi)^2+H\td x^i\td x^i, \label{eq_GHmetric}
\end{equation}
where $x^i$, $i=1,2,3$ are cartesian coordinates on $\mathbb{R}^3$, $W=\partial_\psi$, $H$ is a harmonic function on $\mathbb{R}^3$, 
and the 1-form $\chi$ satisfies
\begin{equation}
	\star_3 \td \chi = \td H \; . \label{eq_chieq}
\end{equation}
Here $\star_3$ denotes the usual Hodge star operator on $\mathbb{R}^3$ with respect to the euclidean metric. In this chart the cartesian 
coordinates are related to $X^{(i)}$ via 
\begin{equation}
	\td x^i = \iota_W X^{(i)}\;. \label{eq_dxi} 
\end{equation}
Note, that (\ref{eq_dxi}) can be regarded as an equation on $\mathcal{B}\subset\Sigma$, as both $X^{(i)}$ and $W$ can be regarded as tensors 
on $\Sigma$, the orbit space of $V$. An important difference to the asymptotically flat case is that (\ref{eq_dxi}) does not define the 
functions $x^i$ globally on the DOC.

A useful result of~\cite{gauntlett_all_2003} is that if the solution admits a $U(1)$ Killing field commuting with $V$ that is also triholomorphic 
on $\mathcal{B}$, then the whole solution is {\it locally} determined by $H$ and three further harmonic functions, $K$, $L$, $M$ on $\mathbb{R}^3$ as 
follows. Let us define a function and two 1-forms (up to a gradient) on $\mathbb{R}^3$ by 
\begin{align}
	\omega_\psi &= \frac{K^3}{H^2}+\frac{3}{2}\frac{KL}{H}+M \; , \label{eq_omegapsi} \\
	\star_3 \td\hat\omega &= H \td M - M \td H +\frac{3}{2} (K \td L-L\td K) \; , \label{eq_omegahat} \\
	\star_3 \td\xi &= -\td K \; . \label{eq_xi}
\end{align}
Then $f$ and $\omega$ can be written as
\begin{align}
	f &= \frac{H}{K^2+HL} \; , \label{eq_f}\\
	\omega &= \omega_\psi(\td \psi+\chi) + \hat\omega \; ,\label{eq_omega}
\end{align}
while the Maxwell field takes the form
\begin{equation}
	F = \td A = \frac{\sqrt{3}}{2} \td \left(f(\td t + \omega)-\frac{K}{H}(\td \psi + \chi)-\xi\right)\;.\label{eq_Maxwell}
\end{equation}
We would like to emphasise again that this is a fully {\it local} result on $\widetilde{\mathcal{M}}$. 
Given a local solution, the corresponding set of harmonic functions $H, K, L, M$ is not unique. Indeed, one can check that 
\begin{align}
	H' = H, \quad K' = K + c H, \quad L' = L -2 c K - c^2 H, \quad M' = M - \frac{3}{2} c L + \frac{3}{2} c^2 K + \frac{1}{2} c^3 H \label{eq_harmonic_gauge}
\end{align}
yield the same solution for any $c\in \mathbb{R}$.

Following~\cite{breunholder_moduli_2019}, let us define a key spacetime invariant as
\begin{equation}
	N :=-\begin{vmatrix}
		g(V, V) & g(V, W)\\
		g(W, V) & g(W,W)
	\end{vmatrix} = \frac{f}{H} = \frac{1}{K^2+HL} = (-\det g)^{-1/2}\; , \label{eq_Ndef}
\end{equation}
where the last three equalities are valid on $\widetilde{\mathcal{M}}$ when $N>0$. Note that $N$ is preserved by both Killing fields, and 
its zeros in $\llangle \mathcal{M}\rrangle$ exactly coincide with $\mathcal{F}=\{p\in\llangle \mathcal{M}\rrangle| W_p=0 \}$ by 
our assumption that the span of Killing fields is timelike on $\llangle \mathcal{M}\rrangle$~\cite{katona_supersymmetric_2023}. It is 
also worth noting that from (\ref{eq_dxi}) and (\ref{eq_XXspacetime}) we have
\begin{equation}
	g^{-1}(\td x^i,\td x^j) = N\delta^{ij}\; \label{eq_dxinorm}
\end{equation}
irrespective of whether $V$ is timelike.

We now use the facts above to prove the following result for a general (not necessarily timelike) solution.
\begin{lemma}\label{lem_timelike_null}
	Either $\widetilde{\mathcal{M}}$ is dense in $\llangle\mathcal{M}\rrangle$ or $V$ is null on $\llangle\mathcal{M}\rrangle$.
\end{lemma}
\begin{proof}
	Let $\mathcal{N}$ be the set on which $V$ is null in the DOC, i.e. $\mathcal{N}:=f^{-1}(\{0\})\cap\llangle\mathcal{M}\rrangle$. 
	This is closed in the DOC by the continuity of $f$, hence 
	$\mathcal{N}=\operatorname{int} \mathcal{N}\cup \partial \mathcal{N}$. If $\operatorname{int} \mathcal{N}=\emptyset$, then 
	$\widetilde{\mathcal{M}}$ is dense in $\llangle\mathcal{M}\rrangle$. If $\operatorname{int} \mathcal{N}\neq\emptyset$, then assume for 
	contradiction that there exists some $p\in \partial\operatorname{int} \mathcal{N}$, and look at a simply connected neighbourhood $U$ of 
	$p$ in $\llangle\mathcal{M}\rrangle$. Since $U$ is simply connected, we can integrate (\ref{eq_dxi}) to 
	obtain $\bm x=(x^1, x^2, x^3):U\to \mathbb{R}^3$. Since $f(p)=0$, $p$ is not a fixed point of $W$ (otherwise the span of Killing fields 
	would be null contradicting assumption \ref{ass5D_timelike}). It follows that $N(p)>0$, and thus by continuity it is also positive on 
	(a possibly smaller) $U$, and by (\ref{eq_dxinorm}), $\bm x$ is a submersion, therefore open. We extend the definition of $H$ to $U$ by
	$H:=f/N$, which is smooth on $U$, and harmonic on the open set $\bm x(U)\subset\mathbb{R}^3$. In particular, it is zero on the open set 
	$\bm x(U\cap \operatorname{int}\mathcal{N})$. Since $H$ is harmonic and therefore analytic in $x^i$, it follows that $H\equiv 0$ on $U$, 
	but then $f\equiv0$ on $U$, which is a contradiction. We conclude that $\partial\operatorname{int}\mathcal{N}=\emptyset$, and hence it must 
	be that $\mathcal{N}=\llangle\mathcal{M}\rrangle$, so $V$ is null on $\llangle\mathcal{M}\rrangle$.
\end{proof}
In the context of classifying {\it global} solutions, we will refer to the first case in Lemma \ref{lem_timelike_null} as {\it timelike case}, 
and to the latter one as {\it null case}.

\subsection{Null case with axial symmetry} \label{sec_nullcase}

In this section we consider the case where $V$ is null on the DOC. The metric locally takes the form 
\cite{gauntlett_all_2003}
\begin{equation}
	g = -\mathcal{G}^{-1}(\mathcal{Q}\td u^2 + 2\td u\td v)+\mathcal{G}^2(\td x^i+b^i\td u)(\td x^i+b^i\td u)\;, \label{eq_nullmetric}
\end{equation}
where $\mathcal{Q}(u, \bm x)$, $b^i(u, \bm x)$ are functions, $\mathcal{G}(u, \bm x)$ for each $u$ is a harmonic function on $\mathbb{R}^3$ 
with cartesian coordinates $x^i$ (i.e. $\partial_i\partial_i\mathcal{G}(u, \bm x)=0$), $V=\partial_v$, and 
\begin{equation}
	X^{(i)}=\td u\wedge \td x^i.\label{eq_Xinull}
\end{equation}

We now assume the existence of an axial Killing field that preserves $V$ and $X^{(i)}$, and deduce the following lemma.
\begin{lemma} \label{lem_Wnull}
	For any point in $\llangle\mathcal{M}\rrangle$ there exist local coordinates in which the metric takes the form (\ref{eq_nullmetric}) 
	and the Killing field $W=\partial_u$, hence $\mathcal{G}, \mathcal{Q}, b^i$ only depend on $x^i$.
\end{lemma}
\begin{proof}
	We will use properties of $W$ together with the gauge freedoms that preserve the form of the solution \cite{gauntlett_all_2003}. 
	From $[V, W]=0\implies W^\mu=W^\mu(u, x^i)$, and
	\begin{align}
		0=\lie_WV^\flat = \left(\frac{W(\mathcal{G})}{\mathcal{G}^2}-\frac{\partial_uW^u}{\mathcal{G}}\right)\td u-\frac{\partial_iW^u}{\mathcal{G}}\td x^i\;,\label{eq_lieWV}
	\end{align}
	so $W^u$ only depends on $u$.
	Note that $N=(W^u)^2/\mathcal{G}^2>0$ on $\llangle\mathcal{M}\rrangle$, thus $W^u(u)\neq 0$. We define $\td u'=(W^u(u))^{-1}\td u$,  
	$\mathcal{G}' = \mathcal{G} / W^u$, and $x^{i}{}' = W^u x^i$ so that the forms of $g, V, X^{(i)}$ are unchanged, while in these coordinates 
	$W = \partial_{u'}+W^{i}{}'\partial_{i}{}' + W^v\partial_v$. Then, in the new coordinates (omitting primes)
	\begin{equation}
		0=\lie_WX^{(i)}=\td\left(\iota_WX^{(i)}\right) =\td ( \td x^i - W^i\td u ) = \td u\wedge \td W^i\;,
		\end{equation}
	which implies $\partial_j W^i=0$, i.e. $W^i=W^i(u)$. The coordinate transformation $x^{i}{}'=x^i+v^i(u)$ preserves the form 
	of the solution (after redefining $\mathcal{Q}$ and $\bm b$), and the choice $\td v^i/\td u = -W^i(u)$ yields $W=\partial_u+W^v\partial_v$. 
	Finally, the remaining gauge freedom allows us to change the $v=\const$ surfaces by $v' = v+h(u, \bm x)$. With 
	$\partial_u h(u, \bm x) = -W^v(u, \bm x)$ the axial Killing field becomes $W=\partial_u$ as claimed. In this coordinate system all 
	metric functions are independent of $u$. 
\end{proof}

It is useful to note that the remaining gauge freedom that preserves the form of the metric, the two-forms $X^{(i)}$, and the Killing fields is 
\begin{equation}
	v'=v+h(\bm x)\;. \label{eq_gauge_freedom}
\end{equation} 
This, however, changes $\mathcal{Q}$ and $\bm b$ as~\cite{gauntlett_all_2003}
\begin{align}
	\mathcal{Q}' = \mathcal{Q} - 2 b^i \partial_i h + \mathcal{G}^{-3} \partial_i h \partial_i h\;, \qquad \bm b' = \bm b - \mathcal{G}^{-3}\td h\;, \label{eq_gauge_freedom_Qa}
\end{align} 
hence these quantities are not gauge-invariant.

The lack of dependence on $u$ simplifies the analysis of~\cite{gauntlett_all_2003}, and the full solution can be obtained as follows.
$b^i$ is determined up to a gradient term (corresponding to the gauge freedom (\ref{eq_gauge_freedom}-\ref{eq_gauge_freedom_Qa})) by 
\begin{equation}
	\star_3 \td (\mathcal{G}^3\bm b)=\mathcal{G}\td\mathcal{K}-\mathcal{K}\td \mathcal{G}\;, \label{eq_GK}
\end{equation}
where $\star_3$ is the Hodge-star on flat $\mathbb{R}^3$, and $\mathcal{K}(x^i)$ is another harmonic function on $\mathbb{R}^3$. Choosing 
positive orientation to be given by $\td v\wedge \td u \wedge \td x^1\wedge \td x^2\wedge \td x^3$, (\ref{eq_starVF}-\ref{eq_DX}) 
yields\footnote{
	The sign difference compared to~\cite{gauntlett_all_2003} is due to the opposite orientation chosen.
} 
\begin{equation}
	F = \frac{1}{2\sqrt{3}}\td u\wedge \td\left(\frac{\mathcal{K}}{\mathcal{G}}\right)+\frac{\sqrt{3}}{2}\star_3 \td \mathcal{G}\;.\label{eq_Maxwell_null}
\end{equation}
Defining $D_{\bm b}:=b^i\partial_i$
and $\mathcal{W}_{ij}:=-\delta_{ij}D_{\bm b}\mathcal{G}-\mathcal{G}\partial_jb^i$, $\mathcal{Q}$ is a solution of 
\begin{equation}
	\partial_i\partial_i\mathcal{Q}=-2\mathcal{G}^2D_{\bm b}\mathcal{W}_{ii} + 2\mathcal{G}\mathcal{W}_{(ij)}\mathcal{W}_{(ij)}+\frac{2}{3}\mathcal{G}\mathcal{W}_{[ij]}\mathcal{W}_{[ij]}\;.\label{eq_Fcalc}
\end{equation}
Note, that (\ref{eq_Fcalc}) only determines $\mathcal{Q}$ up to a harmonic function $\mathcal{Q}_0$, hence the local solution 
is determined by three harmonic functions $\mathcal{G}$, $\mathcal{K}$, $\mathcal{Q}_0$.

From Lemma \ref{lem_Wnull} follows that locally $\td x^i=\iota_WX^{(i)}$ (as in the timelike case (\ref{eq_dxi})). Our assumption that the 
span of Killing fields must be timelike excludes any fixed points of $W$ in the DOC for the null case, and therefore $N>0$ on the DOC. It 
follows that 
$\mathcal{G}$ is globally defined by
\begin{equation}
	\mathcal{G}^{-1} = -g(W, V)=\pm\sqrt{N}\neq0\;. \label{eq_Gdef}
\end{equation}
We can flip the sign of $\mathcal{G}$ by redefining $u\to-u$ and $\mathcal{Q}\to-\mathcal{Q}$, so without loss of generality we will take 
\begin{equation}
	\mathcal{G}>0 \label{eq_Gpos}
\end{equation} on the DOC. On the horizon $W$ must be orthogonal to the generators of the horizon (since every Killing field is tangent to the horizon), which must be 
proportional to $V$, thus $N=0$ and $\mathcal{G}$ must diverge.

Note that from (\ref{eq_Maxwell_null}) and Lemma \ref{lem_Wnull} it 
follows that 
\begin{equation}
	\iota_W F = \frac{1}{2\sqrt{3}}\td\left(\frac{\mathcal{K}}{\mathcal{G}}\right)\;. \label{eq_magnetic_pot_null}
\end{equation}
The left-hand side is invariantly defined, but since the DOC is not necessarily simply connected, this does not define $\mathcal{K}$ 
globally. Still, in each local (sufficiently small) patch we find that $\mathcal{K}/\mathcal{G}$ is bounded. In particular, near the horizon 
$\mathcal{K}$ can diverge at most as $\mathcal{G}$ does. Later we will show that (\ref{eq_magnetic_pot_null}) indeed globally defines 
$\mathcal{K}$ on $\llangle \mathcal{M}\rrangle$.


\section{Asymptotics} \label{sec_asymptotics}

In this section we determine the asymptotic behaviour of the $U(1)$ Killing field $W$, and the cartesian coordinates $x^i$ using 
Definition \ref{def_KK}. As discussed in the previous section, generally we take the stationary Killing field, $\partial_0$ in 
asymptotic coordinates (\ref{eq_metric_AKK}), to be a linear 
combination of the other two, i.e.
\begin{equation}
	\partial_0 = \gamma^{-1} V + \frac{v_H}{\tilde L} W \label{eq_stat_ss}
\end{equation}
with $\gamma, v_H$ constants. For this we first need to look at the asymptotic form of Killing fields in an asymptotically Kaluza-Klein 
spacetime.

\subsection{Asymptotic form of Killing fields}

We first narrow down the possible form of the axial Killing field near spatial infinity. Proposition 2.1 of~\cite{chrusciel_killing_2006} 
for asymptotically flat spacetime states that the Killing fields asymptotically approach those of Minkowski. The statement carries over to 
asymptotically Kaluza-Klein spacetimes.
\begin{lemma}
	Let $(\mathcal{M}, g)$ be stationary and asymptotically Kaluza-Klein as in Definition \ref{def_KK}, and $K$ a Killing field that commutes with 
	$\partial_0$. Then there exist constants $\Lambda_{ij}=-\Lambda_{ji}$ such that
	\begin{align}
		K_j-u^i\Lambda_{ij}=\begin{cases}
			\ord(\tilde r^{1-\tau}) \text{ for } \tau \neq 1\;,\\
			\ord(\log \tilde r)\text{ for } \tau=1 \;.\label{eq_KFLambda}
		\end{cases}
	\end{align}
	If all $\Lambda$ vanish then there exist constants $A_\mu$ such that
	\begin{align}
		K_\mu-A_\mu=\ord(\tilde r^{-\tau})\; .
	\end{align}
	If all $\Lambda = A=0$ then $K=0$.
	\label{lem_AKKKilling}
\end{lemma}
\begin{proof}
	The proof is identical to the one in Appendix C of~\cite{beig_killing_1996}, but for completeness we outline it here in a 
	bit more detail. It is well-known that for a Killing field $K$
	\begin{equation}
		\nabla_\mu\nabla_\nu K_\lambda =R^\kappa{}_{\mu\nu\lambda}K_\kappa\;.\label{eq_KFR}
	\end{equation}
	Using this, we can write
	\begin{align}
		\partial_\mu K_\nu &= \nabla_\mu K_\nu + \Gamma_{\mu \nu}^\kappa K_\kappa \;, \\
		\partial_\mu \nabla_\nu K_\lambda &= R^\kappa{}_{\mu\nu\lambda}K_\kappa + \Gamma_{\mu\nu}^\kappa\nabla_\kappa K_\lambda + \Gamma_{\mu\lambda}^\kappa \nabla_\nu K_\kappa \;. \label{eq_KdK}
	\end{align}
	This means that for the $\tilde r$-derivatives we obtain
	\begin{align}
		\partial_{\tilde r} K_\nu &= \frac{u^i}{\tilde r}\left(\nabla_i K_\nu + \Gamma_{i \nu}^\kappa K_\kappa\right) \;, \label{eq_rKdKa}\\
		\partial_{\tilde r} \nabla_\nu K_\lambda &=\frac{u^i}{\tilde r}\left( R^\kappa{}_{i\nu\lambda}K_\kappa + \Gamma_{i\nu}^\kappa\nabla_\kappa K_\lambda + \Gamma_{i\lambda}^\kappa \nabla_\nu K_\kappa \right)\;, \label{eq_rKdK}
	\end{align}
	where we used $\partial_{\tilde r} = \tilde r^{-1}u^i\partial_i$.
	In Appendix \ref{app_asymptotics} we derive from Definition \ref{def_KK} that
	\begin{align}
		&R^\kappa{}_{\mu\nu\lambda}=\ord(\tilde r^{-\tau-2})\;,\qquad \Gamma^{\mu}_{i\nu}=\ord(\tilde r^{-\tau-1})\;, \nonumber\\
		&\Gamma^{\mu}_{0\nu}=\ord(\tilde r^{-\tau-1})\;,\qquad \Gamma^{\mu}_{{\tilde{\psi}}{\tilde{\psi}}}=\ord(\tilde r^{-\tau})\;.\label{eq_falloff}
	\end{align}
	Let us define $X = \sum_A f^A f^A$ for $f^A = (K_\mu, \tilde r\nabla_\mu K_\nu)$, for which
	\begin{equation}
		|\partial_{\tilde{r}} X| = \frac{1}{\tilde r}\left|2\sum_{AB}C_{AB}f^Af^B\right|\le \frac{2C'X}{\tilde{r}} \;, \label{eq_Xest}
	\end{equation}
	where the explicit form of the matrix $C_{AB}(u^i, \tilde \psi)$ can be obtained from (\ref{eq_rKdK}) and $C'>0$ is a constant. 
	For a uniform bound on $C_{AB}$ (second relation in (\ref{eq_Xest})) we used that due to (\ref{eq_falloff}) 
	$C_{AB}(\tilde r, \tilde\theta,\tilde\phi,\tilde\psi)\le B(\tilde\theta,\tilde\phi,\tilde\psi)\tilde r^0\le C'$,
	where $\tilde \theta,\tilde\phi$ (together denoted by $\theta^A$) are angular coordinates on $S^2$ and $B$ is some function. Therefore, by integrating (\ref{eq_Xest}), there exists a $\beta$ 
	such that $K_\mu = \ord(\tilde{r}^{\beta})$, $\nabla_\mu K_\nu = \ord(\tilde{r}^{\beta-1})$. Let us assume that $\beta\ge1+\tau$.  Using 
	our estimates for $K, \nabla K$ in the right-hand side of (\ref{eq_rKdK}) and then (\ref{eq_rKdKa}), we obtain an improved estimate with 
	$\beta\to\beta-\tau$. We iterate this procedure until $1\le\beta<1+\tau$. Then (\ref{eq_KdK}) yields
	\begin{align}
		\partial_i\nabla_\mu K_\nu =\ord(\tilde r^{\beta-\tau-2})\;, \qquad \partial_{\tilde \psi}\nabla_\mu K_\nu =\ord(\tilde r^{\beta-\tau-1})\;.
	\end{align}
	From the Lemma\footnote{
		The Lemma considers a function on a boost-type domain in $4D$. Since here nothing depends on $u^0$, in our case it holds 
		for $\mathbb{R}\times \Sigma_0$. The proof carries over with the only change being that (A.2) has an extra integral in the 
		Kaluza-Klein direction that vanishes in the $\tilde r\to\infty$ limit due to $\partial_{\tilde \psi}\nabla_\mu K_\nu\to0$.
	} of Appendix A in~\cite{chrusciel_invariant_1988} it follows that $\nabla_\mu K_\nu-\Lambda_{\mu\nu}=\ord(\tilde r^{\beta-\tau-1})$
	with some $\Lambda_{\mu\nu}$ constants, which we first assume that are not all zero. This substituted back into (\ref{eq_KdK}) improves 
	the estimates to $\beta\to1$ and $K_\mu=\ord(\tilde r)$. Finally, we obtain
	\begin{equation}
		|K_\mu(u^0, \tilde r, \theta^A,\tilde \psi) -\Lambda_{\nu\mu}u^\nu|\le |K_\mu(u^0_0, \tilde r_0, \theta^A_0,\tilde \psi_0) -\Lambda_{\nu\mu}u_0^\nu| + \bigg|\int_\Gamma \td (K_\mu -\Lambda_{\nu\mu}u^\mu)\bigg|=\ord(\tilde r^{1-\tau})
	\end{equation}
	with $\Gamma$ connecting $u^\mu_0$ with $u^\mu$ for some $u^\mu_0$. Now, consider for $k\in\mathbb{Z}$ at a given $(t, \tilde r,\tilde\theta,\tilde\phi)$
	\begin{equation}
		\left|\int_0^{4\pi k}(\partial_{\tilde \psi} K_\mu-\Lambda_{\tilde \psi\mu})\td\tilde \psi \right|= 4\pi \left|k\Lambda_{\tilde \psi\mu}\right|\le C\tilde r^{1-\tau}\;.
	\end{equation}
	This should hold for $\forall k \in\mathbb{Z}$, which implies that $\Lambda_{\tilde \psi\mu}=0$. $\Lambda_{0i}=0$ by 
	$0=\lie_{\partial_0}K^\mu = \partial_0 K^{\mu}$. Thus, we obtained (\ref{eq_KFLambda}).

	If $\Lambda_{ij}=0$ we have $\nabla_\mu K_\nu = \ord(\tilde r^{-\tau})$ and $K_\mu = \ord(\tilde r^{1-\tau})$. By a similar procedure we 
	can improve this estimate by $-\tau$ at each iteration until $\partial_rK_\mu=\ord(\tilde r^{-k\tau})$ with $1-k\tau<0\le 1-(k-1)\tau$
	and $k\in \mathbb{N}$. Then again by Lemma of Appendix A in~\cite{chrusciel_invariant_1988} and the arguments above, there is a constant 
	$A_\mu$ such that $K_\mu - A_\mu=\ord(r^{-\tau})$.

	If $\Lambda_{ij}=0$ and $A_\mu=0$, by the iterative process $K_\mu = \ord(\tilde r^{-\kappa})$ for any $\kappa>0$. Then integrating 
	(\ref{eq_Xest}) as
	\begin{equation}
		-\frac{2C'X}{\tilde r}\le \partial_{\tilde r}X
	\end{equation}
	we obtain $\tilde r_0^{2C'}X(\tilde r_0) \le \tilde r^{2C'}X(\tilde r)\to 0$ as $\tilde r\to \infty$, thus $X(\tilde r_0)=0$, which means 
	that the Killing field and its first derivative is zero at a point, hence it is zero everywhere.
\end{proof}

\subsection{Chart at spatial infinity from supersymmetry}

In this section we construct a chart at spatial infinity defined from Killing spinor bilinears as described in Section \ref{sec_recap}. We 
also show that $W$ is tangent to the Kaluza-Klein direction.

The proof is quite different depending on whether $v_H=0$ or 
$v_H\neq0$ in (\ref{eq_stat_ss}), and for the former, we first need to derive the asymptotic form of the hyper-K\"ahler structure. 
In this case the supersymmetric Killing field $V$ coincides with the stationary one, which means $V$ is timelike 
in the asymptotic region (thus we are in the timelike case since $V$ is not globally null). Furthermore, its limit at infinity is also 
timelike (i.e. $\lim_{\tilde r\to\infty}g(V, V)<0$), hence we are in the {\it asymptotically timelike} case, as opposed to the (timelike or null) case 
when $\lim_{\tilde r\to\infty}g(V, V)=0$, which we call {\it asymptotically null}.

\pagebreak[1]
For the case $v_H=0$, let us normalise the Killing spinor such that $\gamma=1$, i.e. $V=\partial_0$ in the 
asymptotic coordinates 
(\ref{eq_metric_AKK}). Definition \ref{def_KK} implies that on the asymptotically Kaluza-Klein end we have\footnote{The solution is invariant 
under changing the sign of $f, H, L, \psi, t$ simultaneously, so without loss of generality we take $f>0$ at infinity.}
\begin{align}
	f &= 1+ \ord(\tilde r^{-\tau}) \; ,\\
	\omega &=  \ord(\tilde r^{-\tau}) \td u^a  \; ,\\
	h &= \underbrace{\delta_{ij}\td u^i \td u^j + \tilde L^2\td \tilde\psi^2}_{=:h_0}  +\ord(\tilde r^{-\tau})\td u^a\td u^b \;, \label{eq_hAE}
\end{align}
and the asymptotically Kaluza-Klein end $\Sigma_0\subset \mathcal{B}$ is a circle fibration over $\mathbb{R}^3\setminus B^3$. For the asymptotic 
form of the hyper-K\"ahler structure we prove the following lemma.

\begin{lemma}\label{lem_XAKK}
	Assuming $v_H=0$, on the asymptotically Kaluza-Klein end the complex structures of $(\mathcal{B}, h)$ can be written as\footnote{
		Error terms of tensors throughout this section refer to their components in the asymptotic coordinates of Definition \ref{def_KK}.
	}
	\begin{equation}
		X^{(i)} = \Omega_-^{(i)} + \mathcal{O}(\tilde r^{-\tau}) \; , \label{eq_XAF}
	\end{equation}
	where $\Omega_-^{(i)}$ are a standard basis of anti-self-dual 2-forms on $\mathbb{R}^3\times S^1$ with respect to the orientation 
	$\tilde L\td\tilde\psi\wedge \td u^1\wedge \td u^2\wedge\td u^3$,
	\begin{align}
		\Omega_-^{(i)} = \tilde L\td\tilde\psi\wedge \td u^i-\frac{1}{2}\epsilon_{ijk}\td u^j\wedge \td u^k.\label{eq_Omega}
	\end{align} 
\end{lemma}
\begin{proof}
	The proof is analogous to the asymptotically flat case (Lemma 4 of~\cite{katona_supersymmetric_2023}). From the quaternion algebra 
	$X^{(i)}\cdot X^{(i)}=-4$ (no sum over $i$), hence to leading order $X^{(i)}_{ab}=\ord(1)$. The curvature
	of $h$ is Ricci-flat (since it is hyper-K\"ahler), and for that one can show that $\Gamma_{ia}^b=\ord(\tilde r^{-\tau-1})$ and 
	$\Gamma^a_{\tilde\psi\tilde\psi}=\ord(\tilde r^{-\tau})$ (see Appendix \ref{app_asymptotics}). Then $\nabla^{(h)} X^{(i)}=0$ implies that
	\begin{align}
		\partial_jX^{(i)}_{ab}=\ord(\tilde r^{-\tau-1}), \qquad \partial_{\tilde\psi} X^{(i)}_{jk}=\ord(\tilde r^{-\tau-1}), \qquad  \partial_{\tilde\psi} X^{(i)}_{\tilde\psi j}=\ord(\tilde r^{-\tau})\;,
	\end{align}
	which  after integration yields $X^{(i)}_{ab}=\bar{X}^{(i)}_{ab}+\ord(\tilde{r}^{-\tau})$, where $\bar{X}^{(i)}_{ab}$ are constants.

	Let us now define $\bar{X}^{(i)}_\pm:= \frac{1}{2}(1\pm\star_{h_0})\bar{X}^{(i)}$ as the SD/ASD part of $\bar{X}^{(i)}$ with respect to $h_0$ 
	(defined in (\ref{eq_hAE})). Then using 
	\begin{equation}
		\star_h X^{(i)} = \star_{h_0} X^{(i)} + \ord(\tilde r^{-\tau}) = \star_{h_0} \bar X^{(i)} + \ord(\tilde r^{-\tau}) = \bar X^{(i)}_+- \bar X^{(i)}_- + \ord(\tilde r^{-\tau})
	\end{equation}
	we deduce by anti self-duality of $X^{(i)}$ that the constant $\bar X^{(i)}_+ = \ord(\tilde{r}^{-\tau})=0$. Equation (\ref{eq_quaternion})
	implies that $\bar X^{(i)}$ obeys the quaternion algebra with respect to $h_0$. Since (\ref{eq_Omega}) forms a basis of ASD 2-forms with 
	respect to $h_0$ also satisfying the quaternion algebra, we can always perform a global $SO(3)$ rotation of ${X}^{(i)}_-$ such that 
	$\bar{X}^{(i)}_- = \Omega^{(i)}_-$.
\end{proof}

Next we use Lemma \ref{lem_AKKKilling} and triholomorphicity to deduce the asymptotic form of $W$. The following result holds for both 
timelike and null cases.
\begin{lemma}\label{lem_Wasymp}
	For any values of $v_H$, on the asymptotically Kaluza-Klein end we can choose coordinates such that the metric 
	is of the form (\ref{eq_metric_AKK}) with the stationary Killing field $\partial_0$, and the $U(1)$ Killing field is given by
	\begin{equation}
		W = \partial_{\tilde\psi}\; . \label{eq_WAE}
	\end{equation}
\end{lemma}
\begin{proof}
	By Lemma \ref{lem_AKKKilling} the leading order behaviour of $W$ is determined by constants $\Lambda_{ij}$ and $A_\mu$ (using the 
	notation of the Lemma). Since $W$ has closed orbits, $A_0=0$ and $W^0$ must be subleading in $\tilde r$.

	First we will consider the general case $v_H\neq0$, and assume that $\Lambda_{ij}$ are not all zero for $W$. Then $W$ generates a rotation on 
	$\mathbb{R}^3$ (possibly simultaneously with a rotation in the Kaluza-Klein direction), therefore, without loss of generality we can write
	\begin{equation}
		W  =  \Lambda \partial_\phi+\ord(\tilde r^{1-\tau})\; , \label{eq_Wansatzr1}
	\end{equation}
	where $(\tilde r, \theta,\phi)$ are the usual spherical coordinates on $\mathbb{R}^3$ and $\Lambda$ is a constant. Then the norm of $V$
	\begin{equation}
		0\ge \gamma^{-2} g(V, V) = g_{00}+ \frac{v_H^2}{\tilde L^2} g(W,W) - 2\frac{v_H}{\tilde L} W_0 = \frac{v_H^2}{\tilde L^2} \Lambda^2\tilde r^2 \sin^2\theta + \ord(\tilde r^{2-\tau}) + \ord(1).
	\end{equation}
	It follows that $\Lambda=0$, therefore $\Lambda_{ij}=0$, 
	and by Lemma \ref{lem_AKKKilling}
	\begin{equation}
		W = A_i\partial_i + A_\psi\partial_\psi + \ord(\tilde r^{-\tau}).
	\end{equation}
	Since $W$ has closed orbits, $A_i=0$, therefore $W=\partial_{\tilde\psi} + \ord(\tilde r^{-\tau})$, where the normalisation has 
	been chosen such that $W$ has $4\pi$-periodic orbits.

	In the spacetime the integral curves of $W$ wind around the Kaluza-Klein direction, hence we can adapt 
	coordinates $(u^0, u^i,\tilde \psi)\to( u^0{}',u^i{}', \tilde\psi') = (u^0+ \lambda^0, u^i+ \lambda^i,\tilde \psi+ \lambda^\psi)$ 
	to the action of $W$ such that $W=\partial_{\tilde\psi'}$ exactly\footnote{
		$W$ acts freely for large enough $\tilde r$, as fixed points are excluded by the form of the metric, and we will see 
		in Section \ref{sec_orbit} that exceptional orbits are excluded by triholomorphicity.
	} and $V = \partial_0{}'$. To see this, note that $W$ commutes with the stationary Killing field, thus $\partial_0\lambda^\mu=0$, and the stationary Killing field is 
	unchanged by the coordinate transformation, i.e. $\partial_0=\partial_0{}'$. Since $\lambda=\ord(\tilde r^{-\tau})$, the 
	metric only receives $\ord(\tilde r^{-\tau})$ corrections, so $g'$ has the same form as (\ref{eq_metric_AKK}). Thus, (after omitting primes) 
	we get the claimed result.

	For the special case $v_H=0$ we work on the hyper-K\"ahler base and use the triholomorphic property of $W$. $\lie_WV = 0$ and $\lie_Wf=0$, 
	thus the projection of $\pi_*(W^0, W^a) = W^a$ defines a Killing vector on the base space. To leading order $W$ is equal to $\pi_*W$ 
	(thus in the following we do not distinguish between the two). Again, assuming that $\Lambda_{ij}$ are not all zero for $W$, without 
	loss of generality we can write it as (\ref{eq_Wansatzr1}). Using the asymptotic form of the complex structures from 
	Lemma \ref{lem_XAKK} one can check that $\partial_\phi$ preserves only one of them, which would contradict triholomorphicity. 
	It follows again that $\Lambda=0$, and by the same arguments as for the general case, $W$ must have the claimed form.

\end{proof}

Next, we construct the asymptotic charts for the cases when $V$ is asymptotically timelike or null separately. For the asymptotically timelike case we 
have the following result.

\begin{lemma} \label{lem_AKKGH}
	Assume that $V$ is asymptotically timelike. Then the base of the asymptotically Kaluza-Klein end is covered by a single chart in which 
	(together with the vertical coordinate $\psi$) the spacetime metric is of Gibbons-Hawking form ((\ref{eq_hdef}) with (\ref{eq_GHmetric})). 
	The Gibbons-Hawking coordinates are related to the asymptotic coordinates by 
	\begin{align}
		t = \gamma^{-1}u^0\;,& \qquad\qquad \psi = \tilde \psi+\frac{v_H}{\tilde L}u^0\sim\psi+4\pi\;, \label{eq_tAKK}\\
		&x^i =\tilde L\gamma u^i + \ord(\tilde r^{1-\tau})\;, \label{eq_xAKK}
	\end{align}
	where (with appropriate normalisation of $V$)
	\begin{equation}
		\gamma = \frac{1}{\sqrt{1-v_H^2}}\;,\label{eq_gammanorm}
	\end{equation}
	and the Killing fields are given by $V = \partial_t$ and $W=\partial_{\psi}$, and the cartesian coordinates $x^i$ provide a surjection to
	$\mathbb{R}^3\setminus B^3_R$ for some $R>0$.
\end{lemma}
\begin{proof}
	The supersymmetric Killing field is given by $V=\gamma \partial_0-\gamma v_H\tilde L^{-1}\partial_{\tilde\psi}$. Since it is asymptotically timelike, 
	we can normalise it such that
	\begin{equation}
		f^2 = -g(V, V) = \gamma^2(1-v_H^2)+ \ord(\tilde r^{-\tau}) = 1+\ord(\tilde r^{-\tau}),
	\end{equation}
	which in terms of the constants $\gamma, v_H$ yields (\ref{eq_gammanorm}) with $|v_H|\le 1$. Let us define
	\begin{align}
		t = \gamma^{-1}u^0, \qquad\qquad \psi = \tilde \psi+\frac{v_H}{\tilde L}u^0\;,
	\end{align}
	so that $W=\partial_\psi$ and $V=\partial_t$. In these coordinates the base metric becomes
	\begin{equation}
		h = \tilde L^2\gamma^2\td\psi^2 + \delta_{ij}\td u^i\td u^j +\ord(\tilde r^{-\tau}).
	\end{equation}
	This has the same form as the base metric in Lemma \ref{lem_XAKK} with $\tilde L^2\to \tilde L^2\gamma^2$. One can use similar 
	arguments to deduce that the hyper-K\"{a}hler two-forms are 
	\begin{equation}
		X^{(i)} = \tilde L\gamma \td\psi\wedge\td u^i + \frac{1}{2}\epsilon_{ijk}\td u^j\wedge\td u^k + \ord(\tilde r^{-\tau}).
	\end{equation}
	Hence, for the cartesian one-forms we get
	\begin{equation}
		\iota_WX^{(i)} = \tilde L\gamma\td u^i + \ord(\tilde r^{-\tau}).
	\end{equation}
	Since $\iota_W(\iota_WX^{(i)} )=0=\mathcal{L}_W(\iota_WX^{(i)})$, and $W$ is tangent to the circle fibres, the closed 1-form 
	$\iota_WX^{(i)}$ descends to the base of the fibration $\mathbb{R}^3\setminus B^3$, which is simply connected. Hence, by (\ref{eq_dxi}) 
	the cartesian coordinates $x^i$ can be globally integrated on (and uplifted to) the asymptotic end to get (\ref{eq_xAKK}). Since $(u^0, u^i)$ 
	form a single chart of the base of the asymptotically Kaluza-Klein end, so do the Gibbons-Hawking coordinates $(t, x^i)$. On this chart 
	$N = \tilde L^2\gamma^2+\ord(\tilde r^{-\tau})>0$, hence by (\ref{eq_Ndef}) the metric is invertible.
\end{proof}

Finally, we consider the case when $V$ is asymptotically (or globally) null, i.e. $\lim_{\tilde r\to\infty}g(V,V)=0$, and thus without loss of generality 
we can take
\begin{equation}
	V = \partial_0 - \tilde L^{-1}\partial_{\tilde{\psi}}. \label{eq_Vnull}
\end{equation} 
We then have the following result.
\begin{lemma}\label{lem_nullR3}
	Assume $V$ is asymptotically (possibly globally) null as in (\ref{eq_Vnull}). Then the base of the asymptotically Kaluza-Klein end is covered by a single 
	coordinate chart of $(t, x^i)$ with 
	\begin{align}
		t = u^0\;, \qquad x^i = \tilde L u^i+ \ord(\tilde r^{1-\tau/2})\;. \label{eq_xnull}
	\end{align}
	Let 
	\begin{equation}
		\psi = \tilde\psi + L^{-1}u^0\sim\psi+4\pi \label{eq_psinull}
	\end{equation}
	parametrise the fibres. In such a chart 
	$V=\partial_t$ and $W=\partial_\psi$.
\end{lemma}
\begin{proof}
	From (\ref{eq_VX}) it follows that the most general form of $X^{(i)}$ is 
	\begin{equation}
		X^{(i)} = T^{(i)}_j(\td u^0+ \tilde L\td\tilde\psi)\wedge\td u^j+\frac{1}{2}R^{(i)}_{jk}\td u^j\wedge \td u^k\;
	\end{equation}
	for some functions $T^{(i)}_j$ and $R^{(i)}_{jk}$. From the $00$ component of (\ref{eq_XXspacetime}) using (\ref{eq_metric_AKK}) we obtain
	\begin{equation}
		(\delta_{kl}+\ord(\tilde r^{-\tau}))T^{(i)}_k T^{(j)}_l = \delta_{ij} + \ord(\tilde r^{-\tau}),\label{eq_TO3}
	\end{equation}
	hence $T^{(i)}_j=\ord(1)$. From (\ref{eq_metric_AKK}), (\ref{eq_Vnull}) and (\ref{eq_Vnorm}) follows $f = \ord(\tilde r^{-\tau/2})$, and 
	using (\ref{eq_VXdual}), one can check that $R^{(i)}_{jk}=\ord(\tilde r^{-\tau/2})$.
	The spacetime covariant derivative of $X^{(i)}$ has the form $\nabla X^{(i)}\sim F \star X^{(i)}$ (\ref{eq_DX}), thus 
	\begin{equation}
		\partial_\mu T^{(i)}_{j} = \partial_\mu X^{(i)}_{0j} = \ord(\tilde r^{-\tau/2-1})\;, \label{eq_dT}
	\end{equation}
	where we used that the relevant Christoffel symbols decay as $\ord(\tilde r^{-\tau-1})$ (Appendix \ref{app_asymptotics}) and that 
	$F=\ord(\tilde r^{-\tau/2-1})$ (see Remark after Definition \ref{def_KK}). After integration of (\ref{eq_dT}) we get
	\begin{equation}
		T^{(i)}_j = \overline T^{(i)}_j + \ord(\tilde r^{-\tau/2})
	\end{equation}
	for some constants $\overline T^{(i)}_j$. Furthermore, these constants satisfy (\ref{eq_TO3}), hence $\overline T\in O(3)$. By a global 
	orthogonal transformation of $u^{i}$, we can arrange that 
	\begin{equation}
		X^{(i)}= (\td u^0+ \tilde L\td\tilde\psi)\wedge\td u^i + \ord(\tilde r^{-\tau/2})\;, 
	\end{equation}
	and thus 
	\begin{equation}
		\iota_WX^{(i)} =  \tilde L\td u^i+ \ord(\tilde r^{-\tau/2}).
	\end{equation}
	By the same argument as in Lemma \ref{lem_AKKGH} $x^i$ can be integrated to get 
	\begin{equation}
		x^i = \tilde L u^i+ \ord(\tilde r^{1-\tau/2}), 
	\end{equation}
	and $(t=u^0, x^i)$ is a global chart on the base of the asymptotically Kaluza-Klein end. 
\end{proof}

The notation of the coordinates ($t, \psi, x^i$) matches that of the timelike (but here asymptotically null) case ((\ref{eq_hdef}) and 
(\ref{eq_GHmetric})), but the result is valid for the globally null case as well. Then the coordinates of (\ref{eq_nullmetric}) are given 
by $v=u^0$ and $u=\psi$ instead.

We have thus in each case constructed a chart on the asymptotically Kaluza-Klein end that is adapted to the Killing spinor bilinears. 
In particular, a suitable time coordinate $t$ and the cartesian coordinates $x^i$, related to the bilinears $X^{(i)}$ by (\ref{eq_dxi}), 
provide a single chart of the base of the fibration at infinity.


\section{Near-horizon analysis of null supersymmetric solutions}\label{sec_NHnull}

\subsection{Classification of null supersymmetric near-horizon geometries}
The near-horizon analysis and classification of near-horizon geometries in~\cite{reall_higher_2004} assume that the supersymmetric Killing 
field becomes timelike in some neighbourhood of the horizon, which is no longer true in the globally null case. In this section we show that 
the same results apply in the null case, and we deduce some technical results for later use.

The results of~\cite{kayani_symmetry_2018} show that a supersymmetric near-horizon geometry, even in the null class, must be maximally 
supersymmetric. These solutions were classified in~\cite{gauntlett_all_2003}, and thus the near-horizon geometry must be $AdS_3\times S^2$ 
or a plane wave solution. While there is only one way $AdS_3\times S^2$ is realised as a near-horizon geometry of a null solution, as 
it has a unique null supersymmetric Killing field, if one were to use these results, one would need to determine how the plane wave geometry 
is compatible with being a near-horizon geometry. Instead, we prefer to give our self-contained treatment using the method 
of~\cite{reall_higher_2004}, which will reveal that the only other possible near-horizon geometry is the trivial plane wave with flat 
geometry.

Following~\cite{reall_higher_2004} we introduce Gaussian null coordinates near a connected component of the horizon. In these coordinates 
the horizon is at $\lambda =0$, $V=\partial_v$, and the metric takes the form
\begin{equation}
	g = 2\td \lambda \td v + 2\lambda h \td v + \gamma\;,\label{eq_metricNH}
\end{equation}
where $\gamma$ and $h$ are a family of metrics and 1-forms on the 3-manifolds $H_{v,\lambda}$, which are the $v=\const$, 
$\lambda = \const$ hypersurfaces. Note that since $V=\partial_v$ is Killing, the metric components do not depend on $v$. 
Let $H=H_{v, 0}$ denote a spatial cross-section of the horizon. 
Equations (\ref{eq_VX}-\ref{eq_VXdual}) imply that we can choose a coframe $\{Z^{(i)}\}_{i=1}^3$ of $\gamma$ such that
\begin{equation}
	X^{(i)} = (\td\lambda + \lambda h)\wedge Z^{(i)}\;. \label{eq_XiNH}
\end{equation}
Equation (\ref{eq_starVF}) implies that $V$ is hypersurface orthogonal, which is equivalent to 
\begin{equation}
	\hat \td h = \lambda h\wedge \partial_\lambda h \;,\label{eq_dh} 
\end{equation}
where $\hat \td$ denotes the exterior derivative projected onto $H_{v, \lambda}$ (i.e. it does not include $\lambda$ derivatives). 
Closedness of $X^{(i)}$ is equivalent to (\ref{eq_dh}) together with 
\begin{equation}
	\hat\td Z^{(i)} = \partial_\lambda\left(\lambda h\wedge Z^{(i)}\right) \;.\label{eq_dZ}
\end{equation}

To determine the Maxwell field, following~\cite{gauntlett_all_2003}, it is convenient to work in the spacetime coframe 
\begin{align}
	e^+ = \td \lambda + \lambda h\;, \qquad e^- = \td v\;, \qquad e^i = Z^{(i)}\;.\label{eq_vielbeinNH}
\end{align}
Equation (\ref{eq_VF}) implies that $F = F_{+i}e^+\wedge e^i + \frac{1}{2}F_{ij}e^i\wedge e^j$. $F_{ij}$ is determined by (\ref{eq_starVF}), 
while $++j$ component of (\ref{eq_DX}) determines $F_{+i}$. Choosing $e^-\wedge e^+\wedge e^1\wedge e^2\wedge e^3$ to be positively oriented, 
after some algebra we get
\begin{equation}
	F = \frac{\sqrt{3}}{2}\left[\frac{1}{3}\epsilon_{ijk}\gamma^{-1}\left( Z^{(j)}, \partial_\lambda Z^{(k)}\right)X^{(i)}  -\star_\gamma(h+\lambda\partial_\lambda h) \right]\;,\label{eq_MxwNH}
\end{equation}
where $\star_\gamma$ is the Hodge star of $\gamma$ with orientation $e^1\wedge e^2\wedge e^3$, and we used that $X^{(i)}=e^+\wedge e^i$.

Equation (\ref{eq_MxwNH}) agrees with (3.35) of~\cite{reall_higher_2004} in the null limit, but its derivation does not rely on the assumption 
that $V$ becomes timelike outside the horizon. The analysis of~\cite{reall_higher_2004} determines the leading order behaviour of the metric 
quantities $h, Z^{(i)}, \gamma$, which applies without modification. For completeness, we sketch it here as well. (\ref{eq_dh}) and Bianchi 
identity for (\ref{eq_MxwNH}) implies 
\begin{align}
	\hat\td h = 0\; \qquad\text{and}\qquad \hat \td \star_\gamma h =0 \quad \text{ on $H$,} \label{eq_dhdh}
\end{align}
respectively, from which follows
\begin{equation}
	0=-(\hat\td\star_\gamma\hat\td\star_\gamma +\star_\gamma\hat\td\star_\gamma\hat\td)h =\hat\nabla^2 h - \widehat{\operatorname{Ric}}\cdot h\quad \text{ on $H$}\;, \label{eq_Laplace}
\end{equation}
where $\hat\nabla$ and $\widehat{\operatorname{Ric}}$ are the Levi-Civita connection of $\gamma$ and its Ricci tensor, and $\cdot$ is with respect to $\gamma^{-1}$. 
The $j+k$ component of (\ref{eq_DX}) at $\lambda=0$ yields an expression for $\hat\nabla Z^{(i)}$ on $H$, which after taking another derivative 
and antisymmetrising yields 
\begin{equation}
	\widehat{\operatorname{Ric}} =h^2\gamma - h\otimes h - \hat\nabla h \quad \text{ on $H$.} \label{eq_Ricci}
\end{equation}
Then considering the integral $I = \int_H |\nabla h|^2_\gamma \operatorname{dvol}$, after integration by parts and using (\ref{eq_dhdh}-\ref{eq_Ricci}), one can 
show that 
\begin{equation}
	\hat\nabla h =0 \quad \text{ on $H$.} \label{eq_nablah}
\end{equation}
Equation (\ref{eq_dZ}) implies that $Z^{(i)}$ are hypersurface orthogonal, and without loss of generality there exist coordinates $z^i$ and a function 
$K(z)$ such that 
\begin{align}
	Z^{(i)} = K \td z^i + \ord(\lambda)\; \qquad\text{and}\qquad h = \hat\td \log K + \ord(\lambda).
\end{align}
Equation (\ref{eq_nablah}) imposes a condition on $K$, which has two solutions, one corresponding to flat near-horizon geometry with 
$T^3$ horizon topology, which is not allowed~\cite{galloway_generalization_2006}, the other one is
\begin{equation}
	K = K_0 \exp(-\psi)\;, 
\end{equation} 
where $\psi = \frac{1}{2}\log(z^iz^i)$ and $K_0$ is some constant. This corresponds to $S^2\times S^1$ horizon geometry with 
\begin{align}
	h = -\td\psi+\ord(\lambda)\;, \quad Z^{(i)} = K_0(\hat x^i \td \psi + \td \hat x^i)+\ord(\lambda)\;, \quad  \gamma = K_0^2(\td\psi^2 + \td\hat x^i\td\hat x^i)+\ord(\lambda)\;,
\end{align}
where we introduced $\hat x^i = z^i \exp(-\psi)$ that satisfy $\hat x^i\hat x^i=1$ and $\td \hat x^i\td\hat x^i$ is the round metric on $S^2$.

\subsection{Imposing axial symmetry}\label{ssec_axialNH}

Building on the above results of~\cite{reall_higher_2004}, it has been shown in~\cite{katona_supersymmetric_2023} that a $U(1)$ Killing field 
that preserves $V$ and $X^{(i)}$ must be of the form $W = W^\psi \partial_\psi + W^v\partial_v +\ord (\lambda)$ for a constant non-zero $W^\psi$ 
and a function $W^v$ on $H$. In fact, since $W$ is spacelike on the horizon\footnote{If $W$ were null at the horizon, it would be parallel to 
$V=\partial_v$, and since $[V, W]=0$, $W^v$ would not depend on $v$. Thus, the coordinate $v$ would be periodic by the periodicity of the orbits of $W$, 
which cannot happen.}, one can choose $v=\const$ surfaces such that $W$ is tangent to $H$, i.e. $W^v=0$, 
in which case one can show that 
\begin{equation}
	W=W^\psi\partial_\psi \label{eq_WNH}
\end{equation}
exactly in some neighbourhood of the horizon (see Remark after Lemma 6 
of~\cite{katona_supersymmetric_2023}). In the following we will work in such a coordinate system.

Another result of~\cite{katona_supersymmetric_2023} is that $\iota_W X^{(i)}\propto\td\lambda$\footnote{In fact, the explicit form of $W$ 
is not necessary for this result. From (\ref{eq_dxinorm}) $\td x^i$ are null on the horizon, and 
since $\iota_V\iota_WX^{(i)}=0$ by (\ref{eq_VX}), $\iota_WX^{(i)}\propto V^\flat = \td \lambda$ on $H$.}, hence if one can integrate 
(\ref{eq_dxi}) to obtain functions $x^i$, these are constant on each connected component of the horizon (i.e. each connected component of the 
horizon is mapped to a point in $\mathbb{R}^3$). In the null case the horizon topology is $S^2\times S^1$ and $W$ is tangent to the 
circle fibres, thus in a neighbourhood of the horizon one can integrate (\ref{eq_dxi}) to obtain (up to an additive constant which for 
simplicity we set to zero)~\cite{katona_supersymmetric_2023}
\begin{equation}
	x^i = -W^\psi K_0 \lambda\hat x^i +\ord(\lambda^2)\;, \label{eq_xNH}
\end{equation}
and the radial distance on $\mathbb{R}^3$ from a horizon component is
\begin{equation}
	r := \sqrt{x^ix^i} = |W^\psi K_0|\lambda + \ord(\lambda^2).
\end{equation}

For later reference, we now look at the next order in $\lambda$ and prove the following lemma.
\begin{lemma} \label{lem_WNH}
	Near a horizon component $g(W,W) = \alpha_0 + r \alpha_1 +\ord(r^2)$ for some constants $\alpha_0$, $\alpha_1$.
\end{lemma}
\begin{proof}
	From the leading order metric and (\ref{eq_WNH}) follows that $g(W,W)= (W^\psi K_0)^2 +\ord(r)$ thus $\alpha_0=(W^\psi K_0)^2$ is 
	trivially constant. In the second-to-leading order we introduce
	\begin{align}
		Z^{(i)} &= K_0\left[\hat x^i \td\psi + \td \hat x^i + \lambda\left(Z^{(i)}_{1,\psi}\td \psi + \hat Z^{(i)}_1 \right) +\ord(\lambda^2)\right]\;, \\
		h &= -\td\psi + \lambda(h_{1,\psi}\td\psi + \hat h_1) +\ord(\lambda^2)\;,
	\end{align}
	where $Z^{(i)}_{1,\psi}$ and $ h_{1,\psi}$ are functions, and $\hat Z^{(i)}_1$ and $\hat h_1$ are one-forms on $S^2$. (Note that $h$ and 
	$Z^{(i)}$ cannot depend on $\psi$, since $\partial_\psi$ is Killing that also preserves $X^{(i)}$.) (\ref{eq_dh}) implies that 
	\begin{equation}
		\hat h_1 =  \td h_{1,\psi} \;, \label{eq_dh1}
	\end{equation}
	and from (\ref{eq_dZ}) it follows that
	\begin{align}
		Z^{(i)}_{1,\psi} = 2h_{1,\psi}\hat x^i + 4\zeta^i\;, \qquad \hat Z^{(i)} = 2h_{1,\psi}\td \hat x^{i}+2\td \zeta^i\;
	\end{align}
	for some functions $\zeta^i$ on $S^2$. The norm of $W$ is given by
	\begin{equation}
		g(W,W) = \left(K_0W^\psi\right)^2 + 4|K_0W^\psi| \left(h_{1,\psi}+ 2 \zeta^i\hat x^i\right)r + \ord(r^2).\label{eq_WnormNH}
	\end{equation}

	We now show that $h_{1,\psi}$ and $\zeta^i\hat x^i$ are constants. Bianchi identity for (\ref{eq_MxwNH}) in second-to-leading order 
	yields that 
	\begin{align}
		&\td\star_2\td\left(h_{1,\psi}+ \zeta^i\hat x^i\right) = 0\,\label{eq_hpsi}\\
		&\epsilon_{ijk}\hat x^i\td \langle\td \hat x^j, \td \zeta^k\rangle = \epsilon_{ijk}\hat x^i\zeta^k\td x^j\;,\label{eq_zeta}
	\end{align}
	where $\star_2$ and $\langle,\rangle$ are the Hodge star operator and the inverse metric on the unit $S^2$, where the orientation is 
	related to the horizon orientation by $\epsilon_\gamma = K_0^3 \td \psi \wedge \epsilon_{S^2}$. From (\ref{eq_hpsi}) and 
	the compactness of $S^2$ follows that $h_{1,\psi}+ \zeta^i\hat x^i$ is constant. (\ref{eq_zeta}) can be written as 
	\begin{equation}
		\star_2\td\star_2\td\left(\zeta^i\td \hat x^i\right) = -\td\left(\zeta^i\hat x^i\right)\;.\label{eq_zetax}
	\end{equation}
	Acting with $\star_2\td \star_2$ on (\ref{eq_zetax}) yields that $\zeta^i\hat x^i$ is harmonic on $S^2$, hence it is constant. Thus, 
	both $h_{1,\psi}$ and $\zeta^i\hat x^i$ are constants, and the norm of $W$ has the claimed form.
\end{proof}

The above proof solely relies on the near-horizon geometry, and does not use the form of the solution in the DOC. If one uses (\ref{eq_nullmetric}),  
(\ref{eq_Maxwell_null}), and Lemma  \ref{lem_Wnull}, there is an alternative way of proving that $h_{1,\psi}$ and $\zeta^i\hat x^i$ are constants. Integrating
(\ref{eq_dxi}) to second-to-leading order yields 
\begin{equation}
		r = |W^\psi K_0|\left[\lambda + \left(h_{1,\psi}+2\zeta^i\hat x^i\right)\lambda^2\right] + \ord (\lambda^3)\;.\label{eq_rlambda}
\end{equation}
Using this, the inner product of Killing fields becomes
\begin{equation}
		g(V, W) = \lambda\iota_Wh = -\frac{W^\psi r}{|K_0W^\psi|} +\frac{2r^2}{W^\psi K_0^2}\left(h_{1,\psi}+\zeta^i\hat x^i\right) + \ord(r^3)\;.
\end{equation}
On the other hand, from (\ref{eq_nullmetric}) $g(V,W)=-\mathcal{G}^{-1}$, so $\mathcal{G}$ diverges as $\sim1/r$ and harmonicity on 
$\mathbb{R}^3$ implies that 
\begin{equation}
	h_{1,\psi}+ \zeta^i\hat x^i = C \label{eq_hzeta_C}
\end{equation}
for some constant $C$.

The two expressions for the Maxwell field (\ref{eq_Maxwell_null}) and (\ref{eq_MxwNH}) evaluated on the horizon using (\ref{eq_Xinull}), 
(\ref{eq_XiNH}), and (\ref{eq_xNH}) yield respectively
\begin{align}
	F|_{\lambda=0} &= -\frac{\operatorname{sgn}(W^\psi K_0)(\gamma_{-1}\kappa_0-\kappa_{-1}\gamma_0)}{2\sqrt{3}\gamma_{-1}^2}\td \lambda\wedge \td\psi + \frac{\sqrt{3}}{4}\gamma_{-1}\operatorname{sgn}(W^\psi K_0) \epsilon_{ijk}\hat x^i\td \hat x^j\wedge \td \hat x^k \;,\label{eq_MXWcompare2a}\\
	F|_{\lambda =0} &= \frac{K_0}{\sqrt{3}}\epsilon_{ijk}\hat x^i\langle \td \hat x^j,\td\zeta^k\rangle\td \lambda\wedge \td\psi + \frac{\sqrt{3}}{4} K_0 \epsilon_{ijk}\hat x^i\td \hat x^j\wedge \td \hat x^k +\nonumber\\
					& \quad + \frac{K_0}{\sqrt{3}}\td\lambda \wedge \epsilon_{ijk}\left(2\hat x^j\zeta^k \td\hat x^i + \langle\td \hat x^j,\td\zeta^k\rangle\td \hat x^i\right)\;,\label{eq_MXWcompare2}
\end{align}
where $\mathcal{G}=:\gamma_{-1}/r+ \gamma_0 + \ord(r)$, $\mathcal{K}=:\kappa_{-1}/r+ \kappa_0 + \ord(r)$ with constants 
$\gamma_{-1}, \gamma_0, \kappa_{-1}, \kappa_0$, where we used the boundedness of $\mathcal{K}/\mathcal{G}$ and the harmonicity of 
$\mathcal{K}$. Thus, by comparison one obtains
\begin{align}
	K_0\epsilon_{ijk}\hat x^i \langle\td \hat x^j,\td \zeta^k\rangle &= -\frac{1}{2}\operatorname{sgn}(W^\psi K_0)\gamma_{-1}^{-2}(\gamma_{-1}\kappa_0-\kappa_{-1}\gamma_0)\;, \label{eq_zdx1}\\
	2\zeta^i\td\hat x^i &= \hat x^i\td\zeta^i\;, \label{eq_zdx2}
\end{align}
where for (\ref{eq_zdx2}) we have applied $\star_2$ on the one-form in the second line of (\ref{eq_MXWcompare2}). Taking $\star_2\td$ of 
(\ref{eq_zdx2}) and substituting into (\ref{eq_zdx1}) yields 
\begin{equation}
	\gamma_{-1}\kappa_0-\kappa_{-1}\gamma_0=0\;,\label{eq_c1NH}
\end{equation}
and therefore the first term of (\ref{eq_MXWcompare2a}) vanishes, which, to leading order, corresponds to the first term of (\ref{eq_MxwNH}). 
Thus, the Maxwell field must be of the form
\begin{align}
	F &=  -\frac{\sqrt{3}}{2}\star_\gamma(h+\lambda\partial_\lambda h) +\ord(\lambda)\td\lambda +\ord(\lambda^2)\td v+\ord(\lambda^2)\td z^i \nonumber\\
	&=\frac{\sqrt{3}K_0}{4}\epsilon_{ijk} \hat x^i \td \hat x^j\wedge \td \hat x^k-\frac{\sqrt{3}K_0\lambda}{2}\left[(2\zeta^i\hat x^i - \langle\td \zeta^i,\td \hat x^i\rangle)\epsilon_{klm}\hat x^k\td \hat x^l\wedge\td\hat x^m \right.\nonumber\\
	&\qquad\qquad\qquad\qquad\qquad\left.+ 2\td\psi \wedge\star_2(\zeta^i\td\hat x^i)\right]+\ord(\lambda)\td\lambda +\ord(\lambda^2)\td v+\ord(\lambda^2)\td z^i\;,\label{eq_Maxwell_leading}
\end{align}
where we used (\ref{eq_hzeta_C}). The $\td\lambda\wedge\td \psi \wedge \td \hat x$ terms of the Bianchi identity for (\ref{eq_Maxwell_leading}) imply 
\begin{equation}
	\zeta^i\td\hat x^i =0 \implies \zeta^i = \zeta\hat x^i,
\end{equation}
for some function $\zeta(\hat x)$, and by (\ref{eq_zdx2}) $\zeta = \zeta^i \hat x^i$ is a constant, and by (\ref{eq_hzeta_C}) so is $h_{1,\psi}$. 
\begin{corollary} \label{cor_Q}
	There exists a gauge choice of (\ref{eq_gauge_freedom}-\ref{eq_gauge_freedom_Qa}) such that around a horizon component $\mathcal{Q}= q_{-1}/r + q_0 + \ord (r)$ for some constants $q_{-1}, q_{0}$.
\end{corollary}
\begin{proof}
	The invariant $g(W,W) = -\mathcal{Q}/\mathcal{G} + \mathcal{G}^2|\bm b|^2$. By (\ref{eq_GK}) and using that $\mathcal{K}=\ord(r^{-1})$ 
	and $\mathcal{G}=\ord(r^{-1})$, one can see that (with an appropriate gauge choice of (\ref{eq_gauge_freedom}-\ref{eq_gauge_freedom_Qa})) we have $\mathcal{G}^2|\bm b|^2=\ord(r^2)$. 
	Thus, from harmonicity of $\mathcal{G}$ and Lemma \ref{lem_WNH} follows the claim.
\end{proof}

\pagebreak[1]
\section{Orbit space and the general global solution} \label{sec_orbit}

The next step is to determine the structure and topology of the orbit space. We will show that even if the DOC is not simply connected, 
it is still possible to define invariants (scalar functions) from certain closed 1-forms that are left invariant by $V,W$ and thus descend to the 
three-dimensional orbit space $\hat\Sigma:= \llangle\mathcal{M}\rrangle/[\mathbb{R}_V\times U(1)_W]$. Using these functions, one can 
`almost globally' define the harmonic functions, and prove that the solution is globally defined by a set of harmonic functions of multi-centred 
type, i.e. with isolated singularities where they diverge as $1/r$ on $\mathbb{R}^3$.

\subsection{Orbit space analysis and invariants}

Let us now look at the structure of the orbit space $\hat\Sigma$, which has been analysed in detail in~\cite{hollands_further_2011}. It 
is a topological manifold~\cite{fintushel_circle_1977, fintushel_classification_1978} with boundary $\partial\hat\Sigma=S_\infty^2\cup_i S_i^2\cup_j \hat H_j$, 
where $S^2_\infty$ is a 2-sphere at infinity, $S_i^2$ are non-isolated fixed points corresponding to `bolts'~\cite{gibbons_classification_1979}, 
and $\hat H_j$ are quotients of horizon components. Generally, the orbit space can be written as $\hat\Sigma = \hat F\cup \hat E\cup\hat L$, 
where $\hat F$, $\hat E$, $\hat L$ denotes fixed points ($U(1)$ isotropy), exceptional orbits  (discrete isotropy), and regular 
orbits (trivial isotropy), respectively. $\hat L$ is open in $\hat \Sigma$, and has the structure of a smooth manifold. 
$\hat E$ consists of curves in $\hat\Sigma$ that are either closed, end on an isolated fixed point, or on a horizon component $\hat H_i$. 
$\hat F$ consists of isolated fixed points or the previously mentioned 2-spheres $S_i^2$.

A crucial observation is that the orbit space of the asymptotic region $\hat\Sigma_0:=\Sigma_0/U(1)_W\simeq \mathbb{R}^3\setminus B^3$ 
is simply connected, so we can use topological censorship for Kaluza-Klein asymptotics (Theorem 5.5 of \cite{chrusciel_topological_2009}) 
to conclude that $\hat\Sigma=\Sigma/U(1)_W$ is simply connected. This, however, is not sufficient to define potentials for closed 1-forms, 
as the orbit space, in general, fails to be a {\it smooth} manifold at fixed points and exceptional orbits. To rule the existence of the latter
out, we have the following lemma.

\begin{lemma}
	There are no exceptional orbits in $\llangle\mathcal{M}\rrangle$. \label{lem_exceptional}
\end{lemma}
\begin{proof}
	Let us assume for contradiction that there exists an exceptional orbit $E_e\subset \llangle\mathcal{M}\rrangle$ with isotropy group $\mathbb{Z}_p$, 
	and let $e\in E_e$ be a point on that orbit, i.e. $e^{i2\pi/p}\cdot e =e$, where $e^{i\alpha}\in U(1)$ with $\alpha\sim\alpha+2\pi$. First, we will 
	construct a local chart around $E_e$ as follows.
	
	By assumption $S_e:=\operatorname{span}\{W_e, V_e\}\subset T_e\mathcal{M}$ is timelike and non-degenerate, hence $S_e^\perp$ is spacelike. Since 
	$\lie_WV=\lie_WW=0$ and $W$ is Killing, $e^{i2\pi/p}{}_*$ is an isometry that preserves $S_e$ and $S^\perp_e$, so $\mathcal{R}:=e^{i2\pi/p}{}_*|_{S^\perp_e}$ is a three-dimensional orientation-preserving\footnote{
		This can be seen by the fact that $W$ preserves $V^\flat\wedge W^\flat\wedge \iota_WX^{(1)}\wedge \iota_WX^{(2)}\wedge \iota_WX^{(3)}\neq0$ 
		for $N>0$.
	} rotation. Furthermore, $\mathcal{R}^p=\operatorname{Id}_{S^\perp_e}$, thus
	$\mathcal{R}$ must be a rotation by $2\pi n/p$ for some $n\in\mathbb{Z}$. Let $Z_e\in S^\perp_e$ be the unit vector preserved by 
	$\mathcal{R}$, and $X_e, Y_e\in S^\perp_e$ such that $\{Z_e, X_e, Y_e\}$ is an orthonormal frame in $S^\perp_e$ and 
	\begin{equation}
		\mathcal{R}\begin{pmatrix}
			X_e\\ Y_e
		\end{pmatrix}=
		\begin{pmatrix}
			\cos(2\pi n/p) & \sin(2\pi n/p)\\
			-\sin(2\pi n/p) & \cos(2\pi n/p)
		\end{pmatrix}
		\begin{pmatrix}
			X_e\\ Y_e
		\end{pmatrix}\;. \label{eq_rotate_XY}
	\end{equation}
	Now let us extend ${Z_e, X_e, Y_e}$ along $E_e$ by 
	\begin{align}
		\lie_W Z=0\;, \qquad \lie_W X = -\frac{n}{2} Y\;, \qquad \lie_W Y = \frac{n}{2} X\;, \label{eq_extend_XY}
	\end{align}
	(The factor of $2$ only appears because $W$ is normalised such that it is $4\pi$-periodic.) Simply Lie-dragging $X,Y$ along $W$ would rotate them 
	according to (\ref{eq_rotate_XY}), which is cancelled by the right-hand side of (\ref{eq_extend_XY}) (recall that the exceptional orbit is 
	$2\pi/p$-periodic), so extending $X$ and $Y$ as in (\ref{eq_extend_XY})
	is necessary in order for them to be single-valued along $E_e$. Finally, let us extend $X, Y, Z$ along the integral curves of $V$ to some 2-surface $E$ by 
	\begin{equation}
		\lie_VX=\lie_VY=\lie_VZ=0.
	\end{equation}
	This is possible since $[V,W]=0$. Thus, we have constructed an orthonormal frame along $E$ in the normal 
	bundle $NE$. Let $\lambda$ parametrise the integral curve of $2p^{-1}W$ along $E$, so 
	that $\lambda\sim \lambda + 2\pi$ on $E$, and $\tau$ be the affine parameter distance from $E_e$ along an integral curve of $V$. By the 
	Tubular Neighbourhood Theorem (see e.g.~\cite{lee_introduction_2018}) we can introduce 
	coordinates $\{x, y, z\}$ by exponentiating linear combinations of $\{X, Y, Z\}$ such that $E$ is at $x=y=z=0$, which together with $\lambda, \tau$ 
	form a local chart in some neighbourhood $U$ of $E_e$. Since $W, V$ are Killing, they map geodesics to geodesics, hence $V=\partial_\tau$, 
	and components of $W$ can be deduced from its action on the orthonormal frame in $NE$, that is (cf. (2.20) of~\cite{hollands_further_2011})
	\begin{equation}
		W=\frac{p}{2}\partial_\lambda+\frac{n}{2}(x\partial_y-y\partial_x).
	\end{equation} 
	
	Now we look at $X^{(i)}$ on $E$. From (\ref{eq_VX}) it follows that $X^{(i)}$ have no $\tau$ leg. The two-forms are preserved by $W$, thus 
	\begin{align}
		0=\lie_{W} X^{(i)}_{a x}|_E &= \frac{p}{2}\partial_\lambda X^{(i)}_{a x} + \frac{n}{2} X^{(i)}_{a y}\;,\\
		0=\lie_{W} X^{(i)}_{a y}|_E &= \frac{p}{2}\partial_\lambda X^{(i)}_{a y} - \frac{n}{2} X^{(i)}_{a x}\;, 
	\end{align}
	for $a=\lambda, z$. Taking a $\lambda$ derivative of these equations, and substituting 
	the originals back yields that 
	\begin{equation}
		\partial_\lambda^2X^{(i)}_{a x} = -\frac{n^2}{p^2}X^{(i)}_{a x}, \label{eq_Xlx}
	\end{equation}
	and similarly for $X^{(i)}_{ay}$. Equation (\ref{eq_Xlx}) admits solutions which are $2\pi p/n$-periodic in $\lambda$. 
	The exceptional orbits, however, are $2\pi$-periodic in $\lambda$, which means $n/p\in \mathbb{Z}$. 
	In that case the neighbouring orbits (for $x,y>0$) would have the same isotropy as the exceptional one, which cannot happen.
	Hence, the only solution to (\ref{eq_Xlx}) is $X^{(i)}_{a x}=X^{(i)}_{a y}=0$ for $a=\lambda, z$.
	Thus, at each point of the exceptional orbit $X^{(i)}\in \operatorname{span}\{\td \lambda\wedge \td z, \td x\wedge \td y\}$, which is two-dimensional, contradicting 
	linear independence of the three two-forms following from (\ref{eq_XXspacetime}).
\end{proof}

\noindent{\bf Remark.} The key assumptions for Lemma \ref{lem_exceptional} is that the axial Killing field preserves three linearly 
independent 2-forms possibly having legs in four directions (recall that they have no $\td v$ legs). This follows from supersymmetry in both timelike
and null case, but it is a more general result.
\\

From the assumption that the span of Killing fields is timelike (assumption \ref{ass5D_timelike}), it follows that $F\subset\operatorname{int}\mathcal{B}$, so $F$ is a set of 
fixed points of a triholomorphic Killing field, which must be isolated (see~\cite{gibbons_hidden_1988} and the proof of Lemma 8 
in~\cite{katona_supersymmetric_2023}). This rules out any `bolts'~\cite{gibbons_classification_1979} in the orbit space. Thus, we have the 
following result (cf. Lemma 8 of~\cite{katona_supersymmetric_2023}).

\begin{corollary}
	The orbit space $\hat\Sigma = \hat L\cup\hat F$, where $\hat L$ corresponds to regular orbits, and $\hat F$ is a finite  
	set of points corresponding to isolated fixed points of the $U(1)$ action. Its boundary is $\partial\hat \Sigma= \hat H \cup S_\infty^2$. \label{cor_orbitspace}
\end{corollary}
\begin{corollary}
	$\hat L$ is simply connected. \label{cor_Lsc}
\end{corollary}
\begin{proof}
	Corollary \ref{cor_orbitspace} is an immediate consequence of Lemma \ref{lem_exceptional}. Corollary \ref{cor_Lsc} follows from the 
	fact that a simply connected 3-manifold stays simply connected after the removal of finite many points.
\end{proof}

\begin{lemma}
	The equations 
	\begin{align}
		\iota_WX^{(i)}=\td x^i, \qquad \iota_WF=\frac{\sqrt{3}}{2}\td \Psi \label{eq_xdef}
	\end{align}
	define smooth functions $x^i, \Psi$ globally on $\llangle \mathcal{M}\rrangle\cup \mathcal{H}$ (up to an additive constant). \label{lem_scalarinvariants}
\end{lemma}
\begin{proof}
	From their definition, we have $\iota_W\td x^i=\iota_W\iota_W X^{(i)}=0$, $\lie_W\td x^i =\td(\iota_W \td x^i)= 0$, 
	$\iota_V\td x^i = \iota_V\iota_WX^{(i)}=0$ by (\ref{eq_VX}), $\lie_V\td x^i =\td(\iota_V \td x^i)= 0$ and similarly for 
	$\td\Psi$ using (\ref{eq_VF}), so they descend as smooth 1-forms on $\hat L$. Since $\hat L$ is a simply connected smooth manifold by Corollary \ref{cor_Lsc}, 
	they define smooth functions $x^i$ and $\Psi$ on $\hat L$. In fact, since $W|_\mathcal{H}\neq0$ and the orbits of connected horizon components 
	$\hat{\mathcal{H}_i} = \mathcal{H}_i/(\mathbb{R}\times U(1))\simeq S^2$ are simply connected~\cite{hollands_further_2011}, the functions 
	can be defined on the horizon as well. Then we can uplift these functions to $(\llangle \mathcal{M}\rrangle\cup \mathcal{H}\setminus\mathcal{F})/\mathbb{R}$ by 
	$\lie_Wx^i=\lie_W\Psi=0$. We then extend the functions to the isolated fixed points continuously. This is possible, since the functions 
	can be integrated on some small 4-ball around a fixed point with a constant chosen such that the function agrees with the one on 
	$(\llangle \mathcal{M}\rrangle\cup \mathcal{H}\setminus\mathcal{F})/\mathbb{R}$. Finally, we uplift $x^i$ and $\Psi$ to the spacetime 
	by $\lie_Vx^i=\lie_V\Psi=0$.

\end{proof}

\subsection{General global solution}

We are now ready to deduce the global form of the general solution that satisfies our assumptions, heavily relying on 
results of~\cite{reall_higher_2004,breunholder_moduli_2019,katona_supersymmetric_2023}. 
\begin{theorem}
	For any solution $(\mathcal{M}, g, F)$ of $D=5$ minimal supergravity satisfying assumptions \ref{ass5D_susy}-\ref{ass5D_smooth}, the DOC 
	is globally determined by a set of harmonic functions 
	($H,K,L,M$ when $f\not\equiv 0$, $\mathcal{G},\mathcal{Q}_0,\mathcal{K}$ when $f\equiv0$) on $\mathbb{R}^3\setminus\cup_{i=1}^N \{\bm a_i\}$ 
	of the form
	\begin{equation}
		H(\bm x) = h +\sum_{i=1}^N \frac{h_i}{|\bm x-\bm a_i|}\;, \label{eq_multi-centred}
	\end{equation}
	where $h, h_i$ are constants and $\bm a_i$ are the positions of centres corresponding to connected components of the horizon or 
	fixed points of the axial symmetry (the latter only possible when $f\not\equiv 0$).

	In the timelike case ($f\not\equiv 0$) the metric is of Gibbons-Hawking form 
	\begin{equation}
		g = -f^2(\td t +\omega_\psi(\td\psi +\chi)+\hat\omega)^2 + \frac{1}{Hf}(\td\psi +\chi)^2 + \frac{H}{f}\td x^i \td x^i\;, \label{eq_gGH}
	\end{equation}
	determined by (\ref{eq_omegapsi}-\ref{eq_f}), and the Maxwell field is determined by (\ref{eq_Maxwell}). In particular, (\ref{eq_gGH}) and 
	(\ref{eq_Maxwell}) smoothly extends in coordinates $(t,\psi\, x^i)$ to regions where $f=0$.

	In the null case ($f\equiv 0$) the solution 
	is determined by (\ref{eq_nullmetric}-\ref{eq_Maxwell_null}), where for the inhomogeneous part $\mathcal{Q}_I:=\mathcal{Q}-\mathcal{Q}_0$, which 
	satisfies (\ref{eq_Fcalc}), we impose the boundary conditions that $\mathcal{Q}_I$ is bounded on $\mathbb{R}^3$ and vanishes as $|\bm x|\to \infty$.
\label{thm_necessary}
\end{theorem}

\begin{proof}
	The functions $x^i$ are globally defined on $\llangle\mathcal{M}\rrangle \cup \mathcal{H}$ by Lemma \ref{lem_scalarinvariants}. By Lemma 
	\ref{lem_timelike_null} $V$ is either globally null, or $V$ is timelike on some dense subset of $\llangle\mathcal{M}\rrangle$. The 
	near-horizon analysis of Section \ref{sec_NHnull} for the null case, and that of~\cite{reall_higher_2004, katona_supersymmetric_2023} 
	for the timelike case implies that each connected component of the horizon is mapped to a single point of $\mathbb{R}^3$ by 
	$\bm x$ (Lemma 6 of~\cite{katona_supersymmetric_2023}). Furthermore, there are no exceptional orbits of $W$ (Lemma \ref{lem_exceptional}). 
	Thus, we can apply Lemma 9 of ~\cite{katona_supersymmetric_2023} to deduce that $\bm x: \hat L\to \mathbb{R}^3\setminus \bm x(\hat{\mathcal{H}}\cup \hat F)$ 
	is a global diffeomorphism. Hence, the three-dimensional orbit space with each horizon component added as a single point, 
	$\hat\Sigma\cup_i\{\mathcal{H}_i\}$ is in bijection with $\mathbb{R}^3$. We next analyse the null and timelike cases in turn.
	\\

	\noindent \underline{\bf Null case ($f\equiv 0$):} We have seen that the solution must have the form 
	(\ref{eq_nullmetric}-\ref{eq_Maxwell_null}). $\bm a_i$ must correspond to horizon components, as fixed points are ruled out (see discussion 
	after Lemma \ref{lem_Wnull}). $\mathcal{G}$ is globally defined and non-zero on $\llangle\mathcal{M}\rrangle$ by (\ref{eq_Gdef}). This 
	means that $\mathcal{G}$ is a non-zero harmonic function on $\mathbb{R}^3\setminus \cup_i\{\bm a_i\}$, which diverges at each $\bm a_i$ 
	as $\sim1/|\bm x-\bm a_i|$. This follows from the near-horizon analysis (see proof of Lemma \ref{lem_WNH}), or by B\^{o}cher's Theorem 
	(see e.g. \cite{axler_harmonic_2001}). This determines the singular structure of $\mathcal{G}$, and its regular part is a harmonic 
	function on $\mathbb{R}^3$ which approaches $\mathcal{G}\to \tilde L^{-1}$ by (\ref{eq_metric_AKK}), (\ref{eq_Gdef}) and 
	(\ref{eq_Vnull}), therefore it is constant.

	From (\ref{eq_magnetic_pot_null}) and (\ref{eq_xdef}) follows that (up to an additive constant) 
	\begin{equation}
		\mathcal{K} = 3\mathcal{G}\Psi\;, \label{eq_Kdef}
	\end{equation}
	which defines $\mathcal{K}$ on $\llangle\mathcal{M}\rrangle$. It also follows that $\mathcal{K}$ has (at most) simple poles at horizon 
	components. Since $\td \Psi\sim F\sim \ord(|\bm x|^{-\tau/2-1})$ (see Remark after Definition \ref{def_KK}), $\Psi\to\const$ at infinity, which 
	we are free to choose to be zero. (This corresponds to the `gauge' freedom $\mathcal{K}\to \mathcal{K}+ c \mathcal{G}$ for some constant 
	$c$.) Thus, the regular part of $\mathcal{K}$ on $\mathbb{R}^3$ is constant (in fact zero in this `gauge').

	Finally, let us look at $\mathcal{Q}$. The norm of the axial Killing field $g(W,W) = -\frac{\mathcal{Q}}{\mathcal{G}}+ \mathcal{G}^2|\bm b|^2$.
	From the asymptotic behaviour of $\mathcal{G}$ and $\mathcal{K}$, by (\ref{eq_GK}) we deduce that $|\bm b|=\ord(|\bm x|^{-1})$, and thus 
	$\mathcal{Q}=-\tilde L+\ord(|\bm x|^{-1})$. By Corollary \ref{cor_Q}, near a horizon component there exists a gauge such that 
	\begin{equation}
		\mathcal{Q}= \frac{q^i_{-1}}{|\bm x-\bm a_i|} + q^i_0 + \ord(|\bm x-\bm a_i|),
	\end{equation}
	with constants $q^i_{-1}$ and $q^i_0$. It follows that we can impose the following boundary conditions on the inhomogeneous part 
	$\mathcal{Q}_I$: we require that it vanishes at infinity, and it is bounded at each horizon component. This fixes the inhomogeneous 
	part\footnote{This can be seen by the fact that the difference of any two such particular solutions is a bounded harmonic function on 
	$\mathbb{R}^3\setminus\cup_i\{\bm a_i\}$ that goes to zero at infinity, and thus identically zero.}, and thus the homogeneous part 
	$\mathcal{Q}_0=q^i_{-1}/|\bm x-\bm a_i|+\ord(1)$ at the $i^{th}$ centre and approaches $-\tilde L$ at infinity.
	By the same arguments as for $\mathcal{G}$, $\mathcal{Q}_0$ has the claimed form.
	\\

	\noindent \underline{\bf Timelike case ($f\not\equiv 0$):}
	We have already established that the base has Gibbons-Hawking form on $\widetilde{\mathcal{M}}\setminus\mathcal{F}$. Now let us focus on 
	the associated harmonic functions. Since $N>0$ on $\llangle\mathcal{M}\rrangle\setminus\mathcal{F}$, by Lemma 1 
	of~\cite{breunholder_moduli_2019} the metric (\ref{eq_gGH}) is smooth, invertible with smooth inverse, the Maxwell field given by (\ref{eq_Maxwell}) is 
	smooth, and the harmonic functions are smooth, and they are invariantly defined by
	\begin{align}
		H  = \frac{f}{N} \; , \qquad \qquad &\qquad \qquad L=\frac{fg(W, W)+ 2g(V, W)\Psi-f\Psi^2}{N} \; , \label{eq_harmonicinvariant_1}\\
		K = \frac{f\Psi - g(V, W)}{N} \; , \quad   &M = \frac{g(W,W)g(V, W)-3f\Psi g(W,W)-3\Psi^2g(V, W)+f\Psi^3}{2N} \; , \label{eq_harmonicinvariant}
	\end{align}
	even on the set $f=0$, where $N$ must be non-zero by our assumption that the Killing fields span a timelike vector space.
	Here we used that $\Psi$ is globally defined due to Lemma \ref{lem_scalarinvariants} up to an additive constant, changing of which 
	corresponds to shifting the harmonic functions as in (\ref{eq_harmonic_gauge}).

	Now that we have established that the harmonic functions are smooth at generic points of the spacetime, let us look at their behaviour 
	at fixed points of $W$ and the horizon. For fixed points, note that the expression for $H$ in (\ref{eq_harmonicinvariant}) implies 
	that the zeros of $H$ and $f$ must coincide on $\llangle\mathcal{M}\rrangle$. Also, $f(p)\neq0$ at a fixed point $p$ and (by continuity) 
	on some neighbourhood, thus $H$ is non-zero in some neighbourhood of $p$. As $N(p)=0$, by  (\ref{eq_harmonicinvariant}) $H$ must diverge 
	at $p$ and by B\^ocher's theorem $H$ must have a simple pole at $\bm x(p) = \bm a_i$ for some $i$. From (\ref{eq_harmonicinvariant}) it 
	follows that all harmonic functions have the same type of singularity at $\bm a_i$. For the horizon components, using the near-horizon 
	analysis of~\cite{reall_higher_2004}, triholomorphicity implies that the harmonic functions have (at most) simple poles at the horizon 
	components (Lemma 9 of~\cite{katona_supersymmetric_2023}).

	This fixes the form of the harmonic functions up to a globally defined harmonic function on $\mathbb{R}^3$, which, as in the 
	null case, is determined by the asymptotic conditions. By calculating the invariants $f, \Psi, g(V, W), g(W,W)$ at infinity and using 
	(\ref{eq_harmonicinvariant}) we find that the regular part of the harmonic functions must be a constant. Depending on whether 
	$f\to1$ or $f\to 0$ at infinity (asymptotically timelike and null cases respectively), the values of these constants are different.

	In the asymptotically timelike case by Lemma \ref{lem_AKKGH} $f^2 = 1+\ord(|\bm x|^{-\tau})$, $g(W, W) = \tilde L^2 + \ord(|\bm x|^{-\tau})$, 
	and $g(V,W) = -\tilde L\gamma v_H+\ord(|\bm x|^{-\tau})$, thus $N = \gamma^2\tilde L^2 + \ord(|\bm x|^{-\tau})$. Without loss of generality 
	we can require that $f\to 1$. Then 
	using (\ref{eq_harmonicinvariant}) we see that $H=(\tilde{L}\gamma)^{-2}+\ord(|\bm x|^{-1})$, where the fall-off of the subleading terms are determined 
	by harmonicity on $\mathbb{R}^3$. The Maxwell field (\ref{eq_Ftimelike}) falls off as $F= \ord(|\bm x|^{-1-\tau/2})$, hence by (\ref{eq_xdef})
	$\Psi = \Psi_0 + \ord(|\bm x|^{-\tau/2})$, where $\Psi_0$ is a constant of integration, which we are free to set $\Psi_0 = -\tilde L\gamma v_H$ 
	for convenience\footnote{We choose the constant such that the regular part of $K$ vanishes at infinity. This is the same gauge as 
	the one used in \cite{elvang_supersymmetric_2005} and \cite{tomizawa_kaluza-klein_2018}}. This implies through  
	(\ref{eq_harmonicinvariant}) that near spatial infinity 
	\begin{align}
		H = (\tilde{L}\gamma)^{-2}+\ord(|\bm x|^{-1})\;, \quad L = 1+ \ord(|\bm x|^{-1}), \quad K = \ord(|\bm x|^{-1}), \quad M =\tilde L\gamma v_H+\ord(|\bm x|^{-1}) \;,\label{eq_harmonic_as_TL}
	\end{align}
	where the subleading fall-off has again been determined by harmonicity.

	In the asymptotically null case we obtain by a similar calculation that
	\begin{align}
		H = \ord(|\bm x|^{-1})\;, \quad L = \ord(|\bm x|^{-1})\;,\quad K = \tilde{L}^{-1}+\ord(|\bm x|^{-1})\;, \quad M = -\tilde{L}/2 + \ord(|\bm x|^{-1})\;,\label{eq_harmonic_as_Null}
	\end{align}
	where for simplicity we set $\Psi_0=0$.
\end{proof}

\noindent {\bf Remark.} Note that for all cases harmonicity on $\mathbb{R}^3$ sets $\tau=1$ in Definition \ref{def_KK} for the fall-off at infinity.

\section{Regularity and asymptotic conditions} \label{sec_sufficient}

Theorem \ref{thm_necessary} necessarily includes the most general global solution under the stated assumptions, however it does not guarantee that all such solutions 
are smooth, asymptotically Kaluza-Klein black hole solutions. In order to establish the sufficient criteria for this, we need to check 
the asymptotics and regularity of these solutions. This analysis is different in the timelike ($f\not\equiv0$) and globally null ($f\equiv0$) 
case, so we consider them in turn.

\subsection{Timelike case}

By Theorem \ref{thm_necessary} the solution is globally determined by harmonic functions 
\begin{align}
	H(\bm x) = \sum_{i=1}^N \frac{h_i}{|\bm x-\bm a_i|} + h\;, \qquad K(\bm x) = \sum_{i=1}^N \frac{k_i}{|\bm x-\bm a_i|} + k\;, \nonumber\\
	L(\bm x) = \sum_{i=1}^N \frac{l_i}{|\bm x-\bm a_i|} + l\;, \qquad M(\bm x) = \sum_{i=1}^N \frac{m_i}{|\bm x-\bm a_i|} + m\;, \label{eq_harmonicansatz}
\end{align}
and the asymptotic values of the harmonic functions depend on whether $f\to1$ or $f\to0$ as $|\bm x|\to\infty$, as in 
(\ref{eq_harmonic_as_TL}-\ref{eq_harmonic_as_Null}), that is respectively
\begin{align}
	h = (\tilde L\gamma)^{-2}\;, \qquad k = 0\;, &\qquad l = 1\;, \qquad m = \tilde L\gamma v_H\;, \label{eq_harmoniccosntants_AT}\\
	&\text{ and }\nonumber\\
	h = 0\;, \qquad k =\tilde{L}^{-1} \;, &\qquad l = 0\;, \qquad m = -\tilde{L}/2\;. \label{eq_harmoniccosntants_AN} 
\end{align}
First we analyse the asymptotic geometry, then we determine the regularity conditions near the horizon and around fixed points.

\subsubsection{Geometry at spatial infinity}\label{ssec_asymptotics}

At spatial infinity ($|\bm x|\to\infty$), we can expand the harmonic functions as 
\begin{equation}
	H(\bm x) = h + \frac{\sum_{i=1}^N h_i}{|\bm x|} + \ord(|\bm x|^{-2})\;, 
\end{equation}
and similarly for $K,L,M$.

As we will shortly see, in the timelike case the geometry of the `sphere' at infinity is governed by 
the second-to-leading order terms in $H$. After integration, 
(\ref{eq_chieq}) yields 
\begin{equation}
	\chi = \tilde \chi_0 \td \phi + \tilde h_0\cos\theta\td \phi + \ord(r^{-2})\td x^i\;,
\end{equation}
where we defined $\tilde h_0 :=\sum_{i=1}^Nh_i$, on $\mathbb{R}^3$ we use standard spherical coordinates ($r, \theta, \phi$), and $\tilde\chi_0$ 
is a constant of integration. We can set the latter to an arbitrary value by a coordinate change $(\psi,\phi)\to(\psi+c\phi,\phi)$ which shifts 
$\tilde\chi_0\to\tilde\chi_0-c$. We will shortly see that the bundle structure requires that $\tilde h_0$ is an 
integer, so for convenience we work in a gauge in which $\tilde\chi_0\equiv \tilde h_0\mod 2$.

We now look at the geometry of a constant time hypersurface in the limit $r\to\infty$, in which the spatial metric has the form 
\begin{align}
	g\big|_{t=\const}= \begin{cases}
		\frac{\td r^2}{\gamma^2\tilde L^2}+\tilde h_0^2\tilde L^2\left(\frac{\td \psi+\tilde\chi_0\td \phi}{\tilde h_0}+\cos\theta \td \phi\right)^2 + \frac{r^2}{\gamma^2\tilde L^2}(\td\theta^2 + \sin^2\theta\td\phi^2) +\ldots,\, &\text{ if } \tilde h_0\neq0\;, \\
		\frac{\td r^2}{\gamma^2\tilde L^2}+\tilde L^2\left(\td \psi+\tilde\chi_0\td \phi\right)^2 + \frac{r^2}{\gamma^2\tilde L^2}(\td\theta^2 + \sin^2\theta\td\phi^2)+\ldots,\; &\text{ if } \tilde h_0=0 \;,
	\end{cases}
\end{align}
with $\gamma=1$ in the asymptotically null case, and $\ldots$ represent lower order terms in each metric component. It is explicit that at 
constant $r$ the metric takes the local form that of a squashed $S^3$ if $\tilde h_0\neq0$, or $S^2\times S^1$ if $\tilde h_0 =0$. In the 
latter case the angles are identified as $(\psi,\phi)\sim(\psi+4\pi,\phi)\sim(\psi,\phi+2\pi)$.

In the locally spherical case, let us define $\phi^\pm : = (\psi +(\tilde \chi_0 \pm \tilde h_0)\phi)/\tilde h_0$ so that the leading order angular metric 
(on the `sphere' at infinity) takes the form 
\begin{equation}
	g\big|_{t, r} = \tilde h_0^2\tilde L^2\left(\td\phi^\pm \mp (1\mp\cos\theta)\td\phi\right)^2 + r^2\gamma^{-2}\tilde L^{-2}(\td\theta^2 + \sin^2\theta\td\phi^2) + \ldots \;.
\end{equation}
In these coordinates the $U(1)$ connection $\mp (1\mp\cos\theta)\td\phi$ is regular on the northern and southern hemisphere, respectively, so we use 
$\phi^\pm$ as vertical coordinates on the fibres on the N/S hemisphere. Independent $4\pi$-periodicity of $\psi$ for each $\phi, \theta$ 
(which we assume by Definition \ref{def_KK}) implies that $\phi^\pm$ are $4\pi/\tilde h_0$-periodic. Therefore, we may parametrise $U(1)$ 
(now identified with the complex unit circle) as $\exp(i \tilde h_0 \phi^\pm/2)$. Also, $\phi^+ = \phi^- +2\phi$, so the transition function 
between the N/S hemisphere is $\exp(i\tilde h_0 \phi)$, which is single-valued on the equator only if $\tilde h_0$ is an integer as previously 
stated. Also, independent periodicity of $\phi\sim\phi+2\pi$ for fixed $\phi^\pm$ implies
\begin{equation}
	(\psi, \phi)\sim(\psi + 2\pi(\tilde \chi_0 \mp\tilde h_0), \phi + 2\pi).
\end{equation}
Since we have set $\tilde \chi_0 \mp\tilde h_0$ to be even, and $\psi$ is independently $4\pi$-periodic, this implies that 
\begin{equation}
	(\psi, \phi)\sim(\psi , \phi + 2\pi).
\end{equation}
One can check that the geometry of the `sphere' at infinity is $L(|\tilde h_0|, 1)$ or $S^3$ if $|\tilde h_0|=1$.

In summary, the geometry at infinity is $S^2\times S^1$ for $\tilde h_0 =0$, $S^3$ for $\tilde h_0 =\pm 1$, and $L(p, 1)$ for 
$|\tilde h_0|=p\in \mathbb{Z}$. In all cases, the angles are identified as $(\psi,\phi)\sim(\psi+4\pi,\phi)\sim(\psi,\phi+2\pi)$.

\subsubsection{Regularity}\label{ssec_regularity}

Regularity has to be established at \begin{enumerate*}[label=(\roman*)]\item generic points, \item the horizon, \item at fixed points of $W$.
\end{enumerate*} Recall that assumption \ref{ass5D_timelike} implies that $N>0$ on $\llangle\mathcal{M}\rrangle\setminus\mathcal{F}$.
For generic points, smoothness of the solution is established by Lemma 1 of~\cite{breunholder_moduli_2019}, which states that if 
\begin{equation}
	N^{-1}=K^2 +HL >0 \label{eq_condN}
\end{equation} and $H,K,L,M$ are smooth, then $(g, F)$ is smooth, and $g$ is invertible with smooth inverse. Therefore, (\ref{eq_condN}) must 
be imposed on the harmonic functions everywhere on their domain. Much like in the asymptotically flat case, it is currently an open problem 
to reformulate (\ref{eq_condN}) as an explicit condition on the parameters $h_i,k_i, l_i, \bm a_i$.

The smoothness of a solution determined by harmonic functions of the form (\ref{eq_harmonicansatz}) has been analysed for the  
asymptotically flat case, and the sufficient conditions have been determined at the horizon and fixed points 
in~\cite{breunholder_moduli_2019} for axisymmetric, and in~\cite{katona_supersymmetric_2023} for general (non-symmetric) harmonic functions.
Ref.~\cite{katona_supersymmetric_2023} \begin{enumerate*}[label=(\alph*)]\item\label{assumption_harmonic} uses a {\it local} expansion of harmonic functions in terms of 
spherical harmonics, \item\label{assumption_id} assumes that the periodicity of the angular coordinates are determined by asymptotic conditions, 
and they are $\psi\sim\psi+4\pi$ and $\phi\sim\phi+2\pi$ independently, where $\phi$ is the azimuthal angle in a spherical coordinate system of the 
$\mathbb{R}^3$ base of (\ref{eq_gGH}).\end{enumerate*} For \ref{assumption_harmonic}, around each centre (which we take to be the origin) 
we expand the harmonic functions as 
\begin{equation}
	H = \frac{h_{-1}}{r} + h_0 + \sum_{\substack{k=1\\|m|\le k}}^\infty h_{km}r^k Y_k^m(\theta, \phi)\;,\label{eq_harmonic_expansion}
\end{equation}
and similarly for the other harmonic functions. This has the same form irrespective of asymptotics, thus 
the same sufficient conditions hold for asymptotically Kaluza-Klein spacetimes. For \ref{assumption_id} we have seen in Section 
\ref{ssec_asymptotics}, that the angle coordinates $\psi, \phi$ admit the same identifications as in the asymptotically flat case. 
This means that the regularity analysis is identical. We here only present the results, details can be found in~\cite{katona_supersymmetric_2023}.

If the centre corresponds to a horizon component, existence of a coordinate transformation to Gaussian 
null coordinates (and also positivity of the horizon area) requires that 
\begin{equation}
	-h^2_{-1}m^2_{-1}-3h_{-1}k_{-1}l_{-1}m_{-1}+h_{-1}l^3_{-1}-2k_{-1}^3m_{-1}+\frac{3}{4}k_{-1}^2l_{-1}^2>0 \; . \label{ineq_horizon}
\end{equation}
Integrating (\ref{eq_chieq}) and (\ref{eq_omegahat}) yields
\begin{align}
	\chi &= (\chi_0 +h_{-1}\cos\theta)\td\phi + \tilde\chi \;, \\
	\omega &= (\omega_0 +\omega_{-1} \cos \theta) \td \phi +\tilde\omega \;,
\end{align}
where $\chi_0$, $\omega_0$ are constants of integration, 
\begin{equation}
	\omega_{-1} := h_0m_{-1}-m_0h_{-1}+\tfrac{3}{2}(k_0l_{-1}-l_0k_{-1}) \;, 
\end{equation}
and $\tilde\chi$, $\tilde\omega$ contain higher order terms in $r$. Smoothness at the axes $\theta =0, \pi$ is equivalent to 
\begin{equation}
	\omega_{-1} = \omega_0=0\; . \label{eq_omegareq_hor}
\end{equation}
The horizon topology is determined by
\begin{equation}
	h_{-1}\in\mathbb{Z}\;. \label{eq_hor_h}
\end{equation} 
For $h_{-1}=0$ the topology is $S^2\times S^1$, 
for $h_{-1}=\pm 1$ it is $S^3$, otherwise $L(|h_{-1}|,1)$. Correct identification of the angles requires that 
\begin{equation}
	\chi_0\equiv h_{-1}\mod 2 \;.  \label{eq_hor_chicond}
\end{equation}

If the centre is a fixed point of $W$, then there is a curvature singularity at the centre unless
\begin{align}
	h_{-1}=\pm 1 \;.\label{eq_fixed_h}
\end{align}
The Killing fields having timelike span implies that $f\neq0$ at the centre, which is equivalent to
\begin{equation}
	l_{-1}+h_{-1}k_{-1}^2=0 \;,\label{eq_fcond1}
\end{equation}
and the spacetime has the correct signature around a fixed point if and only if 
\begin{equation}
	h_{-1}( l_0- h_0 k_{-1}^2+ 2 h_{-1} k_{-1} k_0) >0  \; .\label{eq_fcond2}
\end{equation}
Smoothness of $\omega$ is equivalent to 
\begin{align}
	m_{-1} &= \tfrac{1}{2} k_{-1}^3   \; ,  \label{eq_ompsicond1} \\
	\omega_{-1} &= \omega_0=0\; , \label{eq_omegareq_fix}
\end{align}
and the correct identification of the angle coordinates, similarly to the horizon, requires that 
\begin{equation}
	\chi_0\equiv 1\mod 2\;. \label{eq_fix_chicond}
\end{equation}

\subsection{Null case}\label{ssec_sufficient_null}

By Theorem \ref{thm_necessary} the solution is determined by the harmonic functions on $\mathbb{R}^3\setminus\cup_i\{\bm a_i\}$ 
\begin{align}
	\mathcal{G}(\bm x) = \frac{1}{\tilde L}+\sum_{i=1}^N \frac{g_i}{|\bm x-\bm a_i|}\;,\quad \mathcal{K}(\bm x) = \sum_{i=1}^N \frac{k_i}{|\bm x-\bm a_i|}\;,\quad\mathcal{Q}_0(\bm x)=-\tilde L +\sum_{i=1}^N \frac{f_i}{|\bm x-\bm a_i|}\;.
\end{align}

Our assumption that there exists a timelike linear combination of the Killing fields in the DOC (assumption \ref{ass5D_timelike}) excludes 
any fixed point of $W$, hence all centres must correspond to connected components of the horizon. We start our analysis at the horizon.

Let us focus on a single horizon component, around which we will use standard spherical coordinates $(r, \theta, \phi)$ on $\mathbb{R}^3$, 
so the horizon is at the origin. Then we can expand locally the harmonic functions and, using Corollary \ref{cor_Q}, $\mathcal{Q}$ as
\begin{align}
	\mathcal{G} = \frac{\gamma_{-1}}{r} + \gamma_0 + \tilde{\mathcal{G}}\;, \quad \mathcal{Q} = \frac{q_{-1}}{r}+ q_0 + \tilde{\mathcal{Q}}\;, \quad \mathcal{K} = \frac{\kappa_{-1}}{r}+ \kappa_0+ \tilde{\mathcal{K}}\;,\label{eq_harmonic_exp_null}
\end{align} 
with some constants $\gamma_{-1}, \gamma_0, q_{-1},q_0, \kappa_{-1},\kappa_0$. The quantities with tilde are of $\ord(r)$, and in the 
case of $\tilde{\mathcal{G}},\tilde{\mathcal{K}}$ harmonic\footnote{$\tilde{\mathcal{Q}}$ contains the inhomogeneous part of $\mathcal{Q}$, 
hence not harmonic.}. Recall that $\mathcal{G}$ is positive in the DOC and must diverge at the horizon, which means that 
\begin{equation}
	\gamma_{-1}>0.\label{eq_horizon_G}
\end{equation}
In Section \ref{sec_NHnull} we have seen that the axial Killing field $W$ is spacelike on the horizon. Recall that 
$0<g(W,W)= -\mathcal{Q}/\mathcal{G}+\mathcal{G}^2|\bm b|^2$ and by (\ref{eq_GK}),
$\mathcal{G}^2|\bm b|^2 = \ord(r^2)$ (with a suitable gauge choice of (\ref{eq_gauge_freedom}-\ref{eq_gauge_freedom_Qa})), which implies that 
\begin{equation}
	q_{-1}<0.\label{eq_horizon_q}
\end{equation}

Equation (\ref{eq_GK}) can be locally integrated to obtain
\begin{align}
	\bm c:=\mathcal{G}^3\bm b = (c_{-1}\cos\theta+c_0)\td\phi + \tilde{\bm c}\;,\nonumber\\
	\tilde{\bm c}:= \sum_{\substack{l\ge1\\|m|\le l}}\frac{c_{lm}r^l}{l}\star_2\td Y_{l}^m\;,
\end{align}
where $c_0$ is a constant of integration, $Y_l^m$ are the spherical harmonics, $\star_2$ is the Hodge star on the two-sphere, and the 
coefficients are defined by
\begin{align}
	c_{-1} := \gamma_0\kappa_{-1}-\kappa_0\gamma_{-1}\;, \qquad \sum_{\substack{l\ge1\\|m|\le l}}c_{lm}r^lY_{l}^m:=\gamma_{-1}\tilde{\mathcal{K}}-\kappa_{-1}\tilde{\mathcal{G}}\;.\label{eq_harmonic_nullexpansion}
\end{align}
Note that $\tilde{\bm c}$ is analytic in $r$ and smooth on the two-sphere. The metric (\ref{eq_nullmetric}) has a coordinate singularity 
at the horizon. The coordinate change that removes this singularity is of the form
\begin{equation}
	\td v' = \td v - \left(\frac{A_0}{r^2}+\frac{A_1}{r}\right)\td r\;, \qquad \td u' = \td u -\frac{B}{r}\td r\;,
\end{equation}
with $A_0, A_1, B$ constants. In these coordinates the metric becomes
\begin{align}
	g &= -\frac{2\td u \td v}{\mathcal{G}} - \frac{2B \td v \td r}{r\mathcal{G}} + 2\left(-\frac{A_0+B\mathcal{Q}r}{\mathcal{G}r^2} -\frac{A_1}{\mathcal{G}r}+ \frac{B|\bm c|^2}{\mathcal{G}^4r}\right) \td u \td r\nonumber\\
	&\qquad + \frac{2\bm c \td u}{\mathcal{G}}  + \left(\mathcal{G}^2-\frac{2A_0B+2A_1Br+B^2 \mathcal{Q}r}{\mathcal{G}r^3}+\frac{B^2|\bm c|^2}{\mathcal{G}^4r^2}\right)\td r^2 + 2\frac{B\bm c\td r}{\mathcal{G}r} \nonumber\\
	&\qquad\qquad +\left(-\frac{\mathcal{Q}}{\mathcal{G}}+\frac{|\bm c|^2}{\mathcal{G}^4}\right)\td u^2 + \mathcal{G}^2r^2(\td \theta^2 + \sin^2\theta\td\phi^2)\;,\label{eq_metric_horizon_null}
\end{align}
where the norm is with respect to the flat metric on $\mathbb{R}^3$. Setting the $1/r^2$ and $1/r$ terms in $g_{rr}$ and the $1/r$ terms in 
$g_{ur}$ to zero is equivalent to 
\begin{align}
	A_0 &= \mp \sqrt{-q_{-1}\gamma_{-1}^3}\;, \qquad A_1 = \pm \frac{\sqrt{\gamma_{-1}} (q_0\gamma_{-1}+3q_{-1}\gamma_0)}{2\sqrt{-q_{-1}}}\;, \nonumber\\
	&\qquad \qquad B = \mp \sqrt{-\frac{\gamma_{-1}^3}{q_{-1}}}\;. \label{eq_coord_params}
\end{align}
The near-horizon geometry ($(r, v)\to(\varepsilon r,v/\varepsilon)$, in the limit $\varepsilon\to0$) is given by 
\begin{equation}
	g_{NH}= -\frac{2r\td v \td u}{\gamma_{-1}} \pm2\sqrt{-\frac{\gamma_{-1}}{q_{-1}}}\td v \td r  -\frac{ q_{-1}}{\gamma_{-1}}\td u^2 + \gamma_{-1}^2(\td\theta^2 + \sin^2\theta\td\phi^2)\;,
\end{equation}
which corresponds to a black ring, in agreement with the near-horizon analysis in Section \ref{sec_NHnull}. 
Regularity of (\ref{eq_metric_horizon_null}) at the axes $\theta=0,\pi$ requires that $\bm c$ is a smooth 1-form on $S^2$. Setting $c_0=\mp1$ 
makes it smooth at $\theta=0, \pi$, respectively, however the required coordinate change between the two charts 
(covering the northern and southern hemisphere) by (\ref{eq_gauge_freedom}-\ref{eq_gauge_freedom_Qa}) is $v'' = v' - 2\phi$, which is well-defined only if $v$ is periodic, but that cannot happen. Hence, 
smoothness at the horizon is equivalent to 
\begin{align}
	c_0 &= 0\;,\label{eq_c0}\\
	c_{-1} &= \gamma_0\kappa_{-1}-\kappa_0\gamma_{-1}=0\;. \label{eq_cconstraint}
\end{align}
Note that the same constraint as (\ref{eq_cconstraint}) for $\mathcal{K}$ and $\mathcal{G}$ has been derived in the 
near-horizon analysis in (\ref{eq_c1NH}).

The regularity condition (\ref{eq_cconstraint}) has the following consequence.
\begin{lemma}
	$\mathcal{K}=0$ on $\mathbb{R}^3$. \label{lem_K0}
\end{lemma}
\begin{proof}
	By Theorem \ref{thm_necessary} $\mathcal{G}(\bm x)= \sum_i^N \gamma_i/|\bm x-\bm a_i| + \tilde{L}^{-1}$ and 
	$\mathcal{K}(\bm x)= \sum_i^N \kappa_i/|\bm x-\bm a_i|$, and all centres correspond to horizon components. In terms of the 
	parameters of the harmonic functions, (\ref{eq_cconstraint}) at each centre (horizon component) is equivalent to
	\begin{equation}
		\sum_{\substack{j=1}}^N\frac{\gamma_j\kappa_i-\gamma_i\kappa_j}{a_{ij}}+\tilde L^{-1}\kappa_i=0\qquad \text{ for each }i,\label{eq_kconstraint}
	\end{equation}
	where $a_{ij} := |\bm a_i-\bm a_j| + \delta_{ij}$\footnote{The second term is added so that the sum in (\ref{eq_kconstraint}) can be 
	taken over all indices, the numerator will cancel for the case when $i=j$.}. We look at this as a system of $N$ linear equations 
	for $N$ unknown $\kappa_i$ of the form $A_{ij}\kappa_j=0$, where the matrix $A$ is given by 
	\begin{equation}
		A_{ij} = \left(\tilde L^{-1}+\sum_{k=1}^N\frac{\gamma_k}{a_{ik}}\right)\delta_{ij} - \frac{\gamma_i}{a_{ij}}\;.
	\end{equation}
	Since $\tilde L^{-1}$ and $\gamma_i$ are positive, $A$ is a strictly row diagonally dominant matrix, that is 
	\begin{equation}
		|A_{ii}| = \tilde L^{-1}+\sum_{k\neq i}^N\frac{\gamma_k}{a_{ik}} > \sum_{k\neq i}^N\frac{\gamma_k}{a_{ik}} = \sum_{k\neq i}^N|A_{ik}|\qquad \text{for each }i.
	\end{equation}
	As a consequence, $A$ is invertible (see e.g. (5.6.17) of~\cite{horn_matrix_1985}), and thus $\kappa_i =0$ for all $i$, so 
	$\mathcal{K}\equiv 0$.
\end{proof}
\begin{corollary} \label{cor_a}
	$\bm b=0$, and thus $\mathcal{Q}=\mathcal{Q}_0$ is harmonic.
\end{corollary}
\begin{proof}
	As a consequence of Lemma \ref{lem_K0} and (\ref{eq_c0}), using (\ref{eq_GK}), $\bm b$ is pure gauge, so it can be set to zero 
	by a transformation of the form (\ref{eq_gauge_freedom}-\ref{eq_gauge_freedom_Qa}). Thus, the right-hand side of (\ref{eq_Fcalc}) vanishes, 
	hence the inhomogeneous part $\mathcal{Q}_I$ is harmonic, and by our boundary conditions in Theorem \ref{thm_necessary}, $\mathcal{Q}_I=0$.
\end{proof}
As a consequence, the Maxwell field is simply given by
\begin{equation}
	F = \frac{\sqrt{3}}{2}\star_3\td \mathcal{G}\;.\label{eq_Mxw_null_K0}
\end{equation}
In (\ref{eq_harmonic_exp_null}) $\gamma_0+\tilde{\mathcal{G}}$ is smooth at $r=0$, thus so is their contribution in (\ref{eq_Mxw_null_K0}). Also,
\begin{equation}
	\star_3\td \frac{\gamma_{-1}}{r} = -\gamma_{-1}\sin\theta\td\theta\wedge\td\phi\;,
\end{equation} 
which is smooth on $S^2$, so the Maxwell field is smooth at the horizon. We have thus seen that given (\ref{eq_horizon_G}-\ref{eq_horizon_q}), 
(\ref{eq_cconstraint}), the solution is smooth at the horizon.

In the DOC $\mathcal{G}>0$ and $\mathcal{Q}$ are smooth, hence all metric components are smooth. The metric is invertible, 
its inverse is given by
\begin{equation}
	g^{-1} = \mathcal{G}\mathcal{Q}\partial_v\otimes\partial_v -2\mathcal{G}\partial_u\odot\partial_v + \mathcal{G}^{-2}\partial_i\otimes\partial_i \;,\label{eq_inversemetric}
\end{equation}
which is explicitly smooth. The Maxwell field is trivially smooth on the DOC by (\ref{eq_Mxw_null_K0}). This concludes the regularity analysis.

\section{Classification theorem for supersymmetric Kaluza-Klein black holes} \label{ssec_classifictaion}

We are now ready to present our main result, the classification theorem for supersymmetric, axisymmetric, asymptotically Kaluza-Klein black 
holes (cf. Theorem 3 of~\cite{katona_supersymmetric_2023}). Recall that by Lemma \ref{lem_timelike_null} the supersymmetric Killing field is 
either generically timelike or globally null. We present the two cases in turn.
\begin{theorem}\label{thm_classification}
	An asymptotically Kaluza-Klein (in the sense of Definition \ref{def_KK}), supersymmetric black hole or soliton solution 
	$(\mathcal{M}, g, F)$ of $D=5$ minimal supergravity with an axial symmetry satisfying assumptions \ref{ass5D_susy}-\ref{ass5D_smooth} 
	of Section \ref{sec_recap}, with supersymmetric Killing field which is not globally null,
	must have a Gibbons-Hawking base (wherever $f \neq 0$), and is globally determined by four associated harmonic functions, which are of 
	`multi-centred' form, i.e.
	\begin{align}
	H =  h+\sum_{i=1}^N \frac{h_i}{r_i} \; , && K =  k+\sum_{i=1}^N \frac{k_i}{r_i} \; , && L = l+ \sum_{i=1}^N \frac{l_i}{r_i} \; , &&M =m+	\sum_{i=1}^N  \frac{m_i}{r_i}  \; , \label{eq_thm_harmonic}
	\end{align}
	where  $r_i := |\bm x-\bm a_i|$, $\bm a_i=(x_i, y_i, z_i) \in \mathbb{R}^3$, and the parameters are given by
	\begin{align}
		h = (\tilde L\gamma)^{-2}\;, \qquad k = 0\;, &\qquad l = 1\;, \qquad m = \tilde L\gamma v_H\;, \label{eq_harmoniccosntants_AT_thm}\\
		&\text{ or }\nonumber\\
		h = 0\;, \qquad k =\tilde{L}^{-1} \;, &\qquad l = 0\;, \qquad m = -\tilde{L}/2\;, \label{eq_harmoniccosntants_AN_thm} 
	\end{align}
	in the asymptotically timelike ($f\to 1$) and asymptotically null case ($f\to0$), respectively, with constants $\tilde L>0$, 
	$|v_H|<1$ and $\gamma = (1-v_H^2)^{-1/2}$. The centres either correspond to fixed points of the axial Killing field, or connected 
	components of the horizon. The 1-forms can be written as 
	\begin{align}
		\chi &= \sum_{i=1}^N \left(\chi_0^i+ \frac{h_i(z-z_i)}{r_i}   \right) \td \phi_i \; , \qquad \xi =-  \sum_{i=1}^N   \frac{k_i(z-z_i)}{r_i}  \td \phi_i \; , \label{chixi_sol}\\ 
		\hat\omega & = \sum_{\substack{i,j=1 \\ i \neq j}}^N \left(h_im_j+\frac{3}{2}k_il_j\right)\beta_{ij}\; ,\label{omega_sol_expl}
	\end{align}
	where $\chi_0^i$ are integers such that 
	\begin{equation}
		\chi_0^i +h_i\in 2\mathbb{Z} \;, \label{eq_chi_int}
	\end{equation} 
	and
	\begin{align}
	\td \phi_i &:= \frac{(x-x_i)\td y- (y-y_i)\td x}{(x-x_i)^2+ (y-y_i)^2} \; , \\
		\beta_{ij}&:= \left(\frac{(\bm x-\bm a_i)\cdot (\bm a_i-\bm a_j)}{|\bm a_i-\bm a_j|r_i}-\frac{(\bm x-\bm a_j)\cdot (\bm a_i-\bm a_j)}{|\bm a_i-\bm a_j|r_j}-\frac{(\bm x-\bm a_i)\cdot (\bm x-\bm a_j)}{r_ir_j}+1\right)\nonumber \\
		&\qquad\qquad\times\frac{((\bm a_i-\bm a_j)\times(\bm x-\bm a_j))\cdot \td\bm x}{|(\bm a_i-\bm a_j)\times(\bm x-\bm a_j)|^2} \; . \label{eq_beta_sol}
	\end{align}
	The parameters $h_i, k_i, l_i, m_i$  must satisfy for each centre $i=1, \dots, N$, 
	\begin{align}
		hm_i -mh_i+\frac{3}{2}(kl_i-lk_i) + \sum_{\substack{j=1 \\j\neq i}}^N\frac{h_j m_i -m_j h_i +\frac{3}{2}(k_jl_i-l_jk_i)}{|\bm a_i -\bm a_j|}=0 \; , \label{eq_thm_ctr}
	\end{align}
	and for the asymptotically null case 
	\begin{equation}
		\exists i \text{ such that } h_i \neq0\;.\label{eq_thm_TLcond}
	\end{equation}
	Moreover, if $\bm a_i$ corresponds to a fixed point of the axial Killing field,
	\begin{align}
		h_i = \pm 1 \; , && l_i +h_ik_i^2=0 \; , && m_i = \frac{1}{2}k_i^3 \; , \label{eq_thm_fix}
	\end{align}
	satisfying
	\begin{align}
		h_il - hk_i^2h_i +2kk_i+ \sum_{\substack{j=1 \\j\neq i}}^N\frac{2k_ik_j-h_i(h_jk_i^2-l_j)}{|\bm a_i-\bm a_j|}>0 \; ,  \label{eq_thm_fh}
	\end{align}
	whereas if $\bm a_i$ corresponds to a horizon component $h_i\in\mathbb{Z}$ and
	\begin{align}
		-h_i^2m_i^2-3h_ik_il_im_i+h_il_i^3-2k_i^3m_i+\frac{3}{4}k_i^2l_i^2>0 \; . \label{eq_thm_hor}
	\end{align}
	The horizon topology is $S^1\times S^2$ if $h_i = 0$, $S^3$ if $h_i = \pm 1$ and a lens space
	$L(|h_i|, 1)$ otherwise.	
	Finally, for all $\bm x \in \mathbb{R}^3\setminus\{\bm a_1, \dots, \bm a_N\}$, the harmonic functions must satisfy
	\begin{align}
		K^2+HL >0. \label{eq_thm_N}
	\end{align}
	The topology of the `sphere' at infinity ($t=\const$, $|\bm x|\to\infty$) is $S^3$ for $\tilde h_0=\pm1$, $S^2\times S^1$ for $\tilde h_0=0$, 
	or $L(|\tilde h_0|, 1)$ otherwise, where $\tilde h_0 = \sum_{i=1}^Nh_i$.
 \end{theorem}

 \begin{proof}
	The functional form of the harmonic functions (\ref{eq_thm_harmonic}-\ref{eq_harmoniccosntants_AN_thm}) is required by Theorem 
	\ref{thm_necessary} and (\ref{eq_harmoniccosntants_AT}-\ref{eq_harmoniccosntants_AN}). By Lemma \ref{lem_timelike_null} $f\neq 0$ 
	on a dense submanifold, and since zeros of $f$ and $H$ coincide by (\ref{eq_Ndef}) and assumption \ref{ass5D_timelike}, $H$ cannot 
	be identically zero, which implies (\ref{eq_thm_TLcond}). Smoothness requires 
	that $h_i=\pm 1$ at a fixed point (\ref{eq_fixed_h}), and $h_i\in \mathbb{Z}$ at a horizon component (\ref{eq_hor_h}), which also 
	determines the horizon topology. The 1-forms $\chi, \xi$ then obtained by simple integration of (\ref{eq_chieq}, \ref{eq_xi}), where we 
	introduced constants of integration $\chi^i_0$ such that $\chi^i_0+h_i\in2\mathbb{Z}$ for all $i$. The latter requirement follows 
	from the correct identification of the angles around a horizon component (\ref{eq_hor_chicond}) and fixed point (\ref{eq_fix_chicond}).
	For the integration of (\ref{eq_omegahat}), we follow~\cite{bena_black_2008,tomizawa_supersymmetric_2016,dunajski_einstein-maxwell_2007},  
	and introduce 1-forms $\beta_{ij}$ as a solution to 
	\begin{equation}
		\star_3\td\beta_{ij}=\frac{1}{r_i}\td\left(\frac{1}{r_j}\right)-\frac{1}{r_j}\td\left(\frac{1}{r_i}\right)+\frac{1}{r_{ij}}\td\left(\frac{1}{r_i}-\frac{1}{r_j}\right) \;, \label{eq_beta}
	\end{equation}
	with $r_{ij} = |\bm a_i-\bm a_j|$. One can easily check that $\beta_{ij}$ as given in (\ref{eq_beta_sol}) is a smooth 1-form away from the 
	centres on $\mathbb{R}^3$, in particular, it is free of string singularities. (\ref{eq_omegahat}) is then solved by 
	\begin{equation}
		\hat\omega = \sum_{\substack{i,j=1 \\ i\neq j}}^N \left(h_im_j+\frac{3}{2}k_il_j\right)\beta_{ij}   + \sum_{i=1}^N  \omega_{-1}^{i} \beta_i   \; ,
	\end{equation}
	where 
	\begin{equation}
		\omega_{-1}^i:=hm_i -mh_i+\frac{3}{2}(kl_i-lk_i) + \sum_{\substack{j=1 \\j\neq i}}^N\frac{h_j m_i -m_j h_i +\frac{3}{2}(k_jl_i-l_jk_i)}{|\bm a_i -\bm a_j|}\; .  \label{eq_omm1i}
	\end{equation} 
	Notice that for a local expansion around a centre as in (\ref{eq_harmonic_expansion}) the first two coefficients for the harmonic 
	functions are given by
	\begin{align}
		h_{-1}&= h_i \; , \qquad h_0 =h+\sum_{\substack{j=1 \\ j \neq i}}^N \frac{h_j}{r_{ij}} \; , \qquad  k_{-1}= k_i \; , \qquad k_0 = k+ \sum_{\substack{j=1 \\ j \neq i}}^N  \frac{k_j}{r_{ij}} \; , \\
		l_{-1} &= l_i \; , \qquad l_0 =l+ \sum_{\substack{j=1 \\ j \neq i}}^N  \frac{l_j}{r_{ij}} \; , \qquad   m_{-1}= m_i \; , \qquad m_0 =m+\sum_{\substack{j=1 \\ j \neq i}}^N  \frac{m_j}{r_{ij}}   \; .
	\end{align}
	Thus, we see that at each centre (\ref{eq_omegareq_fix}) and (\ref{eq_omegareq_hor}) is equivalent to $\omega_{-1}^i=0$, which yields (\ref{eq_thm_ctr}) 
	and (\ref{omega_sol_expl}). With the lack of string singularities in $\hat\omega$ all conditions of (\ref{eq_omegareq_fix}) and (\ref{eq_omegareq_hor}) are 
	satisfied. The remaining smoothness conditions at fixed points (\ref{eq_fixed_h}-\ref{eq_ompsicond1}) give (\ref{eq_thm_fix}-\ref{eq_thm_fh}) and at 
	a horizon component (\ref{ineq_horizon}) yields (\ref{eq_thm_hor}). (\ref{eq_thm_N}) is the necessary and sufficient condition for smoothness 
	of the solution at generic points (\ref{eq_condN}). Up until this point the proof is in essence identical to that of the asymptotically 
	flat classification in~\cite{katona_supersymmetric_2023}.

	\pagebreak[1]
	All that remains to be checked is that the solution is asymptotically Kaluza-Klein. In the asymptotically timelike case, using 
	(\ref{eq_omegapsi}) and (\ref{eq_f}) we have $f= 1 + \ord(|\bm x|^{-1})$, $\omega_\psi = \tilde L\gamma v_H+\ord(|\bm x|^{-1})$. 
	In the asymptotically null case it is simpler to evaluate the metric components directly by using (10) 
	of~\cite{breunholder_moduli_2019} to obtain
	\begin{align}
		g_{tt} = \ord(|\bm x|^{-1})\;, \; g_{t\psi} = -\tilde L+\ord(|\bm x|^{-1})\;, \; g_{\psi\psi} = \tilde L^2 + \ord(|\bm x|^{-1})\;, \; \frac{H}{f} = \tilde L^{-2}+\ord (|\bm x|^{-1})\;.
	\end{align}
	In both cases from (\ref{eq_beta}) and (\ref{omega_sol_expl}) $\hat\omega=\ord(|\bm x|^{-2})\td x^i$, and $\chi$ is given by 
	\begin{equation}
		\chi = \left(\tilde\chi_0+\frac{\tilde h_0 z}{r}\right)\td\phi + \ord(|\bm x|^{-2})\td x^i \;, 
	\end{equation}
	where 
	\begin{align}
		\tilde h_0 = \sum_{i=1}^N h_i \; ,&& \tilde \chi_0 = \sum_{i=1}^N \chi_0^i \;.
	\end{align}
	It is straightforward to check that the metric is asymptotically Kaluza-Klein with the coordinates defined in 
	(\ref{eq_tAKK}-\ref{eq_xAKK}) for the asymptotically timelike case, and (\ref{eq_xnull}-\ref{eq_psinull}) for the 
	asymptotically null case. From (\ref{eq_chi_int}) it follows that $\tilde h_0+ \tilde \chi_0\in 2\mathbb{Z}$, so the given solution is 
	compatible with our initial assumption made in Section \ref{ssec_asymptotics}, and the geometry of the `sphere' at infinity is indeed 
	that of $S^3$ (for $\tilde h_0=\pm1$), $S^2\times S^1$ (for $\tilde h_0=0$), or $L(|\tilde h_0|, 1)$. 
\end{proof}

\begin{theorem} \label{thm_classification_null}
	An asymptotically Kaluza-Klein (in the sense of Definition \ref{def_KK}), supersymmetric black hole or soliton solution 
	$(\mathcal{M}, g, F)$ of $D=5$ minimal supergravity with an axial symmetry satisfying assumptions \ref{ass5D_susy}-\ref{ass5D_smooth} of Section \ref{sec_recap}, for which 
	the supersymmetric Killing field $V$ is globally null, must be of the form 
	\begin{align}
		g &= -\frac{1}{\mathcal{G}}(\mathcal{Q}_0\td u^2 + 2\td u\td v) + \mathcal{G}^2\td x^i\td x^i\;,\label{thm_null_metric}\\
		F &= \frac{\sqrt{3}}{2}\star_3 \td \mathcal{G}\;,\label{thm_null_Maxwell}
	\end{align}
	with 
	\begin{align}
		\mathcal{G}(\bm x) = \tilde L^{-1} + \sum_{i=1}^N \frac{\gamma_i}{|\bm x-\bm a_i|}\;, \qquad \mathcal{Q}_0(\bm x) = -\tilde L + \sum_{i=1}^N \frac{q_i}{|\bm x-\bm a_i|} \label{thm_harmonic_null}
	\end{align}
	with constants
	\begin{align}
		\tilde L>0\;, \qquad \gamma_i>0\;, \qquad q_i < 0\;,\label{thm_Lgq}
	\end{align}
	and $\bm a_i\in \mathbb{R}^3$ correspond to connected components of the horizon 
	with topology $S^2\times S^1$. The topology at infinity is $S^2\times S^1$.
\end{theorem}
\begin{proof}
	Theorem \ref{thm_necessary} says that the solution is globally determined by three multi-centred harmonic functions through 
	(\ref{eq_nullmetric}) and (\ref{eq_GK}-\ref{eq_Fcalc}), where each centre corresponds to a connected component of the horizon.
	The constant terms in (\ref{thm_harmonic_null}) are fixed by asymptotic behaviour of the metric. Regularity of the horizon requires that 
	$\mathcal{K}=0$ by Lemma \ref{lem_K0}, and thus by Corollary \ref{cor_a} the metric and the Maxwell field are of the claimed form with 
	$\mathcal{Q}=\mathcal{Q}_0$. The constant $\tilde L$ is positive by Definition \ref{def_KK}, and (\ref{eq_horizon_G}-
	\ref{eq_horizon_q}) for each centre translate to (\ref{thm_Lgq}). In Section \ref{ssec_sufficient_null} we have seen that 
	this is sufficient for the solution to be smooth on and outside the horizon. One can also easily check that the metric asymptotes to 
	(\ref{eq_metric_AKK}) with coordinates defined by 
	\begin{align}
		u^0 = v\;, \qquad  \tilde \psi =u -\tilde L^{-1}v\;, \qquad u^i = \tilde L^{-1}x^i\;.
	\end{align}
	One can check that the geometry at infinity ($v=\const$, $r\to\infty$) is $S^2\times S^1$.
\end{proof}

\noindent{\bf Remarks.} 
\begin{enumerate}
	\item If one removes the condition (\ref{eq_thm_TLcond}) from Theorem \ref{thm_classification}, which violates the assumption that the solution 
	is in the timelike class, one exactly obtains the solutions in Theorem \ref{thm_classification_null} with 
	\begin{align}
		H = 0\;, \quad K = \mathcal{G}\;,\quad L=0\;, \quad M =\mathcal{Q}/2\;. \label{eq_nullTLidentification}
	\end{align}
	The proof of Theorem \ref{thm_classification} heavily relies on the timelike Gibbons-Hawking ansatz, so it is not obvious 
	{\it a priori} that one can relax (\ref{eq_thm_TLcond}). A possible explanation is that these solutions have a common six-dimensional origin, 
	from which one can obtain the timelike and null class by a Kaluza-Klein reduction along different directions~\cite{gutowski_all_2003}.\label{rem_TLnull}
	\item The null solutions in Theorem \ref{thm_classification_null} can also be obtained as a limit of certain asymptotically null 
	solutions in Theorem \ref{thm_classification}, where we define
	\begin{align}
		\tilde{\bm{x}} := \epsilon \bm x\;, \quad \tilde{\bm{a}}_i := \epsilon \bm a_i\;, \quad \gamma_i := \epsilon k_i\;, \quad q_i := 2\epsilon m_i\;,
	\end{align}
	and we take $\epsilon\to0$ while keeping $\tilde{\bm{x}}$, $\tilde{\bm{a}}_i$, $\gamma^i$ and $q^i$ fixed. Then $\tilde{\bm{x}}$ becomes 
	the cartesian coordinate of the null solutions and $\tilde{\bm{a}}_i$ are the positions of the centres. One can check that the parameter 
	constraints of Theorem \ref{thm_classification} are consistent with those of Theorem \ref{thm_classification_null}.
	\item In the null case the geometry of a spatial slice $\Sigma$ is given by 
	\begin{equation}
		g|_{v=\const} = -\mathcal{Q}_0\mathcal{G}^{-1}\td u^2 + \mathcal{G}^2\td x^i\td x^i\;,
	\end{equation}
	thus the DOC has the topology of a trivial circle fibration over $\mathbb{R}^3\setminus\cup_{i=1}^N{\bm a_i}$ (and hence the `sphere' at infinity has $S^2\times S^1$ geometry). This is a consequence of the lack of fixed points, and that all horizon components correspond to black rings. The latter is a necessary consequence of the near-horizon analysis, while the former follows from the assumption that the Killing fields have a timelike linear combination at each point of the DOC. It would be interesting to investigate whether solutions violating this assumption exist.
	\item The constants have the following physical meaning. $v_H$ is velocity of the horizon in the Kaluza-Klein direction with respect to 
	the asymptotic observer~\cite{elvang_supersymmetric_2005}, and $\gamma$ is the corresponding relativistic factor. This is apparent from 
	$V = \gamma(\partial_0 -v_H \tilde W)$, where $\tilde W$ is the unit vector in the Kaluza-Klein direction, and $V$ is tangent to the 
	generators of the horizon. $\tilde L$ sets the length of the Kaluza-Klein direction at infinity.
	\item Known constructions of supersymmetric Kaluza-Klein black holes of this theory \cite{elvang_supersymmetric_2005,gaiotto_5d_2006,bena_black_2005,ishihara_kaluza-klein_2006-1,ishihara_kaluza-klein_2006,nakagawa_charged_2008,matsuno_rotating_2008,tomizawa_charged_2008,tomizawa_squashed_2009,tomizawa_compactified_2010,tomizawa_exact_2011,tomizawa_general_2013,tomizawa_kaluza-klein_2018} use the timelike ansatz with a hyper-K\"ahler base, in particular (multi-)Taub-NUT space, hence they all belong to the asymptotically timelike case of Theorem \ref{thm_classification}. This can be seen from the asymptotic behaviour of the harmonic function $H$. For example in \cite{elvang_supersymmetric_2005} one centre corresponds to a `nut'-type fixed point (or a spherical black hole when they `hide' the nut singularity behind a horizon), while another one to a black ring. 
	\item As with flat asymptotics, in the timelike case it is not known whether (\ref{eq_thm_ctr}-\ref{eq_thm_N}) guarantees 
	that the DOC is globally hyperbolic. In fact, it is not clear what the sufficient conditions are for it to be stably causal ($g^{tt}<0$), 
	which is a consequence of global hyperbolicity. In \cite{avila_one_2018} it has been conjectured that positivity of $N^{-1}=K^2+HL$ (which is necessary for 
	smoothness at generic points) implies the lack of closed timelike curves for soliton solutions, which has been supported by numerical 
	evidence. In line with this conjecture, in \cite{breunholder_supersymmetric_2019,katona_supersymmetric_2023} numerical tests found no smooth 
	asymptotically flat black holes with positive ADM mass that violated stable causality. In contrast, in the null case by (\ref{eq_inversemetric})
	$g^{vv} = \mathcal{GQ}<0$, so the spacetime is stably causal automatically with no further constraints on the parameters.
	\item All solutions of Theorem \ref{thm_classification_null}, and those of Theorem \ref{thm_classification} for which $K\equiv0$ (which 
	includes all static solutions) have been argued to be exact string backgrounds~\cite{gutowski_all_2003}.
\end{enumerate}


\section{A classification of four-dimensional supersymmetric black holes from Kaluza-Klein reduction} \label{sec_reduction}

In this section we consider the dimensional reduction of the five-dimensional solutions classified in Theorem 
\ref{thm_classification}-\ref{thm_classification_null}, and determine the subclass for which the reduced solutions are smooth on and 
outside the horizon. We perform the Kaluza-Klein reduction along the direction of $W$ in coordinates adapted to it\footnote{
	This $\psi$ coordinate in general may be different from previous sections, but for certain gauge choice coincides with the 
	$\psi$ coordinate of the timelike case, or the $u$ coordinate of the null case.
}, so that 
$W=\partial_\psi$. For the dimensionally reduced theory, we will follow the 
field definitions of~\cite{elvang_supersymmetric_2005}, which are given by
\begin{align}
	&g =: e^{\Phi/\sqrt{3}}g^{(4)} + e^{-2\Phi/\sqrt{3}}(\td\psi +\mathcal{A})^2\;, &&A =: A^{(4)} + \rho \td \psi\;,\label{eq_def_metric}\\
	&F^{(4)} := \td A^{(4)} - \td\rho\wedge\mathcal{A}\;, && G^{(4)} := \td \mathcal{A}\;.\label{eq_def_Maxwell}
\end{align}
Here we used that since $\lie_WF=0$, we are free to work in a gauge in which $\lie_WA=0$. $g^{(4)}$ is the four-dimensional metric, $A^{(4)}$
and $\mathcal{A}$ are one-form potentials, and  $\Phi, \rho$ are scalar fields. 
The action (\ref{eq_action}) then reduces to
\begin{align}
	S &= \frac{1}{16\pi G_4}\int_{\mathcal{M}_4} \left(R^{(4)}\star 1 -\frac{1}{2}\star \td \Phi\wedge \td \Phi - 2 e^{2\Phi/\sqrt{3}}\star \td \rho\wedge \td \rho\right.\nonumber\\
	&\left.-\frac{1}{2}e^{-\sqrt{3}\Phi}G^{(4)}\wedge \star G^{(4)} - 2e^{-\Phi/\sqrt{3}}F^{(4)}\wedge\star F^{(4)} -\frac{8}{\sqrt{3}}\rho \td A^{(4)}\wedge\td A^{(4)} \right)\;,\label{eq_4action}
\end{align}
where $G_4 = G/4\pi$, and $R^{(4)}$ denotes the Ricci scalar of $g^{(4)}$.

It is important to establish which fields are physical, as we will require the smoothness of those only. Physical fields must be invariant
under five-dimensional coordinate changes of the form $\psi' = \psi + \mu(t, x^i)$ and gauge transformations which 
preserve the gauge condition $\mathcal{L}_W A =0$. Since $\iota_W\td A$ is invariant, and $0=\lie_W(A'-A)= \td (\iota_WA'-\iota_WA)$, 
the allowed gauge transformations must be of the form $A' = A + \td \lambda(t, x^i) + c \td\psi$ with some constant $c$. Under such  
transformations the fields transform as 
\begin{align}
	&A^{(4)}{}' = A^{(4)}+\td (\lambda-c\mu) -\rho\td \mu\;, \qquad \mathcal{A}' = \mathcal{A}-\td\mu\;,\qquad \rho' = \rho+c \;,\\
	g^{(4)}{}' &= g^{(4)}\;, \qquad F^{(4)}{}' = F^{(4)} \;, \qquad G^{(4)}{}'= G^{(4)}\;, \qquad \Phi' = \Phi\;, \qquad \td \rho' = \td \rho \;,\label{eq_physical}
\end{align}
hence the physical fields are those in (\ref{eq_physical}).
\\

\noindent{\bf Remarks.} 
\begin{enumerate}
	\item Even though the last term of (\ref{eq_4action}) contains gauge-dependent fields, the theory is gauge invariant since the five-dimensional 
	theory is. Indeed, one can check that the equations of motion of (\ref{eq_4action}) contain only the physical fields (\ref{eq_physical}).
	\item Ref. \cite{gaiotto_5d_2006} uses an alternative definition for the 2-form field $F^{(4)}$, that is
	\begin{equation}
		\tilde F^{(4)} : = F^{(4)}-\rho G^{(4)}= \td \left(A^{(4)} - \rho \mathcal{A}\right)\;,
	\end{equation}
	thus it is closed (as opposed to $F^{(4)}$). It is evident from its definition that smoothness of $\tilde F^{(4)}$ is equivalent to the 
	smoothness of $F^{(4)}$ (assuming the other fields in (\ref{eq_physical}) are smooth), since for solutions with a simply connected DOC 
	(which are the relevant ones here due to topological censorship), $\rho$ is globally defined by $\td\rho$ up to an additive constant, 
	so $\td\rho$ is smooth if and only if $\rho$ is smooth.
\end{enumerate}

Now we establish the subclass of solutions classified in Theorem \ref{thm_classification} and \ref{thm_classification_null} that reduce to a 
four-dimensional solution that is smooth on and outside the horizon. 
\begin{theorem}\label{thm_5Dto4D}
	A solution to five-dimensional minimal supergravity as in Theorem \ref{thm_classification} defines a four-dimensional, asymptotically flat 
	black hole solution of (\ref{eq_4action}) if and only if all centres correspond to horizon components (i.e. there are no fixed points 
	of the axial Killing field) and
	\begin{equation}
		D := \frac{3}{4}K^2L^2-2K^3M+HL^3-3HKLM-H^2M^2 >0 \label{eq_Dcond}
	\end{equation}
	for all $\bm x\in\mathbb{R}^3\setminus \{\bm a_1, \dots, \bm a_N\}$. Then the four-dimensional solution is given by 
	\begin{align}
		&g^{(4)} = - D^{-1/2}(\td t+ \hat\omega)^2 + D^{1/2}\td x^i\td x^i\;,\label{eq_4metric} \\
		&\Phi = -\frac{\sqrt{3}}{2}\log(DN^{2})\,, \qquad \rho = \frac{\sqrt{3}}{4}\frac{KL + 2HM}{K^2+HL}+c\;, \label{eq_4scalars}\\
		&G^{(4)} = \td \mathcal{A} = \td \left[\chi -\frac{2H^2M+3HKL+2K^3}{2D}(\td t+\hat\omega)\right]\;,\label{eq_4G} \\
		F^{(4)} = \frac{\sqrt{3}}{2}\td&\left[\frac{H}{K^2+HL}(\td t+\hat\omega)-\xi\right] + (\rho-c) \td \chi + \frac{2H^2M+3HKL+2K^3}{2D}\td\rho\wedge (\td t+\hat\omega)\label{eq_4F} \;,
	\end{align}
	where $ N^{-1} = K^2+HL$, $c$ is an arbitrary constant, $H,K,L,M$ are given by (\ref{eq_thm_harmonic}-\ref{eq_harmoniccosntants_AN_thm}), 
	and 1-forms $\chi, \hat\omega,\xi$ are given by (\ref{chixi_sol}-\ref{eq_beta_sol}).
\end{theorem}
\begin{theorem}\label{thm_5Dto4Dnull}
	A solution to five-dimensional minimal supergravity as in Theorem \ref{thm_classification_null} defines a four-dimensional, asymptotically flat 
	black hole solution of (\ref{eq_4action}), given by 
	\begin{align}
		g^{(4)} &= -\frac{\td v^2}{\sqrt{-\mathcal{QG}^3}} + \sqrt{-\mathcal{QG}^3}\td x^i\td x^i\;, \qquad \Phi = -\frac{\sqrt{3}}{2}\log\left(-\frac{\mathcal{Q}}{\mathcal{G}}\right)\;,\label{eq_4metric_null}\\
		\rho &=0\;, \qquad F^{(4)} = \frac{\sqrt{3}}{2}\star_3\td\mathcal{G}\;, \qquad G^{(4)} = \mathcal{Q}^{-2}\td v\wedge \td \mathcal{Q}\;,\label{eq_4F_null}
	\end{align}
	where $\mathcal{G}$ and $\mathcal{Q}$ are given by (\ref{thm_harmonic_null}) with (\ref{thm_Lgq}).
\end{theorem}

\begin{proof}[Proof of Theorem \ref{thm_5Dto4D} and \ref{thm_5Dto4Dnull}]
It is easy to see that the four-dimensional fields are related to five-dimensional smooth invariants by
\begin{align}
	g^{(4)} &= g(W,W)^{1/2} g - g(W,W)^{-1/2}W^\flat\otimes W^\flat\;, \qquad \Phi = -\frac{\sqrt{3}}{2}\log g(W,W)\;,\nonumber \\
	&\td\rho = -\iota_WF\;, \qquad F^{(4)} =F+\frac{(\iota_W F) \wedge W^\flat}{g(W,W)} \;, \qquad G^{(4)}= \td \left(\frac{W^\flat}{g(W,W)}\right)\;.\label{eq_5Dinvariants}
\end{align}
Hence, in Theorem \ref{thm_5Dto4D} the four-dimensional fields (\ref{eq_4metric}-\ref{eq_4F}) are smooth\footnote{
	Since they are invariant under $W$, i.e. $W\cdot X = 0=\mathcal{L}_WX$ for all $X$ in (\ref{eq_4metric}-\ref{eq_4G}), they descend to 
	$\mathcal{M}_4$ smoothly, which we can identify with the orbit space of the $U(1)$ action.
} if and only if $g(W,W)=N^2D>0$. Since $W\neq0$ on the horizon (Corollary 3 of \cite{katona_supersymmetric_2023}), this is satisfied if and 
only if \begin{enumerate*}[label=(\roman*)]\item there are no fixed points of $W$ in the DOC, i.e. all the centres of (\ref{eq_thm_harmonic}) correspond to horizon components, and 
\item $D>0$ away from the centres (\ref{eq_Dcond}).\end{enumerate*} A calculation using the explicit form of the solution together with 
equation (10) of \cite{breunholder_moduli_2019} yields the right-hand sides of (\ref{eq_4metric}-\ref{eq_4F}).

In the null case in Theorem \ref{thm_5Dto4Dnull} $g(W,W)>0$, so the four-dimensional fields are smooth, and 
(\ref{eq_4metric_null}-\ref{eq_4F_null}) comes from direct calculation using (\ref{thm_null_metric}-\ref{thm_null_Maxwell}) 
and (\ref{eq_def_metric}-\ref{eq_def_Maxwell}).

Finally, using the asymptotic explicit form of the harmonic functions, one can easily check that the four-metric approaches the Minkowski 
metric in coordinates\footnote{The factors of $\tilde L^{1/2}$ appear in (\ref{eq_AT_flat}-\ref{eq_AN_flat}) because we chose to reduce 
along a dimensionless coordinate $\psi$, so $g^{(4)}$ has length dimension $3$, and thus the asymptotic coordinates $\bar t$, $\bar x^i$ have $3/2$.} 
\begin{align}
	\bar t =  \tilde L^{1/2}\left(1-v_H^2\right)^{-1/2}t\;, \qquad&\qquad \bar x^i = \tilde L^{-1/2}\left(1-v_H^2\right)^{1/2}x^i\;,\label{eq_AT_flat}\\
	&\text{ and }\nonumber\\
	\bar t =  \tilde L^{1/2}v\;, \qquad&\qquad \bar x^i = \tilde L^{-1/2}x^i\;, \label{eq_AN_flat}
\end{align}
in the asymptotically timelike and null case\footnote{In the asymptotically but not globally null case we have $t$ instead of 
$v$ in (\ref{eq_AN_flat}).}, respectively.

\end{proof}

\noindent{\bf Remarks.} \begin{enumerate}
	\item As seen for the five-dimensional solutions (see Remark \ref{rem_TLnull} after Theorem \ref{thm_classification_null}), Theorem 
	\ref{thm_5Dto4D} can be extended, by omitting (\ref{eq_thm_TLcond}) from its assumptions, to include solutions of Theorem \ref{thm_5Dto4Dnull}. 
	The identification of harmonic functions is then given by (\ref{eq_nullTLidentification}). \label{rem_TLnull4D}
	\item For globally hyperbolic, hence stably causal, five-dimensional solutions of Theorem \ref{thm_5Dto4D}, $W$ must be spacelike in the DOC, so 
	(\ref{eq_Dcond}) must hold. In some neighbourhood of each horizon component (\ref{eq_thm_hor}) guarantees that (\ref{eq_Dcond}) 
	is satisfied, however there is no known sufficient condition for it to hold on the whole of the DOC. There is numerical evidence 
	that this does not restrict the moduli further than the smoothness conditions (see Remark after Theorem \ref{thm_classification}).
	For Theorem \ref{thm_5Dto4Dnull} there is no analogous requirement, as for all such five-dimensional solutions $g(W,W)>0$ in the DOC.
	\item The solutions of Theorem \ref{thm_5Dto4D} and \ref{thm_5Dto4Dnull} have been first described by Denef {\it et al.} 
	\cite{denef_supergravity_2000,denef_split_2001,bates_exact_2011}, with the explicit form of the solutions given in Section 4 
	of~\cite{bates_exact_2011}. Spherically symmetric solutions of the same form appear in~\cite{bertolini_n8_1999}. The connection to 
	five-dimensional solutions has been explored in detail in \cite{gaiotto_5d_2006,gaiotto_new_2006,elvang_supersymmetric_2005,behrndt_exploring_2006}. Here we extend this 
	connection by providing a classification of these solutions.
	\item The five dimensional Killing spinor locally defines a four dimensional Killing spinor of (\ref{eq_4action}) as shown in 
	Appendix \ref{app_KS}. We have not investigated the possible spin structures of $\mathcal{M}$ or $\mathcal{M}_{4}$ and their compatibility 
	(for more details see e.g.~\cite{figueroa-ofarrill_supersymmetric_2004}), hence the Killing spinor might be defined only up to a sign globally. 
	The proofs only use the Killing spinor bilinears which are invariant under such sign change.
\end{enumerate}

Now we establish the converse of Theorems \ref{thm_5Dto4D}-\ref{thm_5Dto4Dnull} to classify the asymptotically flat black hole solutions of (\ref{eq_4action}). For 
asymptotic flatness we use the following definition.\footnote{In this section, indices $abc\ldots$ denote four-dimensional spacetime indices 
in contrast to previous sections.}
\begin{definition}\label{def_AF}
	A four-dimensional spacetime is asymptotically flat if it has an end diffeomorphic to $\mathbb{R}\times(\mathbb{R}^3\setminus B^3)$, and on this 
	end the metric $g^{(4)}=-\td u^0\td u^0  + \delta_{ij}\td u^i\td u^j + \ord(R^{-\alpha})\td u^a\td u^b$ for some $\alpha>0$, where 
	$(u^0, u^i)$, $i=1,2,3$ are the pull-back of the cartesian coordinates on $\mathbb{R}\times\mathbb{R}^3$ and $R :=\sqrt{\delta_{ij} u^i u^j}$, 
	and the $k^{th}$ derivatives of the metric fall off as $\ord (R^{-\alpha-k})$ for $k=1,2$ in these coordinates.
\end{definition}

We assume that $(\mathcal{M}_4, g^{(4)}, \Phi,\rho,F^{(4)}, G^{(4)})$ is a solution of (\ref{eq_4action}) such that
\begin{enumerate}[label=(\roman*)]
	\item the solution admits a globally defined Killing spinor $\epsilon^{(4)}$, i.e. it is supersymmetric, \label{assumption_KS}
	\item the DOC, $\llangle\mathcal{M}_4\rrangle$ is globally hyperbolic,\label{assumption_GH}
	\item $\llangle\mathcal{M}_4\rrangle$ is asymptotically flat in the sense of Definition \ref{def_AF},
	\item the supersymmetric Killing field $V^{(4)}$ is complete, timelike on $\llangle\mathcal{M}_4\rrangle$, and in the asymptotic coordinates of 
	Definition \ref{def_AF} it is given by $V^{(4)}=\partial_0$, \label{assumption_Killing}
	\item the horizon $\mathcal{H}_4$ admits a smooth compact cross-section (which may not be connected),
	\item $\llangle\mathcal{M}_4\rrangle\cup \mathcal{H}_4$ admits a Cauchy surface $\Sigma_4$ that is a union of a compact set and an 
	asymptotically flat end,\label{assumption_cauchy}
	\item the metric and the fields are smooth ($C^\infty$) on and outside the horizon,
	\item there exists a gauge such that the $k^{th}$ derivatives of the fields $\Phi,\rho, \mathcal{A}, A^{(4)}$ in the 
	asymptotic chart fall off as $\ord(R^{-\beta-k})$ for $k=1,2$ and some $\beta>0$,\label{assumption_matter}
	\item $G^{(4)}$ is the curvature of a smooth connection $\eta$ on a principal $U(1)$-bundle over $\mathcal{M}_4$.\label{assumption_bundle}
\end{enumerate}

\begin{theorem}\label{thm_class_4D}
	Let $(\mathcal{M}_4, g^{(4)}, \Phi,\rho,F^{(4)}, G^{(4)})$ be a solution of (\ref{eq_4action}) satisfying assumptions 
	\ref{assumption_KS}-\ref{assumption_bundle}. Then it must belong to the class derived in Theorems \ref{thm_5Dto4D}-\ref{thm_5Dto4Dnull}. 
	In particular, it is globally determined by four harmonic functions of the form (\ref{eq_thm_harmonic}-\ref{eq_harmoniccosntants_AN_thm}) 
	satisfying (\ref{eq_thm_ctr}), (\ref{eq_thm_hor}-\ref{eq_thm_N}), (\ref{eq_Dcond}), and the solution is given by (\ref{eq_4metric}-\ref{eq_4F}) 
	together with (\ref{chixi_sol}-\ref{eq_beta_sol}).
\end{theorem}
\begin{proof}
	By assumption \ref{assumption_bundle} we can uplift the solution to five dimensions, identifying $\mathcal{M}$ with the total 
	space of the $U(1)$-bundle, on which we define the five-dimensional metric $g$ and Maxwell field $F$ as 
	\begin{align}
		g=e^{\Phi/\sqrt{3}}\pi^*g^{(4)} + e^{-2\Phi/\sqrt{3}}\eta^2\;, \qquad F = \pi^*F^{(4)}+\pi^*\td \rho \wedge \eta\,,\label{eq_g_5D_def}
	\end{align}
	where $\pi:\mathcal{M}\to\mathcal{M}_4$ is the bundle projection. The Killing spinor $\epsilon^{(4)}$ lifts to a five-dimensional 
	Killing spinor $\epsilon$, invariant under the $U(1)$-symmetry~\cite{figueroa-ofarrill_kaluza-klein_2022} (also see Appendix 
	\ref{app_KS}). Thus, $(\mathcal{M},g,F, \epsilon)$ is a supersymmetric solution of (\ref{eq_action}) 
	(as we have just undone the dimensional reduction). We will now show that it satisfies the assumptions of Theorem \ref{thm_classification} 
	or \ref{thm_classification_null}, hence $(\mathcal{M}_4,$ $g^{(4)},$ $\Phi,\rho,$ $F^{(4)},$ $G^{(4)})$ must belong to the class of 
	Theorem \ref{thm_5Dto4D} or \ref{thm_5Dto4Dnull}. Note that we include the latter solutions by not assuming (\ref{eq_thm_TLcond}) 
	for the harmonic function $H$ (see also Remark \ref{rem_TLnull4D} after Theorem \ref{thm_5Dto4Dnull}).

	We first prove that the five-dimensional DOC, $\llangle\mathcal{M}\rrangle$ is globally hyperbolic by showing that 
	$\Sigma = \pi^{-1}(\Sigma_4)$ is a Cauchy-surface. Let $p\in \llangle\mathcal{M}\rrangle$, $\gamma$ an inextendible 
	causal curve through $p$ in $(\mathcal{M}, g)$, and $U$ its tangent vector. From causality of $\gamma$ 
	\begin{equation}
		0\ge g(U,U) = e^{\Phi/\sqrt{3}}g^{(4)}(\pi_*U,\pi_*U) + e^{-2\Phi/\sqrt{3}}[\eta(U)]^2\ge e^{\Phi/\sqrt{3}}g^{(4)}(\pi_*U,\pi_*U)\;,
	\end{equation}
	thus $\pi_*U$ defines a causal curve $\gamma_4$ in $\mathcal{M}_4$. By assumption \ref{assumption_GH} $\gamma_4$ goes through $\Sigma_4$, 
	hence $\gamma$ goes through $\Sigma$. Acausality of $\Sigma$ follows from a similar argument, therefore it is a Cauchy surface, and 
	$\llangle\mathcal{M}\rrangle$ is globally hyperbolic.

	Next we show that $(\mathcal{M}, g)$ is asymptotically Kaluza-Klein according to Definition \ref{def_KK}. By compactness of 
	the fibres and assumption \ref{assumption_cauchy}, Definition \ref{def_KK}\ref{def_KK_top} is satisfied. The Dirac-currents of 
	$\epsilon$ and $\epsilon^{(4)}$ define the supersymmetric Killing fields $V$ and $V^{(4)}$ on $\mathcal{M}$ and $\mathcal{M}_4$,
	respectively. Let $W$ be the generator of the $U(1)$ action normalised such that its integral curves are $4\pi$-periodic\footnote{
		Here we use an uncanonical parametrisation of $U(1)$, where the parameter is $4\pi$-periodic.
	}.
	$[W,V]=0$ since $\epsilon$ is $U(1)$-invariant, and $\pi_*V=V^{(4)}$ (for details see Appendix \ref{app_KS}). 
	Let us adapt local coordinates to the vertical vector field so that $W=\partial_\psi$ and $\psi\sim\psi+4\pi$. It is obvious 
	that $W$ is a Killing field of $g$ (that also preserves $F$). In such a chart the connection is given by $\eta=\td\psi+\mathcal{A}$. 
	$V$ preserves $g$ and $W$, so it must also preserve $\eta = W^\flat/g(W,W)$. We may partially fix the gauge by requiring 
	$\lie_V\mathcal{A}=0$, then we have
	\begin{equation}
		0=\lie_V\td\psi = \td\iota_V\td\psi = \td V^\psi \;,\label{eq_Agauge}
	\end{equation}
	thus $V^\psi=c$ for some constant $c$. By assumption \ref{assumption_Killing} on the asymptotic end $\pi_*V =V^{(4)}= \partial_0$, hence 
	\begin{equation}
		V = \partial_0 + c\partial_\psi\;.\label{eq_V}
	\end{equation}
	It follows from assumption \ref{assumption_matter} that
	\begin{equation}
		\Phi = \Phi_0 + \ord(R^{-\beta}) \;, \qquad G^{(4)} = \td\mathcal{A} = \ord (R^{-\beta-1})\td u^a\wedge\td u^b\;. \label{eq_G_FO}
	\end{equation}
	Our gauge condition (\ref{eq_Agauge}) and (\ref{eq_V}) imply
	\begin{equation} 
		\td\iota_{\partial_0}\mathcal{A} = -\iota_{\partial_0}\td\mathcal{A} =  \ord (R^{-\beta-1})\td u^a\; \implies \mathcal{A}_0 = \tilde c + \ord(R^{-\beta})\;,
	\end{equation}
	with some constant $\tilde c$. Equation (\ref{eq_G_FO}) implies that the rest of $\mathcal{A}$ can be written in a gauge without changing the 
	form of the Killing fields such that
	\begin{equation}
		\mathcal{A} = \left(\tilde c+\ord(R^{-\beta})\right)\td u^0 + \ord(R^{-\beta})\td u^i\;.
	\end{equation}
	The leading order behaviour of the five-dimensional metric then becomes
	\begin{equation}
		g = e^{\Phi_0/\sqrt{3}}\eta_{ab}\td u^a\td u^b + e^{-2\Phi_0/\sqrt{3}}(\td\psi+\tilde c \td u^0)^2 + \ord(R^{-\tau})\td u^\mu \td u^\nu
	\end{equation}
	for some $\tau= \min\{\alpha, \beta\}$, and $\eta_{ab}$ denoting the $4D$ Minkowski metric. By defining $\psi' = \psi +\tilde c u^0$ and rescaling 
	the coordinates by constants, we get a metric of the form (\ref{eq_metric_AKK}), with $\Phi_0$ determining the asymptotic length scale $\tilde L$ 
	of the Kaluza-Klein direction. By assumption \ref{assumption_matter} and because there 
	is no dependence on $\psi$, the first two derivatives of the metric have the fall-off as in  Definition \ref{def_KK}\ref{def_KK_FO}, 
	and the components of the Riemann tensor fall off as $\ord(R^{-\tau-2})$. 
	We also see from (\ref{eq_V}) and from the final coordinate change that $\partial_0$ is a constant linear combination of $V$ and $W$.

	Finally, we need to check if the span of the supersymmetric and $U(1)$ Killing field is timelike. For this, by (\ref{eq_Ndef}) and 
	(\ref{eq_g_5D_def}) on $\llangle\mathcal{M}\rrangle$
	\begin{align}
		N = -\left[ e^{\Phi/\sqrt{3}}\pi^*g^{(4)}(V,V) + e^{-2\Phi/\sqrt{3}}[\eta(V)]^2\right]&e^{-2\Phi/\sqrt{3}} + \left[e^{-2\Phi/\sqrt{3}}\eta(V)\right]^2 = \nonumber\\
		=-e^{-\Phi/\sqrt{3}}g^{(4)}&(V^{(4)},V^{(4)})\;,
	\end{align}
	which is positive on $\llangle\mathcal{M}\rrangle$ by assumption \ref{assumption_Killing}. That is, at each point of 
	$\llangle\mathcal{M}\rrangle$ the determinant of the inner product matrix of Killing fields is negative, hence there 
	exists a timelike linear combination of $V$ and $W$.
\end{proof}

\noindent{\bf Remarks.}\begin{enumerate}
	\item Let us emphasise that we have not assumed any isometry apart from stationarity, which is guaranteed by supersymmetry (for a class of 
	$D=4$ supergravities see \cite{tod_more_1995}). Indeed, generically, solutions of Theorem \ref{thm_class_4D} only have a single Killing field.
	\item Assumption \ref{assumption_bundle} quantises the magnetic charges of the black holes associated with $G^{(4)}$, which in terms of the harmonic functions means that 
	$h_i\in \mathbb{Z}$ for $H$ in (\ref{eq_thm_harmonic}). Omitting this requirement leads to a more general class of black holes in four dimensions, 
	however those cannot be uplifted to get a smooth black hole solution in five dimensions.
	\item The requirement that the supersymmetric Killing field is timelike in assumption \ref{assumption_Killing} was also required for uniqueness 
	of four-dimensional supersymmetric, asymptotically flat black holes in minimal supergravity \cite{chrusciel_israel-wilson-perjes_2006}, 
	which shows that the general solution belongs to the Majumdar-Papapetrou (MP) class. Alternatively, in minimal supergravity, one can assume 
	the existence of a maximal hypersurface with a finite number of asymptotically flat or weakly cylindrical ends, and prove that the 
	Killing field is strictly timelike and static (see Theorem 1.2 of \cite{chrusciel_israel-wilson-perjes_2006}). 
	This reasoning would not work for the black holes considered in the current work, as these solutions are not static in general.
	\item The black holes of Theorem \ref{thm_5Dto4Dnull} are a generalisation of the aforementioned MP black holes. They are static solutions 
	that depend on two harmonic function on $\mathbb{R}^3$, and they are, in general, solutions to the full theory (\ref{eq_4action}). 
	If $\mathcal{Q}=-\mathcal{G}$, we obtain solutions to four-dimensional minimal supergravity\footnote{Consistent truncation of (\ref{eq_4action}) 
	to $4D$ Einstein-Maxwell is achieved by $\Phi=0$, $\td \rho=0$, $F^{(4)} = \frac{\sqrt{3}}{2}\star_4 G^{(4)}$~\cite{gauntlett_all_2003}.} 
	(Einstein-Maxwell), which describe magnetically charged MP black holes. In detail, 
	\begin{align}
		g^{(4)} = -\mathcal{G}^{-2}\td v^2 + \mathcal{G}^2\td x^i\td x^i\;,\qquad F^{(4)} = \frac{\sqrt{3}}{2}\star_3\td \mathcal{G}\;.
	\end{align}
	In contrast, the electrically charged and dyonic MP black holes (with quantised charges) belong to the asymptotically timelike class of 
	Theorem \ref{thm_5Dto4D} with harmonic functions $K=0$, $L = \gamma^2 H$, $M = \gamma^3 v_H H$, and parameter $\tilde L =1$. For the 
	detailed derivation see the end of Section 3.7 in~\cite{gauntlett_all_2003}, with the only difference that we used `gauge' freedom 
	(\ref{eq_harmonic_gauge}) to set $K\equiv0$ in agreement with our previous choice in Theorem \ref{thm_necessary} for the asymptotic values 
	of the harmonic functions. The solution is given by 
	\begin{align}
		g^{(4)} = -(\gamma H)^{-2}\td t^2 + (\gamma H)^2\td x^i\td x^i\;,\quad
		F^{(4)} = \frac{\sqrt{3}}{2}\left[(\gamma H)^{-2} \td t\wedge \td H+ \gamma v_H\star_3\td H\right]\;.
	\end{align}
	Thus, for $0<|v_H|<1$ we obtain dyonic MP black holes, while $v_H=0$ yields electrostatic MP black holes. \label{rem_MP}
	\item In Remark \ref{rem_MP} all static, supersymmetric solutions of minimal four-dimensional supergravity (with quantised charges) have been obtained. 
	It is an interesting task to determine all static solutions of the full theory (\ref{eq_4action}) (again, with quantised charges). This 
	includes all solutions in Theorem \ref{thm_5Dto4Dnull}, but also requires deriving the set of harmonic functions that yield $\td\hat\omega=0$ and 
	satisfy the constraints of Theorem \ref{thm_classification} and \ref{thm_5Dto4D}. Staticity in the timelike case implies that $f^{-1}$ 
	is harmonic~\cite{gauntlett_all_2003}, thus by (\ref{eq_f}) $K = cH$ for some constant $c$ (so we are in the asymptotically timelike case by 
	(\ref{eq_harmoniccosntants_AT_thm}-\ref{eq_harmoniccosntants_AN_thm})). By (\ref{eq_omegahat}) $\td\hat\omega=0$ requires $M+3cL/3 = kH$ 
	for some other constant $k$. Again, changing the `gauge' of (\ref{eq_harmonic_gauge}) such that $K\equiv 0$ yields that $M= c'H$ for a 
	constant $c'$, and $L$ is unconstrained. $H, L, c'$ must be such that (\ref{eq_thm_hor}-\ref{eq_thm_N}) and (\ref{eq_Dcond}) are satisfied. 
	Asymptotic constants (\ref{eq_harmoniccosntants_AT_thm}) fix $c' = \tilde L^3\gamma^3v_H$. Inequality (\ref{eq_thm_N}) yields (together 
	with (\ref{eq_harmoniccosntants_AT_thm})) that $H>0$ and $L>0$ on their domain, i.e. $h_i>0$ 
	and $l_i>0$. Then (\ref{eq_Dcond}) (which implies (\ref{eq_thm_hor})) is satisfied if and only if $l_i>|c'|^{2/3}h_i$ at each centre, and 
	does not impose any further constraint on the parameters $v_H$ and $\tilde L$. In summary, as in the null case (Theorem \ref{thm_5Dto4Dnull}), 
	the static solutions of the timelike case are determined by two harmonic functions, with parameters satisfying 
	$l_i> \tilde L \gamma |v_H|^{2/3}h_i >0$ and $h_i\in \mathbb{Z}$.
\end{enumerate}

\noindent{\bf Acknowledgements.} I would like to thank James Lucietti for suggesting this project, the many helpful conversations, and 
his comments on the manuscript. This work is supported by an EPSRC studentship.

\appendix
\section{Asymptotic fall-off of Riemann tensor} \label{app_asymptotics}

We obtain the following estimate on the curvature of the Levi-Civita connection for $g$ on the asymptotically Kaluza-Klein end.
\begin{claim*}
	The components of the Riemann tensor of $g$ fall off as $R^\mu{}_{\nu\lambda\kappa}=\ord(\tilde{r}^{-\tau-2})$ 
	in asymptotic coordinates of Definition \ref{def_KK}.
\end{claim*}
\begin{proof}
	We will use $ijk...$ indices for $\{u^1, u^2, u^3\}$, $abc...$ for $\{u^0, u^i\}$, and $\mu\nu...$  for all coordinates $\{u^a, \tilde\psi\}$.
	Since $\partial_0$ is Killing, all $u^0$ derivatives vanish. It is easy to check from Definition \ref{def_KK} that the 
	Christoffel-symbols have the following fall-off
	\begin{align}
		&\Gamma^c_{ab} = -\frac{1}{2}g^{c\tilde{\psi}}\partial_{\tilde{\psi}} g_{ab} + \ord(\tilde{r}^{-\tau-1})\;, 
		&&\Gamma^{\tilde{\psi}}_{ab} = -\frac{1}{2}g^{\tilde{\psi}\tilde{\psi}}\partial_{\tilde{\psi}} g_{ab} + \ord(\tilde{r}^{-\tau-1})\;, \nonumber\\
		&\Gamma^a_{\tilde{\psi} b} = \frac{1}{2}g^{ac}\partial_{\tilde{\psi}} g_{bc} + \ord(\tilde{r}^{-\tau})\partial_{\tilde{\psi}} g_{cb} + \ord(\tilde{r}^{-\tau-1})\;,
		&&\Gamma^{\tilde{\psi}}_{a\tilde{\psi}} =\ord(\tilde{r}^{-\tau})\partial_{\tilde{\psi}} g_{ca} + \ord(\tilde{r}^{-\tau-1}) \;, \\
		&\Gamma^a_{\tilde{\psi} \tilde{\psi}} = \ord(\tilde{r}^{-\tau}),
		&&\Gamma^{\tilde{\psi}}_{\tilde{\psi}\tilde{\psi}} = \ord(\tilde{r}^{-\tau}) \;.\nonumber
	\end{align}
	Now we use the assumed fall-off of the Ricci tensor in Definition \ref{def_KK} to get
	\begin{align}
		\ord(\tilde{r}^{-\tau-2}) = R_{ab} = R^c{}_{acb} + R^{\tilde{\psi}}{}_{a\tilde \psi b} = -\frac{1}{\tilde L^2}\partial_{\tilde{\psi}}^2g_{ab} + \ord(\tilde{r}^{-\tau})\partial_{\tilde{\psi}} g_{cd}  + \ord(\tilde{r}^{-\tau-1}) \;. \label{eq_hpsiij}
	\end{align}
	One can show using boundedness of the ${\tilde{\psi}}$ direction and the Mean Value Theorem that 
	$\partial_{\tilde{\psi}}^2g_{ab}=\ord(\tilde{r}^{-\alpha})\implies\partial_{\tilde{\psi}} g_{ab}=\ord(\tilde{r}^{-\alpha})$. 
	Using this iteratively with (\ref{eq_hpsiij}), one deduces that $\partial_{\tilde{\psi}} g_{ab}=\ord(\tilde{r}^{-\tau-1})$. Thus, 
	\begin{equation}
		\Gamma^\mu_{a\nu} = \ord(\tilde{r}^{-\tau-1})\;,\qquad \Gamma_{\tilde{\psi}\tilde{\psi}}^\mu=\ord(\tilde{r}^{-\tau}).
	\end{equation} 
	Taking a derivative $\partial_c$ of (\ref{eq_hpsiij}) and using our assumption about the fall-off of $\partial_cR_{ab}$, by the 
	same argument for $\partial_a\partial_{\tilde{\psi}} g_{bc}$, we obtain $\partial_a\Gamma^\mu_{b\nu} = \ord(\tilde{r}^{-\tau-2}).$

	It immediately follows that $R^a{}_{bcd}=\ord(\tilde{r}^{-\tau-2})$, $R^{\tilde{\psi}}{}_{abc}=\ord(\tilde{r}^{-\tau-2})$, and by the 
	fall-off of the  Ricci tensor so does $R^{\tilde{\psi}}{}_{a{\tilde{\psi}} b}=- R^c{}_{ac b}+\ord(\tilde{r}^{-\tau-2})=\ord(\tilde{r}^{-\tau-2})$. 
	Using symmetries of the Riemann tensor, we also have
	\begin{align}
		(\eta_{cd}+\ord(\tilde{r}^{-\tau}))R^d{}_{a\tilde{\psi} b} = -g_{c\tilde{\psi}}R^{\tilde{\psi}}{}_{a{\tilde{\psi}} b} + g_{{\tilde{\psi}}{\tilde{\psi}}}R^{\tilde{\psi}}{}_{bca}+g_{{\tilde{\psi}} d}R^d{}_{bca} = \ord(\tilde{r}^{-\tau-2}) \nonumber\\\implies R^c{}_{a{\tilde{\psi}} b} =\ord(\tilde{r}^{-\tau-2})\;.
	\end{align}
	Using the algebraic Bianchi identity, we obtain $R^c{}_{{\tilde{\psi}} ab}=-R^c{}_{b{\tilde{\psi}} a}+R^c{}_{a{\tilde{\psi}} b}= \ord(\tilde{r}^{-\tau-2})$, 
	and by the fall-off of the  Ricci tensor 
	$R^{\tilde{\psi}}{}_{{\tilde{\psi}}{\tilde{\psi}} a}= -R^d{}_{{\tilde{\psi}} d a}+\ord(\tilde{r}^{-\tau-2}) = \ord(\tilde{r}^{-\tau-2})$. Finally, we have
	\begin{align}
		(\eta_{cd}+\ord(\tilde{r}^{-\tau}))R^d{}_{{\tilde{\psi}}{\tilde{\psi}} b} = -g_{c\psi}R^{\tilde{\psi}}{}_{{\tilde{\psi}}{\tilde{\psi}} b} + g_{{\tilde{\psi}}{\tilde{\psi}}}R^{\tilde{\psi}}{}_{bc{\tilde{\psi}}}+g_{{\tilde{\psi}} d}R^d{}_{bc{\tilde{\psi}}} = \ord(\tilde{r}^{-\tau-2}) \nonumber\\\implies R^c{}_{{\tilde{\psi}}{\tilde{\psi}}b} =\ord(\tilde{r}^{-\tau-2})\;.
	\end{align}
\end{proof}

When $f\neq0$ and the base is well-defined as a Riemannian manifold, it is Ricci flat. The above argument works with just replacing $abc...$ 
indices with $ijk...$.

\section{Killing spinor equation in four and five dimensions} \label{app_KS}

In this section we reduce the five-dimensional Killing spinor equation to four dimensions, and derive the four-dimensional Killing spinor. 
We also prove that its Dirac-current is the four-dimensional projection of the five-dimensional supersymmetric Killing field.

The Killing spinor equation (KSE) of five-dimensional minimal supergravity is given by~\cite{gauntlett_all_2003}
\begin{equation}
	\left(\nabla_a+\frac{1}{4\sqrt{3}}F_{bc}\left(\gamma_a{}^{bc}+4\delta_a^b\gamma^c\right)\right)\epsilon=0\;,\label{eq_5DKSE}
\end{equation}
where $abc\dots$ are five-dimensional orthonormal frame indices, $\nabla_a$ is the spinor Levi-Civita connection, and the gamma matrices 
satisfy $\{\gamma_a,\gamma_b\}=-2g_{ab}$. Let us now choose a frame $\{e_a\} =\{e_A, e_5\}$, $A = 0, \dots, 3$  and co-frame such that 
\begin{align}
	&e^5 := e^{-\Phi/\sqrt{3}}\eta\;, \qquad e^A := e^{\Phi/2\sqrt{3}}E^A\;,\qquad \lie_We^a =0\;,
\end{align}
where $\eta$ is the $U(1)$-connection, and $\{E^A\}_{A=0}^3$ is a co-frame of the four dimensional metric $g^{(4)}$, and the $U(1)$ Killing 
field in this frame is $W = e^{-\Phi/\sqrt{3}}e_5$.

The lie derivative of the Killing spinor with respect to $W$ (which we assume to vanish) is 
defined as~\cite{fatibene_geometric_1996}
\begin{align}
	0=\lie_W\epsilon &:= \nabla_W \epsilon -\frac{1}{4}\td W^\flat \cdot \epsilon = W(\epsilon) -\frac{1}{8}\left[2\iota_W\omega_{ab} + (\td W^\flat)_{ab}\right]\gamma^a\gamma^b\epsilon\; , \label{eq_lie_spinor}
\end{align}
where $\omega_{ab}$ is the spin-connection. One can check that in this frame $(\td W^\flat)_{ab}=-2\iota_W\omega_{ab}$, hence 
(\ref{eq_lie_spinor}) simplifies to 
\begin{equation}
	\lie_W\epsilon = e^{-\Phi/\sqrt{3}}e_5(\epsilon)=0\;. \label{eq_WKS}
\end{equation}

We now derive the four-dimensional Killing spinor equation from the five-dimensional one. Let us write the following field components in the 
four-dimensional co-frame:
\begin{align}
	&\td\Phi =: \Phi_AE^A\;, \qquad \qquad\td\rho =: \rho_AE^A\;, \nonumber\\
	&G^{(4)}=: \frac{1}{2}G^{(4)}_{AB}E^A\wedge E^B\;, \qquad F^{(4)}=: \frac{1}{2}F^{(4)}_{AB}E^A\wedge E^B\;.
\end{align} 
In this frame (\ref{eq_5DKSE}) is given by 
\begin{align}
	KSE_A = &e^{-\Phi/2\sqrt{3}}\left(\nabla^{(4)}_A -\frac{1}{4\sqrt{3}}\Phi_B\gamma_A{}^B-\frac{e^{-\sqrt{3}\Phi/2}}{4}G^{(4)}_{AB}\gamma^B\gamma^5 +\frac{e^{-\Phi/2\sqrt{3}}}{4\sqrt{3}}F^{(4)}_{BC}\left(\gamma^{BC}\gamma_A+2\delta_A^B\gamma^C\right)\right.\nonumber\\
	&\qquad\left.+\frac{e^{\Phi/\sqrt{3}}}{2\sqrt{3}}\rho_B\left(\gamma_A{}^B\gamma^5+2\delta_A^B\gamma^5\right)\right)\epsilon=0\;,\label{eq_KSEA}\\
	KSE_5 = &\left(\frac{e^{-\Phi/2\sqrt{3}}}{2\sqrt{3}}\Phi_B\gamma^5\gamma^B+\frac{e^{-2\Phi/\sqrt{3}}}{8}G^{(4)}_{BC}\gamma^{BC}+\frac{e^{-\Phi/\sqrt{3}}}{4\sqrt{3}}F^{(4)}_{BC}\gamma_5\gamma^{BC}-\frac{e^{\Phi/2\sqrt{3}}}{\sqrt{3}}\rho_B\gamma^B\right)\epsilon=0\;,\label{eq_KSE5}
\end{align}
where $\nabla^{(4)}$ denotes the four-dimensional spinor Levi-Civita connection, and for (\ref{eq_KSE5}) we used that $0=e_5(\epsilon)$ by 
(\ref{eq_WKS}).

Following \cite{gauntlett_all_2003} we choose the four dimensional orientation $\eta^{(4)}$ such that $-e^5\wedge\eta^{(4)}$ is positively oriented. 
Then by multiplying (\ref{eq_KSE5}) by $\gamma_A\gamma^5$, substituting into (\ref{eq_KSEA}), and making use of the gamma matrix identity 
$\varepsilon_{abcde}\gamma^{de}=-2\gamma_{abc}$, after some gamma-matrix algebra we get that (\ref{eq_KSEA}-\ref{eq_KSE5}) is equivalent to
\begin{align}
	\left(\nabla^{(4)}_A-\frac{1}{4\sqrt{3}}\Phi_A+\left(\frac{e^{-\sqrt{3}\Phi/2}}{16}\star G^{(4)}_{BC}+\frac{\sqrt{3}e^{-\Phi/2\sqrt{3}}}{8} F^{(4)}_{BC}\right)\gamma^{BC}\gamma_A+\frac{\sqrt{3}e^{\Phi/\sqrt{3}}}{2}\rho_A\gamma^5\right)\epsilon=0\;, \label{eq_KSEAmod}\\
	\left(\Phi_A\gamma^A-\frac{\sqrt{3}e^{-\sqrt{3}\Phi/2}}{4}\star G^{(4)}_{AB}\gamma^{AB}+\frac{e^{-\Phi/2\sqrt{3}}}{2}F^{(4)}_{AB}\gamma^{AB}+2e^{\Phi/\sqrt{3}}\rho_A\gamma^5\gamma^A\right)\epsilon=0\;.\label{eq_KSE5mod}
\end{align}

In general, a Killing spinor parallel with respect to a supercovariant connection $D_A\epsilon^{(4)} =: \nabla_A\epsilon -\beta_A\epsilon^{(4)}=0$, 
defines a Killing vector by its Dirac-current $\overline\epsilon^{(4)}\gamma^A\epsilon^{(4)}$. The Killing equation for the Dirac-current holds 
only if the supercovariant connection satisfies
\begin{equation}
	\gamma_0\beta_{(A}^\dagger\gamma_0\gamma_{B)}+\gamma_{(B}\beta_{A)} =0\;.\label{eq_SCreq}
\end{equation}
This condition is not satisfied by the $\Phi_A$ term in (\ref{eq_KSEAmod}), therefore $\epsilon$ must be rescaled so that this term is cancelled. This 
is achieved by setting 
\begin{equation}
	\epsilon^{(4)}:= e^{-\Phi/4\sqrt{3}}\epsilon\;,
\end{equation}
so that (\ref{eq_KSEAmod}-\ref{eq_KSE5mod}) becomes
\begin{align}
	\left(\nabla^{(4)}_A+\left(\frac{e^{-\sqrt{3}\Phi/2}}{16}\star G^{(4)}_{BC}+\frac{\sqrt{3}e^{-\Phi/2\sqrt{3}}}{8} F^{(4)}_{BC}\right)\gamma^{BC}\gamma_A+\frac{\sqrt{3}e^{\Phi/\sqrt{3}}}{2}\rho_A\gamma^5\right)\epsilon^{(4)}=0\;, \label{eq_KSE4A}\\
	\left(\Phi_A\gamma^A-\frac{\sqrt{3}e^{-\sqrt{3}\Phi/2}}{4}\star G^{(4)}_{AB}\gamma^{AB}+\frac{e^{-\Phi/2\sqrt{3}}}{2}F^{(4)}_{AB}\gamma^{AB}+2e^{\Phi/\sqrt{3}}\rho_A\gamma^5\gamma^A\right)\epsilon^{(4)}=0\;, \label{eq_KSE45}
\end{align}
which correspond to the gravitino and dilatino transformations, respectively. Conversely, a solution of (\ref{eq_KSE4A}-\ref{eq_KSE45}) defines 
a $W$-invariant solution of (\ref{eq_5DKSE}). Also, the rescaling of the Killing spinor compensates for 
the rescaling of the tetrads, thus, we have that
\begin{equation}
	V^{(4)}=\overline\epsilon^{(4)}\gamma^A\epsilon^{(4)}E_A=\overline\epsilon\gamma^A\epsilon e^{-\Phi/2\sqrt{3}}E_A = \overline\epsilon\gamma^A\epsilon e_A = \pi_*\overline\epsilon\gamma^a\epsilon e_a=\pi_*V,
\end{equation}
where in the penultimate step we used that $\pi_*e_5 =0$ since $e_5$ is vertical.
The supercovariant connection of (\ref{eq_KSE4A}) satisfies (\ref{eq_SCreq}), so the four-dimensional Dirac current $V^{(4)}$ is Killing, 
and coincides with the projection of the five-dimensional Killing vector field, as claimed. 

\section*{Declarations}

\noindent{\bf Competing interests.} The author has no relevant financial or non-financial interests to disclose.
\\

\noindent{\bf Funding.} The author is supported by an EPSRC studentship.
\\

\noindent{\bf Data availability.} Data sharing is not applicable to this article as no datasets were generated or analysed during the current study.

\bibliographystyle{sn-mathphys}
\bibliography{ref}
\end{document}